\documentclass{article}
\usepackage{setspace}
\usepackage[auth-lg,affil-sl]{authblk} 
\usepackage{float}
\usepackage{amsmath}
\usepackage{geometry}
\usepackage[utf8]{inputenc}
\usepackage[nottoc]{tocbibind}
\usepackage{bm}
\usepackage{enumitem}
\usepackage{amssymb}
\usepackage[backref]{hyperref} %
\hypersetup{
    colorlinks=true,      
    linkcolor=blue,       
    citecolor=blue,       
    urlcolor=blue         
}
\usepackage{mathrsfs}
\usepackage{caption}
\usepackage{indentfirst}
\usepackage{array}
\usepackage{amsfonts}
\usepackage{amsthm}
\usepackage{listings}
\usepackage{graphicx}
\usepackage{centernot}
\usepackage{xcolor}
\usepackage{natbib}
\usepackage{booktabs}
\usepackage{threeparttable} 
\usepackage{longtable}
\usepackage{adjustbox}

\newtheorem{condition}{Condition}[section]
\newtheorem{proposition}{Proposition}[section]
\newtheorem{corollary}{Corollary}[section]
\newtheorem{theorem}{Theorem}[section]
\newtheorem{remark}{Remark}
\newtheorem{lemma}{Lemma}[section]

\date{\today}

\newcommand{\norm}[1]{\left\lVert#1\right\rVert}
\newcommand{\normm}[1]{\lVert#1\rVert}
\newcommand{\abs}[1]{\left|#1\right|}

\newcommand\ddfrac[2]{\frac{\displaystyle #1}{\displaystyle #2}}

\newcommand{\mc}[1]{\mathcal{#1}}
\newcommand{\E}{\mathbb{E}}
\newcommand{\R}{\mathbb{R}}

\DeclareMathOperator{\diag}{diag}
\DeclareMathOperator{\Var}{Var}
\DeclareMathOperator{\Cov}{Cov}
\DeclareMathOperator{\Corr}{Corr}

\DeclareMathOperator{\tr}{tr}

\title{Estimation of BLP models with high-dimensional controls\thanks{%
We thank seminar participants at UCL for their valuable feedback. 
We are also deeply grateful to Dennis Kristensen, Lars Nesheim and Ben Deaner for their careful guidance throughout the development of this paper.}}
\date{\today}

\author{Hua Jin\thanks{Department of Economics, University College London. Email: \href{mailto:uctphji@ucl.ac.uk}{uctphji@ucl.ac.uk}}}
\affil{University College London}

\begin{document}

\maketitle
 \normalsize
\doublespacing

\begin{abstract} 
    This paper develops a framework for estimating demand in differentiated product markets with high-dimensional product characteristics, building on the Berry, Levinsohn, and Pakes (1995) model with market-level data; the number of characteristics may exceed the number of markets.
    First, we propose a Neyman orthogonal estimator in which high-dimensional nuisance parameters are estimated by machine-learning methods such as LASSO, and show that it is $\sqrt{T}$-asymptotically normal for the parameters of interest, such as the price coefficient and price heterogeneity, even when the nuisance parameters converge at slower rates.
    Second, under approximate sparsity, we provide LASSO-based nuisance estimators that attain the required rates, although mean utility is recovered from market shares only when the nonlinear parameters are known.
    Monte Carlo simulations show good finite-sample performance. We illustrate the method with an application to canned tuna using Dominick's Finer Foods scanner data.
\end{abstract}

\thispagestyle{empty}
\clearpage


\setcounter{page}{1}

\section{Introduction}
The estimation of demand functions for differentiated products is a primary focus within Industrial Organization (IO) and has become increasingly influential across various fields of empirical microeconomics. 
A seminal work by \citet{berry1995automobile} established a powerful framework for this purpose, demonstrating its utility in applications such as analyzing automobile markets. 
The BLP model has since been widely adopted in empirical industrial organization, demonstrating considerable versatility across a range of applications. 
For example, \citet{nevo2000mergers,nevo2001measuring} applies the BLP framework to study demand in the ready-to-eat cereal industry. 
\citet{petrin2002quantifying} uses the model to quantify the welfare effects of the introduction of the minivan in the U.S. automobile market. 
\citet{berry1999voluntary} evaluate the impact of voluntary export restraints on the U.S. automobile industry, 
while \citet{goldberg2001evolution} analyse price dispersion in the European car market.

A significant limitation in the practical application of the standard BLP model is the lack of clear guidance on how to select product characteristics in the demand specification. 
Researchers typically choose them either by judgment, which risks omitting relevant characteristics, or by data-driven procedures whose implications for subsequent inference have received limited attention in the BLP literature.
Some studies select characteristics on an ad hoc basis \citep{chidmi2007brand}; others reduce dimensionality before estimation, for example by principal components \citep{backus2021common, compiani2026demand}. A third approach uses data-driven selection: \citet{gillen2019blp} adapt the post-double-selection procedure of \citet{belloni2014inference} to BLP setting. What these strategies share is that inference on the demand parameters is conducted as if the selected or constructed specification were known. For the linear or logit model, post-selection inference is justified by existing theory \citep{belloni2014inference, chernozhukov2015post}; for the BLP model with random coefficients it is not, and \citet[p. 277]{gillen2019blp} state explicitly that they do not provide an asymptotic theory justifying their BLP post-selection procedure.

This absence of guidance creates a fundamental trade-off. 
On the one hand, including many product characteristics risks triggering the curse of dimensionality, which can lead to imprecise estimates or even estimation failure. 
On the other hand, restricting the model to a small set of variables—whether by ad hoc choice or aggressive dimensionality reduction—risks omitting important confounders, potentially introducing bias.
The problem is further exacerbated in modern high-dimensional settings, where the number of potential covariates can exceed the number of observations. 
In such cases, standard GMM estimation methods become infeasible, and ad hoc dimension reduction becomes increasingly unreliable. 
For example, \citet{compiani2026demand} extract 512 image features for each product and keep only a few principal components of these features to make a random-coefficient model estimable, at the risk of discarding information relevant for demand.
Similarly, in \citet{gillen2019blp}, over 200 demographic variables enter utility through party-specific coefficients, so the number of characteristics exceeds the 300 electoral districts that serve as markets; the authors therefore rely on LASSO to select the relevant variables, which in turn requires inference that remains valid after selection.

We present an approach to estimation and inference for the BLP model where the dimension of product characteristics is large, potentially much larger than the number of markets.
In the BLP model, the consumer utility from choosing product $j$ in market $t$ is given by
\begin{align}
    u_{ijt} &= f_{u0}(x_{jt}) + p_{jt}\alpha_i+ \xi_{jt} + \epsilon_{ijt} \label{eq:standardBLP_utility}
\end{align}
For notational clarity we write the model with a scalar random coefficient on price, but
this is an expositional choice rather than a restriction of the method: all results hold
whenever the nonlinear parameter $\sigma$ is finite-dimensional with $d_\sigma$ fixed.
This covers the utility of the form $u_{ijt}=f_{u0}(x_{jt}) + x_{jt}^{K\prime}\nu_i +p_{jt}\alpha_i+\xi_{jt}+\epsilon_{ijt}$, where $x_{jt}^{K}$ is a fixed subvector of $x_{jt}\in\R^{d_x}$ with dimension $K$, and the random vector $\nu_i$ follows a distribution indexed by a low-dimensional $\sigma_0$.
The utility equals the value a consumer gets from a product's characteristics, $x_{jt}$, $p_{jt}$, and $\xi_{jt}$, and an individual-specific taste shock, $\epsilon_{ijt}$. The price $p_{jt}$ is endogenous, and identification thus relies on instruments $z_{jt}\in\R^{d_z}$.
A significant computational and statistical challenge arises when the dimension of the product characteristics vector, $d_x$, is large relative to the number of market observations, $T$.
This challenge is further exacerbated when the function $f_{u0}(x_{jt})$, in contrast to the standard BLP model, can be nonlinear. 
In such high-dimensional settings, traditional estimation methods, such as the GMM estimator of \citet{berry1995automobile}, become impractical or infeasible, necessitating additional structure to enable informative inference.

To resolve this, we construct a Neyman orthogonal estimator.
The procedure follows a two-step framework: first, we estimate the high-dimensional nuisance parameters using flexible machine learning methods, such as LASSO; second, we plug these estimates into a Neyman orthogonal moment function and estimate the parameters of interest. The resulting estimator is $\sqrt{T}$-asymptotically normal.
We then specialize to approximate sparsity, under which each nuisance function, e.g. $f_{u0}(x_{jt})$ admits a sparse linear approximation $x_{jt}'\beta_0$, with $\normm{\beta_0}_0\ll T$, plus a negligible approximation error.
This assumption reflects the realistic empirical scenario where researchers have access to a vast set of potential product characteristics but lack prior knowledge of the small subset that explains mean utility and market shares. 
Crucially, approximate sparsity ensures that the nuisance parameters can be estimated at a sufficiently fast rate, which is a prerequisite for establishing the asymptotic properties of the main parameters of interest.

The first contribution is the derivation of an orthogonal moment condition for the BLP model that is economically interpretable and, when the nuisance functions are linear, equivalent to the efficient orthogonal GMM score. Starting from the moment condition $\E[(z_{jt}',x_{jt}')'\xi_{jt}]=0$, our moment uses the characteristics to residualize the instruments, $\E[(z_{jt}-\E[z_{jt}\mid x_{jt}])\xi_{jt}]=0$. The moment condition extends the orthogonal score of \citet{chernozhukov2018double} for the partially linear IV model, to which the BLP model reduces only when $\sigma_0$ is known; it can also be obtained as the locally robust moment of \citet{chernozhukov2022locally} when $f_{u0}$ is treated as a first-step conditional mean. It uses as instruments the part of $z_{jt}$ not explained by the characteristics, and its only nuisance parameters are two functions of characteristics: the conditional mean of the instruments given the characteristics and $f_{u0}$. Under approximate sparsity, the condition needed for the former is that the instruments depend mainly on a few characteristics, a restriction with a direct economic interpretation. By contrast, the generic efficient orthogonal score for GMM \citep{belloni2018high}, applied by \citet{sawada2020estimating} to the BLP model with the moment condition $\E[(z_{jt}',x_{jt}')'\xi_{jt}]=0$ under a linear $f_{u0}$, requires estimating high-dimensional matrices derived from the gradient and the variance of the moments. These matrices are assumed to be row-wise sparse and bounded, a condition without a direct economic interpretation. We show that, when the nuisance functions are linear in a possibly dictionary-expanded set of characteristics, the efficient score equals ours up to a fixed-dimensional matrix, and ours remains orthogonal when they are nonlinear. 

The second contribution is an estimator of the nuisance functions that attains desired rates under approximate sparsity. Existing results on double/debiased estimation take such a first step as given \citep{chernozhukov2018double, chernozhukov2022locally}, but this step faces a complication specific to the BLP model. The function $f_{u0}$ is a component of mean utility, which can be recovered from market shares only at the true $\sigma_0$, while estimating $\sigma_0$ requires an estimate of $f_{u0}$. Hence $f_{u0}$ cannot be estimated by an ordinary first stage: any feasible estimator, whether it plugs in a preliminary estimate of $\sigma_0$ or profiles over $\sigma$, inherits the estimation error in $\sigma$, which breaks the usual separation between the two steps of a double/debiased procedure. We construct an $\ell_1$-penalized $\ell_\infty$-minimum-distance estimator of $\sigma_0$, characterize how its error propagates through the share inversion into a second-stage LASSO, and show that the resulting estimator of $f_{u0}$ attains the required rate.  We then establish $\sqrt T$-asymptotic normality of the cross-fitted orthogonal estimator of the parameters of interest.

In simulations, our estimator attains rejection rates close to the nominal
level, whereas plug-in estimators over-reject severely.
We apply the method to the canned tuna market from Dominick's Finer Foods scanner data, where
it delivers plausible own-price elasticities, whereas standard GMM produces heavy-tailed elasticities and predicts held-out market shares poorly.

\textbf{Literature Review}. 
First, this paper contributes to the literature on demand estimation, in particular to work on nonparametric and high-dimensional extensions of the BLP model.
\citet{berry2014identification} establish nonparametric identification in differentiated products markets. 
Building on this framework, \citet{dunker2023nonparametric} extend these results by providing conditions for the identification of the densities of the random coefficients.
A significant strand of this literature treats the random coefficient density as an unknown function within an infinite-dimensional space.
\citet{lu2023semi} and \citet{wang2023sieve} approximate the random-coefficient density by sieves, the former under large-$J$ asymptotics.
\citet{compiani2022market} instead treats the inverse demand function itself as unknown and approximates it by sieves.
Moving further away from the BLP structure, other work takes a fully nonparametric approach to the choice function itself. For example, \citet{singh2023choice} exploit permutation invariance to approximate choice functions with neural networks, accommodate endogenous product features, and conduct inference on functionals such as price elasticities.
These papers treat the random-coefficient distribution, the inverse demand function, or the choice function itself as an infinite-dimensional unknown, while taking the set of product characteristics as fixed and low-dimensional.
Instead, we assume that the random-coefficient distribution is parametric and indexed by a finite-dimensional $\sigma$, while mean utility contains an unknown function of a high-dimensional vector of characteristics, whose dimension may exceed the number of markets.

Closer to our paper, another strand retains the parametric BLP structure but allows the set of product characteristics to be high-dimensional \citep{gillen2014demand, gillen2019blp, sawada2020estimating}. 
Early applications of machine learning in this context, such as the triple-LASSO estimator of \citet{gillen2014demand} and the double-LASSO control selection of \citet{gillen2019blp}, address high-dimensionality but are not accompanied by inferential theory. 
By contrast, \citet{sawada2020estimating} establish asymptotic linearity of their estimator by applying the regularized GMM approach of \citet{belloni2018high} to the BLP model, assuming that mean utility is linear and exactly sparse; high-dimensional attributes may also affect the variance of mean utility through finitely many indices.
As in \citet{belloni2018high}, debiasing requires estimating high-dimensional matrices that are assumed to be row-wise sparse and bounded.
Our approach relaxes both assumptions: mean utility may contain an unknown, approximately sparse function of the characteristics, and, as discussed above, our orthogonal score involves no such matrices.
On the other hand, we keep the nonlinear parameters low-dimensional.

Beyond the BLP framework, our paper relates to a growing literature that brings machine learning methods, including double/debiased machine learning, into demand estimation more broadly.
\citet{chernozhukov2019demand} use debiased LASSO methods to address the problem of ``many prices'' in demand analysis, and \citet{semenova2023inference} develop orthogonal LASSO for heterogeneous price elasticities in high-dimensional dynamic panels. 
\citet{adam2024machine} use deep learning to correct the selection bias caused by zero or low sales in long-tail markets. 
A related strand constructs product characteristics from unstructured or crowd-sourced data, such as text, images, and similarity judgments, and incorporates the resulting low-dimensional representations into otherwise standard demand models \citep{magnolfi2025triplet, compiani2026demand, bach2026adventures}. 
These representations are complementary to our approach: they can serve as the dictionary $x_{jt}$, and our procedure delivers valid inference without committing ex ante to a small subset of them. 
Debiased machine learning has also begun to enter structural empirical work: \citet{kim2025airline} applies it to a nested logit demand model for the U.S.\ airline industry.
Without random coefficients, such models can be written as a partially linear IV model after inverting market shares, so standard debiased machine learning applies directly; random coefficients break this structure, which is the difficulty our paper addresses.

Second, this paper connects to the literature on high-dimensional inference.
\citet{chernozhukov2018double} and \citet{belloni2018high} provide a general
recipe for constructing Neyman orthogonal scores in GMM problems with
high-dimensional nuisance parameters; \citet{chernozhukov2022locally} obtain locally robust
moments of fixed dimension by adding first step influence functions, and \citet{chernozhukov2022automatic} give automatic debiasing through the Riesz representer.
The obstacle in the BLP model is not orthogonality per se but feasibility.
Treating the characteristic coefficients as a high-dimensional parametric nuisance, the generic GMM construction projects out the nuisance directions in moment space, which involves $O(d_x)\times O(d_x)$ matrices that can be regularized only under sparsity restrictions on the projection itself \citep{belloni2018high, chernozhukov2018double}.
Treating instead $f_{u0}$ as a first-step function, the locally robust construction of \citet{chernozhukov2022locally} applies to general GMM problems, but it requires the first-step influence function and takes as given a first-step estimator attaining suitable rates, while the automatic approach of \citet{chernozhukov2022automatic} is developed for parameters of the form $\theta_0=E[m(W,\gamma_0)]$ with $\gamma_0$ a regression of an observed outcome.
We show that in the BLP model the correction term involves only the conditional mean of the instruments given the characteristics, so that orthogonalization reduces to residualizing the instruments, and we construct the nuisance estimator that attains the required rates even though the mean utility depends on $\sigma_0$ through the share inversion.

Orthogonal scores have also been developed for nonlinear semiparametric
models, including the logistic partially linear model \citep{liu2021double}
and single-index conditional moment restrictions
\citep{nekipelov2022regularised}, and for estimating equations whose nuisance
functions depend on the target parameter \citep{kallus2024localized}. Our
problem shares the last feature: $f_{u0}$ is defined independently of
$\sigma_0$, but it can only be recovered through the share inversion evaluated
at a preliminary estimate of $\sigma_0$. 
Unlike \citet{kallus2024localized}, where the preliminary estimator is an input satisfying a high-level rate
condition, obtaining even a rough estimate of $\sigma_0$ is itself nontrivial
here, because $\sigma_0$ is estimated jointly with the $d_x$-dimensional coefficient on the characteristics. 

Third, our preliminary estimator of $\sigma_0$ relates to the
literature on regularized estimation with endogenous regressors and many
moment conditions \citep{caner2009lasso, belloni2012sparse,
fan2014endogeneity}, and in particular to Dantzig-type estimators (or regularized minimum distance estimators), which
minimize the $\ell_1$ norm of the coefficients subject to a sup-norm
constraint on the moments \citep{candes2007dantzig, belloni2018high}. These
methods regularize all the coefficients. 
In the BLP model, by contrast, only the coefficients on the characteristics are high-dimensional; the random-coefficient parameter $\sigma$ and the price coefficient are low-dimensional, and penalizing them would only shrink the parameters that this step must estimate accurately.
We therefore penalize only the coefficients on the characteristics.
For each value of $\sigma$, the estimation problem is a linear program and the
nonconvexity is confined to a low-dimensional search.

\textbf{Organization of the Paper}. The remainder of this paper is structured as follows. 
Section 2 discusses the theoretical underpinnings of the standard BLP model and its estimation methodology. 
Section 3 presents a high-dimensional extension of the BLP model along with the associated Neyman orthogonalization approach. 
Theoretical results regarding the consistency and asymptotic normality are provided.
Section 4 applies the general theory to a specific BLP model under approximate sparsity, detailing the estimation procedure for the high-dimensional nuisance parameters. 
In Section 5, we conduct a simulation study to evaluate the performance of the estimator. 
Section 6 provides an empirical illustration, applying our methodology to analyse consumer demand in the tuna market using Dominick's Finer Foods data. 
Finally, Section 7 concludes.
Appendices provide detailed proofs for the theoretical results established in Section 3 and Section 4.

\textbf{Notation.}
We use the following empirical process notations. 
Cardinality of a set $S$ is denoted by $|S|$.
Given a vector $\beta$, and a set of indices $S$, $\beta_S$ is the vector that has the same value as $\beta$ on the indices in $S$ and zero elsewhere,
i.e., $\beta_{S,j}=\beta_j$ if $j\in S$ and $\beta_{S,j}=0$ if $j\notin S$.
$\mathbb{E}_J:=\frac{1}{J}\sum_{j=1}^J$. 
$\mathbb{E}_T:=\frac{1}{T}\sum_{t=1}^T$,
$\mathbb{E}_{JT}:=\frac{1}{JT}\sum_{j=1}^J\sum_{t=1}^T$, and
$\mathbb{E}_{T,L}:=\frac{1}{T}\sum_{l=1}^L\sum_{t\in I_l}$, where $I_l$ is the index subset of $\{1,\dots,T\}$.\footnote{This notation is used to describe the cross-fitting, when the observations are divided into $L$ folds, and the $l$-th fold is denoted as $I_l$.
}
The $l_p$ norm of a vector $\norm{a}_p=(\sum_i |a_i|^p)^{1/p}$.
Specifically, $\norm{a}_\infty = \max_{i} |a_i|$, where $a_i$ is the $i$-th element of the vector $a$. 
The $l_p$ norm of a matrix $\normm{X}_p = (\sum_{i,j} |X_{ij}|^p)^{1/p}$.
Specifically, $\norm{X}_\infty = \max_{i,j} |X_{ij}|$, where $X_{ij}$ is the $(i,j)$-th element of the matrix $X$. 
The operator norm of a matrix $X$ is $\normm{X}_{op}=\sup_{\norm{a}_2=1}\normm{Xa}_2$.
$a_t\lesssim b_t$ means there exists a constant $C$ such that $a_t\leq C b_t$ for all $t$.
$a_t\lesssim b_t\quad w.p.a.1$ means there exists a constant $C$ such that $P(a_t\leq C b_t)\to 1$ as $T\to\infty$, which is also denoted as $a_t\lesssim_P b_t$.
$\partial_\sigma \varphi:=\partial \varphi/\partial \sigma$.
$C^2(B):=\{\varphi:B\to\R \mid \varphi \text{ is twice continuously differentiable}\}$.
Finally, for a function $\varphi(w,\hat{\eta}(w))$, where $w$ is a vector and $\hat{\eta}$ is an estimated function using observations independent of $w$, define
$\normm{\varphi(w,\hat{\eta}(w))}_{L^2}=\left(\int \normm{\varphi(w,\hat{\eta}(w))}_2^2 F(dw)\right)^{1/2}$.
For $v \in \R^n$, $\normm{v}_n := \big(\frac1n\sum_{i=1}^n v_i^2\big)^{1/2}$.

\section{Standard BLP Model}
\subsection{Framework}
Berry, Levinsohn, and Pakes (1995) begin with the following problem. 
Consider a market with $J$ competing products and an outside good, denoted as $0$. 
The vector of product characteristics of product $j$ in market $t$ will be denoted by $(\xi_{jt},x_{jt}',p_{jt})$. 
$\xi_{jt}\in\R$ represents the product characteristic which is not observed by the econometrician whereas $x_{jt}\in\R^{d_x}$ are observed. 
$p_{jt}\in\R$ is assumed to be an endogenous variable, e.g., price. 
Assume that the conditional mean of the unobserved characteristics is zero, i.e., 
\begin{equation}\label{eq:MeanBound_xi}
    \E[\xi_{jt}|x_{jt}]=0\qquad \text{and} \qquad \sup_{1\leq j\leq J}\E[\xi_{jt}^2|x_{jt}]<C<\infty
\end{equation}

In addition to those exogenous variables $x_{jt}$, we allow the existence of endogenous variables $p_{jt}$ in the sense of being related to $\xi_{jt}$. 
This thus requires the use of instruments for identification.
An individual $i$ in market $t$ chooses a product from a set of $J$ products plus an outside option, with fixed $J$ and growing $T$.
The utility of individual $i$ from choosing product $j$ in market $t$ is given by equation \ref{eq:standardBLP_utility}.
In the standard setting, $f_{u0}(x_{jt})=x_{jt}'\beta_0$ is linear in the observed characteristics $x_{jt}$.
The utility of the outside option is normalized to 
\begin{equation}\label{eq:standardBLP_utility_outside}
    u_{i0t} = \epsilon_{i0t}
\end{equation}
The idiosyncratic error terms $\epsilon_{ijt}$ are independent and identically distributed (i.i.d.) extreme value random variables. 
$\alpha_i$ is a random variable as consumers are assumed to have different preferences. 
We assume that $\alpha_i=\alpha_0+b_i$, where $\alpha_0$ is a parameter shifting all consumers' tastes for $p_{jt}$, and $b_i$ is a random variable with a pre-specified distribution $F(b_i,\sigma_0)$ that is determined by the parameter vector $\sigma_0\in\R^{d_\sigma}$.
To illustrate, consider two standard specifications for the distribution of $\alpha_i$:
\begin{itemize}
    \item Normal Distribution: If $\alpha_i \sim \mathcal{N}(\alpha_0, \sigma_0^2)$, then $\alpha_0$ is the mean of $\alpha_i$, and $b_i$ follows a centered normal distribution determined by the standard deviation $\sigma_0$.
    \item Lognormal Distribution: If $-\alpha_i \sim \log\mathcal{N}(\sigma_{10}, \sigma_{20}^2)$, then $\alpha_0=0$ and $b_i$ follows the negative of a lognormal distribution determined by the parameters $\sigma_0=(\sigma_{10}, \sigma_{20})$.
\end{itemize}
So the model defines a map from the parameters $\theta_0=(\sigma_0,\alpha_0,\beta_0')\in\Theta$, where $\Theta$ is a compact space, and the vector of product characteristics $(\xi_{jt},x_{jt}',p_{jt})$ to the market shares $s_{t}\in\R^{J}$.
Furthermore, we maintain mutual independence assumptions:  $\alpha_i$, $\epsilon_{ijt}$, and $(\xi_{jt},x_{jt},p_{jt})$ are mutually independent.
Consequently the utility can be rewritten as 
\begin{equation}
u_{ijt} = x_{jt}'\beta_{0} + p_{jt}\alpha_0+ \xi_{jt} + p_{jt}b_i+ \epsilon_{ijt}
\end{equation}
We define the linear component of the utility $y_{jt}=x_{jt}'\beta_{0} + p_{jt}\alpha_0+ \xi_{jt}$.
For each market $t$, we then define the corresponding market-level vectors: $y_t=(y_{1t},\dots,y_{Jt})'$, the price vector $p_t=(p_{1t},\dots,p_{Jt})'$, the unobserved characteristics vector $\xi_t=(\xi_{1t},\dots,\xi_{Jt})'$.
The parameters to be estimated are $\theta_0$. 
The individual market share of product $j$ in market $t$, $s_{jt}$ and the corresponding vector of all shares in the market, $s_t$, are given by
\begin{align}
    s_{jt} &=f_{js}(p_t,y_t,\sigma_0)\\
    &=\int \frac{\exp(\overbrace{x_{jt}'\beta_{0} + p_{jt}\alpha_0+ \xi_{jt}}^{y_{jt}} + p_{jt}b_i)}{1+\sum_{k=1}^{J}\exp(\underbrace{x_{kt}'\beta_{0} + p_{kt}\alpha_0+ \xi_{kt}}_{y_{kt}} + p_{kt}b_i)}dF(b_i,\sigma_0) \nonumber\\
    &=\int \frac{\exp(y_{jt}+p_{jt}b_i)}{1+\sum_{k=1}^{J}\exp(y_{kt}+p_{kt}b_i)}dF(b_{i},\sigma_0)\\
    s_t & = f_s(p_t,y_t,\sigma_0) := (f_{1s}(p_t,y_t,\sigma_0),\dots,f_{Js}(p_t,y_t,\sigma_0))'
\end{align}
This is sometimes referred to as a random coefficient logit model. 
Berry, Levinsohn, and Pakes (1995) show that given fixed $\sigma$, $p_t$ and $s_t$, there is a unique vector $y\in\R^{J}$ that solves the following equation
\footnote{Precisely, they show that given $p_t$, $s_t$, $x_t=(x_{1t},\cdots,x_{Jt})'$, $\beta,\alpha, \sigma$, there is a unique $\xi\in\R^J$ that solves the equation
\[
    s_t = f_s(p_t,x_t\beta+p_t\alpha+\xi,\sigma).
\]
Similarly, one can show the uniqueness of $y$, as it can be seen as a special case where $\beta,\alpha=0$.
}
\[
    s_t = f_s(p_t,y,\sigma)
\]
In other words, $f_s(p_t,\cdot,\sigma)$ is invertible for any fixed $\sigma\in\R^{d_\sigma}$ and $p_t$. Thus, we can impute $y$ for each $\sigma$. 
Define $f_s^{-1}(s_t,p_t,\sigma)\in\R^{J}$ as the unique solution to the above equation and imputed values of $y_t(\sigma)$ and $\xi_t(\theta)$ as follows:
\begin{align}
    y_t(\sigma)&:=f_s^{-1}(s_{t},p_t,\sigma) \label{eq:inverseBLP}\\
    \xi_t(\theta)&:=y_{t}(\sigma)-x_{t}\beta - p_{t}\alpha
\end{align}
where $\theta=(\sigma,\alpha,\beta')$.
Trivially, the mappings satisfy $y_t(\sigma_0)=y_t$, $\xi_t(\theta_0)=\xi_t$, i.e., these mappings map the true value $\sigma_0$ to the true values of linear component of the utility, $y_t$ and unobserved characteristics $\xi_t$.
Let $y_{jt}(\sigma)$ and $\xi_{jt}(\theta)$ be the $j$-th element of $y_t(\sigma)$ and $\xi_t(\theta)$, respectively.

\subsection{Estimation}
We here present the standard GMM estimator for the BLP model as proposed in \citet{berry1995automobile}.
The data are observed at the market level. 
One typically observes $(x_t,p_t,s_t,z_t)$ for $t=1,\dots,T$.
To address endogeneity, $z_t=(z_{1t},\dots,z_{Jt})'$ represents instrumental variables in market $t$.
Then we define $\tilde{z}_{jt}=(z_{jt}',x_{jt}')'$.
Then the estimator of $\theta_0$ is defined as
\[
    \hat{\theta}=arg\min\limits_{\theta}\hat{g}(\theta)'\hat{W}\hat{g}(\theta)
\]
where $\hat{g}(\theta)=\E_{JT}\xi_{jt}(\theta)\tilde{z}_{jt}$ and $\hat{W}$ is an estimator of a positive definite weighting matrix. 
In a two-step GMM estimation, for example, $\hat{W}$ is derived using the first step estimator, i.e., 
$\hat{W}=\left[\E_T\Big(\E_{J}\xi_{jt}(\tilde{\theta})\tilde{z}_{jt}\E_{J}\xi_{jt}(\tilde{\theta})\tilde{z}_{jt}'\Big)\right]^{-1}$, where the preliminary estimator $\tilde{\theta}$ is derived from an initial unweighted GMM minimization.
This estimator requires estimating all $d_x$ coefficients in $\beta_0$ jointly with $\sigma_0$ and $\alpha_0$, which becomes infeasible when $d_x$ is comparable to or exceeds $T$; the next section relaxes this requirement.

\section{BLP models with high-dimensional controls}
\subsection{Framework}
In this framework we generalize the utility function to the following form:
\begin{equation}\label{eq:highdimBLP_utility}
    u_{ijt} = f_{u0}(x_{jt}) + p_{jt}\alpha_0 + \xi_{jt} + p_{jt}b_i + \epsilon_{ijt}\quad \E[\xi_{jt}|x_{jt},z_{jt}]=0
\end{equation}
A key assumption that makes it possible to perform estimation and inference in such cases is the following:
\begin{align}\label{eq:highdimBLP_instrument}
    z_{jt}  &=f_{z0}(x_{jt})+\epsilon^z_{jt},\quad \E[\epsilon^z_{jt}|x_{jt}]=0
\end{align}
We here define the parameters of interest as $\theta_{10}=(\sigma_0,\alpha_0)$.
We adopt a triangular array setting that allows the functions $f_{z0}$ and $f_{u0}$, the dimension $d_x$, and the distribution of the variables to vary with the sample size $T$, though we suppress the dependence on $T$ for notational simplicity.
Meanwhile, the parameters $(\sigma_0,\alpha_0)$ are assumed to remain fixed throughout the triangular array.
This triangular array setting encompasses two distinct scenarios:

\begin{itemize}
    \item[\textbf{(i)}] \textbf{[Fixed target functions]:} $f_{u0}$ and $f_{z0}$ are fixed elements residing in an infinite-dimensional function space, so that a single pair of nonparametric functions is estimated throughout. What grows with $T$ is not the target but the device used to approximate it: $x_{jt}$ collects $d_x$ terms of a dictionary, and the approximation errors $r^{u}_{jt}$ and $r^{z}_{jt}$ of Section \ref{sec:prelim} shrink as $d_x$ increases. Given the underlying characteristics $\upsilon_{jt}$, one may read $x_{jt}=P(\upsilon_{jt})$ as a series expansion in $\upsilon_{jt}$.
    \item[\textbf{(ii)}] \textbf{[As an approximation]:} the triangular array is a device rather than a literal description of the sampling process. A researcher faces one dataset in which $d_x$ is large relative to $T$; embedding it in a sequence along which $d_x$ grows with $T$ yields an asymptotic approximation that keeps the cost of estimating the nuisance parameters visible, whereas holding $d_x$ fixed would make that cost vanish in the limit and would misrepresent the finite-sample problem.
\end{itemize}

\begin{remark}
    Scenario (i) includes the following setting as a special case: $f_{u0}(x^{\infty}_{jt})=\beta_0^{\infty\prime}x^{\infty}_{jt}$, where $x^{\infty}_{jt}\in\R^{\infty}$ and $\beta_0^{\infty}\in\R^{\infty}$. The econometrician observes only a finite truncation of $x^{\infty}_{jt}$, denoted $x_{jt}=(x_{jt,1},\dots,x_{jt,d_x})'$, and $d_x\to\infty$ is interpreted as data availability increasing with the sample size. The target function is again fixed; only the truncation point moves. For notational brevity we continue to write $f_{u0}(x_{jt})$, treating $f_{u0}$ as a function of the finite-dimensional vector actually observed.
\end{remark}

In the following sections, we treat $(f_{z0},f_{u0})$ as nuisance parameters, while considering $\theta_{10}=(\sigma_0,\alpha_0)$ as structural parameters of interest. 
However, this designation is flexible: under the specification $f_{u0}(x_{jt})=x_{jt}'\beta_0$, any subset of $\beta_0$ can be included as parameters of interest, with the remaining elements treated as nuisance parameters.

\subsection{Neyman Orthogonal Estimation of $\theta_{10}$}
The estimation strategy adheres to a standard Neyman orthogonal approach and proceeds in two steps:
\begin{enumerate}
    \item \textbf{Derivation of a preliminary estimator}: Obtain estimators of $f_{z0}$ and $f_{u0}$ with convergence rates faster than $T^{-1/4}$, with machine learning techniques, such as LASSO.
    \item \textbf{Application of the Neyman orthogonal approach}: 
    Refine the preliminary estimator using the Neyman orthogonal estimator to obtain the final estimator $\check{\theta}_1$ that achieves the desired $\sqrt{T}$-asymptotic normality.
\end{enumerate}
In the case where $\sigma_0$ is known, our model collapses to the partially linear IV framework described by Eqs. 4.5–4.6 of \citet{chernozhukov2018double}, with the Neyman orthogonal moment function correspondingly reducing to their Equation (4.7).
\footnote{
    They consider the partially linear IV model
    \begin{align*}
        Y & = D\theta_0 + g_0(X) + U,\quad \E[U|X,Z]=0\\
        Z & = m_0(X) + V,\quad \E[V|X]=0 
    \end{align*}
    and define the Neyman orthogonal moment function as
    \[
        \psi(w,\theta,g,m) = (Z-m(X))\cdot(Y-D\theta-g(X))
    \]
    where $w=(Y,D,Z,X)$, $g$ and $m$ are the nuisance functions.
}
While our approach differs slightly from the ``partialling-out'' Neyman orthogonal moment functions—such as those defined in Eq. 4.8 of \citet{chernozhukov2018double}
\footnote{
    Equation 4.8 defines the Neyman orthogonal moment function
    \[
        \psi(w,\theta,l,m,r) = (Z-m(X))\cdot (Y-l(X)-\theta(D-r(X)))
    \]
    where $l_0(X) = \E[Y|X]$, $m_0(X)=\E[Z|X]$ and $r_0(X)=\E[D|X]$ are the nuisance functions.
} or Algorithm 1 of \citet{chernozhukov2015post}—it maintains the Neyman orthogonality condition.

Suppose we have obtained preliminary estimators $\hat{f}_z$, $\hat{f}_u$ that satisfy the desired preliminary convergence rates and we are interested in $\theta_{10}=(\sigma_0,\alpha_0)\in\R^{d_{\theta_1}}$, a subset of all parameters, where $d_{\theta_1}=d_{\sigma}+d_{\alpha}$. Throughout this paper, we assume that $d_{\alpha}=1$ for notational clarity.
Let $w_t=(x_t,z_t,p_t,s_t)$ represent observations in market $t$. 
$(w_t)_{t=1}^{T}$ is modeled as independent and identically distributed.
A natural approach to estimation of $\theta_{10}$ would be, for example, simply plugging in the preliminary estimator $\hat{f}_u$ into the moment function and then minimizing the quadratic form of the empirical moment function, i.e.,
\[
    \check{\theta}_1^{(1)}=\arg\min_{\theta_1}\hat{g}(\theta_1,\hat{f}_u)'W\hat{g}(\theta_1,\hat{f}_u)
\]
where $\hat{g}(\theta_1,\hat{f}_u)=\E_{T}[\E_J\tilde{z}_{jt}\cdot(y_{jt}(\sigma)-\alpha p_{jt}-\hat{f}_u(x_{jt}))]$, $\tilde{z}_{jt}=(z_{jt}',x_{jt}')'$ and $W$ is a positive definite weighting matrix.
The estimator $\check{\theta}_1^{(1)}$ will generally have a slower than $1/\sqrt{T}$ convergence rate, as the argument in \citet{belloni2018high} shows.
The reason is that plugging in $\hat{f}_u$ leaves a first-order bias in the expansion of the GMM estimator.
A Neyman orthogonal moment removes this bias. The remainder then contains only terms such as
\[
    \sqrt{T}\E_T[\E_J(\hat{f}_z(x_{jt})-f_{z0}(x_{jt}))\cdot\xi_{jt}]+\sqrt{T}\E_T[\E_J\epsilon^{z}_{jt}\cdot (\hat{f}_u(x_{jt})-f_{u0}(x_{jt}))]
\]
which cross-fitting and the i.i.d.\ assumption make easy to control. We therefore use the Neyman orthogonal moment function
\begin{align}\label{eq:NOmoment}
    \psi(w_t,\theta_1,f_z,f_u)&=\E_J\underbrace{(z_{jt}-f_z(x_{jt}))}_{d_z\times 1}\cdot\underbrace{(y_{jt}(\sigma)-\alpha p_{jt}-f_u(x_{jt}))}_{1\times 1}=:\E_J\phi_j(w_{t},\theta_1,f_z,f_u)
\end{align}
where $f_z$ and $f_u$ are the nuisance functions.
The true values of the nuisance functions are $f_{z0}(x_{jt})$ and $f_{u0}(x_{jt})$.
This score coincides with the efficient orthogonal GMM score of \citet{chernozhukov2018double} in the linear case, as shown in Proposition \ref{prop:equivalence}.

\begin{proposition}\label{prop:equivalence}
Suppose $f_{u0}(x)=x'\beta_0$, $f_{z0}(x)=\Pi_0x$, and
$\Sigma_{xx}=\mathbb{E}\,\mathbb{E}_J[x_{jt}x_{jt}']$ is nonsingular. Let
\[
  m(w_t,\theta_1,\beta)=\mathbb{E}_J\big[\tilde z_{jt}\,(y_{jt}(\sigma)-\alpha p_{jt}-x_{jt}'\beta)\big],
  \qquad \tilde z_{jt}=(z_{jt}',x_{jt}')'.
\]
Let $A=\partial_{\theta_1'}\mathbb{E}\,m(w_t,\theta_{10},\beta_0)$,
$G_\beta=\partial_{\beta'}\mathbb{E}\,m(w_t,\theta_{10},\beta_0)$, $V=\mathrm{Var}\big(m(w_t,\theta_{10},\beta_0)\big)$, and
\[
  \mu_0=A'V^{-1}-A'V^{-1}G_\beta\big(G_\beta'V^{-1}G_\beta\big)^{-1}G_\beta'V^{-1}
\]
be the efficient orthogonalization matrix of \citet[eq.~(2.18)]{chernozhukov2018double}.
Then, for all $\theta_{1}$ and $\beta$,
\[
  \mu_0\,m(w_t,\theta_1,\beta)=G\,\Omega^{-1}\,\psi(w_t,\theta_1,f_{z0},f_{u}), \qquad f_{u}(x_{jt})=x_{jt}'\beta,
\]
where $G'=\E_J\E[\epsilon_{jt}^z(\frac{\partial}{\partial \sigma'}y_{jt}(\sigma_0),-p_{jt})]$ and $\Omega=\Var(\E_J\epsilon_{jt}^{z}\xi_{jt})$. 
\end{proposition}

This Neyman orthogonal moment function has the same asymptotic variance as the ``partialling-out'' moment function defined as
\begin{equation*}
    \tilde{\psi}(w_t,\theta_1,\tilde f_z,\tilde f_y,\tilde f_p) = \E_J\underbrace{(z_{jt}- \tilde f_z(x_{jt}))}_{d_z\times 1}\cdot\underbrace{(y_{jt}(\sigma)- \tilde f_y(x_{jt})-\alpha (p_{jt}-\tilde f_p(x_{jt})))}_{1\times 1}
\end{equation*}
where $\tilde f_z, \tilde f_y, \tilde f_p$ are the nuisance functions.
The true values of these nuisance functions are $\tilde f_{z0}(x_{jt})=\E[z_{jt}|x_{jt}]$, $\tilde f_{y0}(x_{jt})=\E[y_{jt}(\sigma_0)|x_{jt}]$ and $\tilde f_{p0}(x_{jt})=\E[p_{jt}|x_{jt}]$.
It can be seen that $ \tilde f_{y0}(x_{jt})-\alpha_0  \tilde f_{p0}(x_{jt})=f_{u0}(x_{jt})$. So at the true values of the nuisance functions the two moment functions coincide at $\theta_1=\theta_{10}$, and so do their Jacobians with respect to $\theta_1$. Hence the two moment functions yield the same asymptotic variance.
The reason for using the Neyman orthogonal moment function defined in equation (\ref{eq:NOmoment}) is that it allows for a more flexible choice of parameters of interest. We may be interested in some elements of $\beta_0$ in $f_{u0}(x_{jt})=x_{jt}'\beta_0$, and the Neyman orthogonal moment function can be easily modified to accommodate this choice, while the ``partialling-out'' moment function cannot.
We note that the Neyman orthogonal moment function $\psi(w_t,\theta_1,f_z,f_u)$ implicitly depends on the sample size $T$, as the dimension of $w_t$ may depend on $T$.
For notational brevity, we suppress this dependence throughout the paper.

Partition the observation indices $\{1,\ldots,T\}$ into $L$ groups $I_l\; (l=1,...,L)$. 
For each $l$, construct estimators $\hat{f}_z^{l}$ and $\hat{f}_u^{l}$ using all observations \textbf{not} in $I_l$.
Then plug in the nuisance parameter estimators $\hat{f}_z^{l}$ and $\hat{f}_u^{l}$ into the Neyman orthogonal moment function
and derive the empirical Neyman orthogonal moment function as follows:
\[
    \psi_t(\theta_1,\hat{f}_z^{l},\hat{f}_u^{l})=\psi(w_t,\theta_1,\hat{f}_z^{l},\hat{f}_u^{l})
\]
and
\[
    \E_{T,L}[\psi_t(\theta_1,\hat{f}_z^{l},\hat{f}_u^{l})]=\frac{1}{T}\sum_{l = 1}^{L} \sum_{t\in I_l}\psi_t(\theta_1,\hat{f}_z^{l},\hat{f}_u^{l}).
\]
Then we have the estimator for $\theta_1$ as
\begin{equation}\label{eq:orthogonal_estimator}
    \check{\theta}_1=\arg\min_{\theta_1\in\Theta_\sigma\times\Theta_\alpha}\{(\E_{T,L}[\psi_t(\theta_1,\hat{f}_z^{l},\hat{f}_u^{l})])'\hat{W}(\E_{T,L}[\psi_t(\theta_1,\hat{f}_z^{l},\hat{f}_u^{l})])\}
\end{equation}
where $\hat{W}$ can be $I_{d_{z}}$ or an estimated weighting matrix, $\Theta_\alpha, \Theta_{\sigma}$ are compact sets.

Cross-fitting allows the first-stage estimators to enter the theory only through their convergence rates. 
Because $\hat f^l_z$ and $\hat f^l_u$ are constructed from observations outside $I_l$, they are fixed functions
conditional on the estimation sample and independent of the observations at which they are evaluated. 
The terms through which first-stage estimation error enters the empirical moment can therefore be controlled using only the mean-square convergence of $\hat f^l_z$ and $\hat f^l_u$ at a rate faster than $T^{-1/4}$, irrespective of how the estimators are constructed. 
The results below thus apply to any first-stage learner attaining this rate; for instance, $f_{z0}$ is the regression of observed instruments on product characteristics and may be estimated by random forests or neural networks.
Without cross-fitting, the estimators depend on the same observations as the empirical moment, and these terms instead have to be controlled uniformly over a class of functions that contains the estimators with probability approaching one \citep{chen2026cross}. 
This requires a bound on the complexity of that class, which has to be verified separately for each learner. For the LASSO estimators of Section~\ref{sec:prelim}, for instance, the results of \citet[Theorem 2]{chen2026cross} require the fitted models to be sparse with probability
approaching one, together with a sparsity condition stronger than the one we impose there. \footnote{Theorem 2 of \citet{chen2026cross} requires $d_0\log(d_x\vee T)/T^{1/4}=o(1)$}
We therefore retain cross-fitting, which keeps the theory independent of the first-stage learner.

\subsection{Theory for the asymptotic normality of $\check{\theta}_1$}
In this subsection, we provide a set of conditions, the nuisance parameter convergence rates and regularity requirements,
that are sufficient for the asymptotic normality of Neyman orthogonal estimator $\check{\theta}_1$.

\begin{condition}\label{cond:Regularity1}\leavevmode
    \begin{enumerate}[label=(\roman*), ref=\thecondition(\roman*)]
        \item The data $(w_t)_{t=1}^T=(x_t,z_t,p_t,s_t)_{t=1}^T$ are independent and identically distributed.\label{cond:IID}
        \item Data $(w_t)_{t=1}^T$ obey the model (\ref{eq:highdimBLP_utility}) and (\ref{eq:highdimBLP_instrument}).\label{cond:Exogeneity}
        \item $f(b,\sigma)$ is the density function of $b_i$. $f(b,\cdot)$ is $C^{2}(\Theta_\sigma)$ a.s. and satisfies $\int \sup_{\Theta_\sigma}\|\partial_\sigma f(b,\sigma)\|db\le M$ and $\int \sup_{\Theta_\sigma}\|\partial^2_\sigma f(b,\sigma)\|db\le M$.
        \label{cond:SmoothDensity1}
    \end{enumerate}
\end{condition}

\begin{remark}
Conditions \ref{cond:IID} and \ref{cond:Exogeneity} are standard regularity conditions in the BLP literature.
Condition \ref{cond:SmoothDensity1} is a regularity condition on the density function of the random coefficients, ensuring the interchange of differentiation and integration and the smoothness of the inverse function $y_{jt}(\sigma)$. This is a mild condition, as many density functions, such as the normal and lognormal distributions, satisfy this condition.
\end{remark}

\begin{condition}[\textbf{Boundedness for Neyman Orthogonal Estimator}]\label{cond:BoundednessNO}\leavevmode
    \begin{enumerate}[label=(\roman*), ref=\thecondition(\roman*)]
        \item $\forall\; j,t$, random variables satisfy $|p_{jt}|$, $\normm{\epsilon^{z}_{jt}}_\infty$, $|\xi_{jt}|$, $\norm{f_{u0}(x_{jt})}_\infty$, $\norm{f_{z0}(x_{jt})}_\infty \le M$ a.s. Parameters satisfy $\norm{\theta_{10}}_1\leq M$. $d_z$ is fixed.
            \label{cond:BoundednessVar1}
        \item Parameter spaces $\Theta_\alpha$ and $\Theta_\sigma$ are fixed compact sets. $\sigma_0$ and $\alpha_0$ lie in the interior of $\Theta_{\sigma}$ and $\Theta_{\alpha}$
         \label{cond:compactnessNO}
    \end{enumerate}
\end{condition}

\begin{remark}
Condition \ref{cond:BoundednessVar1} bounds the parameters and random variables in the model.
One may replace the boundedness condition with conditions on moments, as stated in Conditions SE and SM in \citet{belloni2014inference}.
Here we choose the boundedness condition for simplicity.
Condition \ref{cond:compactnessNO} ensures that the parameter space is compact, which is a standard assumption in the BLP literature.
\end{remark}

\begin{condition}\label{cond:IdMomentWeight}\leavevmode
    \begin{enumerate}[label=(\roman*), ref=\thecondition(\roman*)]
        \item $\E[\psi(w_t,\theta_1,f_{z0},f_{u0})]=0$ iff $\theta_1=\theta_{10}$. For any $\epsilon>0$, there exists $\eta>0$ such that\\
         $\inf\limits_{\normm{\theta_1-\theta_{10}}_2>\epsilon,\theta_1\in\Theta_\sigma\times\Theta_\alpha}\|\E\psi(w_t,\theta_1,f_{z0},f_{u0})\|_2\ge \eta$. \label{cond:IdentificationMoment}
        \item $\normm{\hat{W}-W}_2=o_p(1).$ $\hat{W}$ and $W$ are positive definite. $c\leq\lambda_{min}(W)\leq \lambda_{max}(W)\leq C$
            \label{cond:WeightMatrix}
    \end{enumerate}
\end{condition}

\begin{remark}
Here we suppress the dependence of $W$ and $\hat{W}$ on $T$ for notational simplicity.
Condition \ref{cond:IdentificationMoment} is the standard identification condition. 
In a fixed-dimensional model with a fixed data distribution, its second part follows from the first: if $\E[\psi(w_t,\theta_1,f_{z0},f_{u0})]$ is continuous in $\theta_1$ and $\Theta_\sigma\times\Theta_\alpha$ is compact, a unique zero at $\theta_{10}$ is automatically well separated. 
In high-dimensional settings, this inequality remains the fundamental identification condition: because the distribution of the data changes as the dimension grows, identification would fail if there existed some other $\theta_1$ for which $\normm{\E[\psi(w_t,\theta_1,f_{z0},f_{u0})]}_2$ converges to zero as $T\to\infty$.
\citet{freyberger2015asymptotic} assumes a slightly weaker condition in the fixed-dimensional setting \footnote{\citet{freyberger2015asymptotic} assumes: for all $\delta>0$, $\exists C(\delta)$ s.t.
\[ \lim P(\inf_{\theta\notin\mathcal{N}_{\theta_0}(\delta)}\norm{\E_T\psi(\theta)-\E_T\psi(\theta_0)}_2\geq C(\delta))=1
    \]
    where $\psi$ here is the moment function, and $\mathcal{N}_{\theta_0}(\delta)$ is a $\delta$-neighborhood of $\theta_0$.
    However there is no significant difference.
}.
Condition \ref{cond:WeightMatrix} is the standard weighting matrix condition.
\end{remark}

\begin{condition}\label{cond:IdentificationG}\leavevmode
    \begin{enumerate}[label=(\roman*), ref=\thecondition(\roman*)]
        \item The matrix $G=\E_J\E[
    \underbrace{\begin{bmatrix}
        \frac{\partial}{\partial \sigma}y_{jt}(\sigma_0)\\
        - p_{jt}
    \end{bmatrix}}_{(d_{\theta_1})\times 1}\label{cond:IdentificationGWG}
    \cdot \underbrace{\epsilon_{jt}^{z\prime}}_{1\times d_z}]$ has full rank with $\lambda_{min}(GG')>c>0$.
        \item The matrix $\Omega=var(\E_J\epsilon^z_{jt}\xi_{jt})$ is positive definite with $\lambda_{min}(\Omega)>c>0$.
        \label{cond:IdentificationOmega}
    \end{enumerate}
\end{condition}

\begin{remark}
Unlike identification in fixed-dimensional contexts, Condition \ref{cond:IdentificationG} imposes the inequality uniformly over $T$.
This requirement is crucial as both the underlying data distribution, $f_{u0}$ and $f_{z0}$ are allowed to vary with $T$.
Conditions \ref{cond:IdentificationGWG} and \ref{cond:IdentificationOmega} are commonly assumed in the existing BLP literature. 
For example \citet{berry2004limit} and \citet{freyberger2015asymptotic} assume that $\Gamma=\frac{\partial}{\partial \theta}\E[g(\theta)]|_{\theta=\theta_0}$ has full rank, which is analogous to Condition \ref{cond:IdentificationGWG}.
Assumption B.10 of \citet{freyberger2015asymptotic} also assumes that the matrix $var(\E_J\epsilon^z_{jt}\xi_{jt})$ is positive definite, which is the same as Condition \ref{cond:IdentificationOmega}.
Condition \ref{cond:IdentificationGWG} depicts the local property of the Neyman orthogonal moment function around the true parameter value and is crucial for the asymptotic normality of $\check{\theta}_1$, whereas Condition \ref{cond:IdentificationMoment} depicts the global property of the moment function and is crucial for the consistency of $\check{\theta}_1$.
\end{remark}

\begin{condition}\label{cond:PrelimConvergence}
    The preliminary estimators $\hat{f}_z^l$ and $\hat{f}_{u}^l$ satisfy the following convergence rates for each $l=1,\dots,L$ and $j=1,\dots,J$:
    \begin{align*}
        \normm{\hat{f}_z^l(x_{jt})-f_{z0}(x_{jt})}_{L^2}&=o_p(T^{-1/4})\\
        \normm{\hat{f}_{u}^l(x_{jt})-f_{u0}(x_{jt})}_{L^2}&=o_p(T^{-1/4})
    \end{align*}
\end{condition}

\begin{remark}
    The convergence rates of the preliminary estimators are crucial for the asymptotic normality of $\check{\theta}_1$.
    This convergence rate requirement is standard in Neyman orthogonal estimation literature, such as \citet{chernozhukov2018double}.
\end{remark}

\begin{theorem}[Consistency]
    \label{thm:Consistency}
    Suppose $\hat{f}_z^l$ and $\hat{f}_u^l$ are constructed using all observations \textbf{not} in $I_l$, for $l=1,\dots,L$.
    If Conditions \ref{cond:Regularity1}, \ref{cond:BoundednessNO}, \ref{cond:IdMomentWeight}, \ref{cond:PrelimConvergence} hold, then we have
    \[
        \check{\theta}_1\overset{p}{\to}\theta_{10}.
    \]
\end{theorem}

\begin{theorem}[Asymptotic Normality]
    \label{thm:AsympNorm}
    Suppose Conditions \ref{cond:Regularity1}, \ref{cond:BoundednessNO}, \ref{cond:IdMomentWeight}, \ref{cond:IdentificationG} and \ref{cond:PrelimConvergence} hold. 
     Then we have
    \[
        P^{-1/2}\sqrt{T}(\check{\theta}_1-\theta_{10})\overset{d}{\to}N(0,I_{d_{\theta_1}}),
    \]
    where $P=(GWG ')^{-1}GW\Omega WG '(GWG ')^{-1}\in\R^{d_{\theta_1}\times d_{\theta_1}}$ is a symmetric matrix, $\hat{W}\to_p W$, and $\Omega=var(\E_J\epsilon^z_{jt}\xi_{jt})$,
     and $G=\E_J\E[
            \begin{bmatrix}
                \frac{\partial }{\partial \sigma}y_{jt}(\sigma_0)\\
                - p_{jt}
            \end{bmatrix}
            \cdot \epsilon^{z\prime}_{jt}]$
\end{theorem}

\begin{remark}
This theorem enables consistency and asymptotic normality for a broad class of functionals.
For example, when $f_{u0}(x_{jt})=x_{jt}'\beta_0$, the framework can be extended to
$\theta_{10}=(\sigma_0,\alpha_0,\beta_{0,k})$ by augmenting the instruments with $x_{jt,k}$
and residualizing them on the remaining characteristics $x_{jt,-k}$.
With the rank condition in Condition \ref{cond:IdentificationG} strengthened and the rate condition in Condition \ref{cond:PrelimConvergence} adjusted accordingly, one can show that the ratio estimator $\check{\alpha}/\check{\beta}_{k}$ of $\alpha_0/\beta_{0,k}$ is asymptotically normal.
Furthermore, under mild regularity conditions, this theorem implies the consistency of the estimated price elasticity, which enters the model as a functional of the structural parameters.
\end{remark}

\section{LASSO Estimation of Nuisance Parameters $f_{z0}$, $f_{u0}$}\label{sec:prelim}
This section provides one estimator satisfying Condition~\ref{cond:PrelimConvergence} under approximate sparsity.
To achieve the desired convergence rate, we impose an approximate sparsity assumption on functions $f_{u0}$ and $f_{z0}$ as well as the auxiliary function $f_{p0}$
\begin{equation}\label{eq:highdimBLP_price}
    p_{jt}=f_{p0}(x_{jt},z_{jt})+\epsilon^p_{jt}, \quad \E[\epsilon^p_{jt}\cdot(z_{jt}',x_{jt}')]=0, \quad \E[\epsilon^p_{jt}\cdot f_{p0}(x_{jt},z_{jt})]=0
\end{equation}
To illustrate the estimation strategy under this approximate sparsity assumption, we first define the empirical moment function,  \begin{align*}
    \hat{g}^l(\theta)=\frac{1}{J\cdot |I_l^c|}\sum_{t\in I_l^c}\sum_{j=1}^J \tilde{z}_{jt}(y_{jt}(\sigma)-\alpha p_{jt}-x_{jt}'\beta),
\end{align*}
where $\tilde{z}_{jt}=(z_{jt}',x_{jt}')'\in\mathbb{R}^{d_{\tilde{z}}}$, and $I_l^c$ is the index set of the observations not in the $l$-th group of the sample splitting. Define $g(\theta)$ as its population counterpart, i.e., $g(\theta)=\E[\hat{g}^l(\theta)]$.
Then we present the following approximate sparsity assumption on the functions $f_{z0}$, $f_{u0}$ and $f_{p0}$:
\begin{align}
    f_{u0}(x_{jt})&=x_{jt}'\beta_0 + r_{jt}^u,\quad \E (r_{jt}^{u})^2\leq c_u^2\lesssim d_0/T,\quad \normm{\beta_0}_0 \leq d_0 \label{eq:highdimBLP_approx_u}\\
    f_{z0}(x_{jt})&=\Pi_0x_{jt} + r_{jt}^z,\quad \E \normm{r_{jt}^{z}}_2^2\leq c_z^2\lesssim d_0/T,\quad \normm{\Pi_0}_0 \leq d_0 \label{eq:highdimBLP_approx_z}\\
    f_{p0}(x_{jt},z_{jt})&=z_{jt}'\beta_{pz0} + x_{jt}'\beta_{px0} + r_{jt}^p,\quad \E (r_{jt}^{p})^2\leq c_p^2\lesssim d_0/T,\quad \normm{(\beta_{pz0}',\beta_{px0}')}_0 \leq d_0 \label{eq:highdimBLP_approx_p} 
\end{align}
where $r_{jt}^p:=f_{p0}(x_{jt},z_{jt})-x_{jt}'\beta_{px0}-z_{jt}'\beta_{pz0}$, $r_{jt}^u:=f_{u0}(x_{jt})-x_{jt}'\beta_0$ and $r_{jt}^z:=f_{z0}(x_{jt})-\Pi_0x_{jt}$ are the approximation errors.
This motivates the use of the LASSO method to estimate $f_{p0}$ and $f_{z0}$, as studied in \citet{belloni2014inference}.

\begin{remark}
This assumption extends the sparsity assumption of \citet{chernozhukov2015post} to BLP models, and the same setting has been adopted in \citet{gillen2019blp}.

Sparsity of $f_{u0}$ relaxes the assumptions used in the existing BLP literature. Selecting a small subset of product characteristics on an ad hoc basis is equivalent to assuming that the true function $f_{u0}$ is exactly sparse \citep{chidmi2007brand}. The principal-component approach of \citet{backus2021common} can likewise be viewed as choosing the principal components as the dictionary, under which the parameter vector is exactly sparse.

Sparsity of $f_{z0}$ likewise accommodates instruments already in use. Suppose $z_{jt}$ is a Local Differentiation IV as in \citet{gandhi2019measuring}. Then $f_{z0}(x_{jt})$ is the expected number of local competitors of a product, and the residual $\epsilon_{jt}^z$ is the part of this count left unexplained by the expectation---arising, for example, because new products enter different markets at different times. Sparsity of $f_{z0}$ then amounts to assuming that the expected number of local competitors is determined largely by a small number of product characteristics (e.g., brand, sugar content, flavour).

Sparsity of $f_{p0}$ inherits the assumption of \citet{belloni2012sparse}. It is worth noting that this assumption may instead be replaced by $f_{p0}$ being an unknown function of low-dimensional features, estimable by any nonparametric method such as series regression; in this case regressors $x_{jt}$ can be viewed as basis functions.
\end{remark}

Relying on the estimator $\hat{f}_p$, with a convergence rate faster than $T^{-1/4}$, and the approximate sparsity on $f_{u0}$, we can combine two separate $\ell_1$-penalized minimizations to obtain a preliminary estimator $\hat{f}_u$ with a convergence rate faster than $T^{-1/4}$.

\subsection{Estimation}\label{sec:prelim_estimation}
\subsubsection{Estimation of $f_{z0}$ and $f_{p0}$}
Note that equations (\ref{eq:highdimBLP_approx_z}) and (\ref{eq:highdimBLP_approx_p}) imply that instruments $z_{jt,i}$, where $i=1,\dots,d_z$, can be represented as a linear combination of product characteristics $x_{jt}$ plus the approximation error and the noise term.
Similarly, price $p_{jt}$ can be represented as a linear combination of product characteristics $x_{jt}$ and instruments $z_{jt}$ plus the approximation error and the noise term.
\begin{align*}
    z_{jt,i}&=x_{jt}'\Pi_i + r_{jt,i}^z+\epsilon^{z}_{jt,i}\\
    p_{jt}  &=z_{jt}'\beta_{pz0} + x_{jt}'\beta_{px0} + r_{jt}^p+\epsilon^{p}_{jt}
\end{align*}
where $\Pi_i$ is the i-th row of $\Pi_0$, i.e., $\Pi_0=(\Pi_1,\dots,\Pi_{d_z})'$.
We apply the approach described in \citet{belloni2014inference} geared for non-Gaussian cases to each equation.
With estimated $\hat{\Pi}_i$ and $(\hat{\beta}_{pz}',\hat{\beta}_{px}')$, we can construct the estimators of $f_{z0}$ and $f_{p0}$.
More formally, for each $l=1,\dots,L$, consider the sample $I_l^c$ after the sample splitting. We define
\begin{equation}\label{eq:Pi_hat}
    \hat{\Pi}_i^l = \arg\min\limits_{\Pi_i}\frac{1}{J\cdot |I_l^c|}\sum_{t\in I_l^c}\sum_{j=1}^J(z_{jt,i}-\Pi_i' x_{jt})^2 +\lambda_\Pi\norm{\Pi_i}_1
\end{equation}
and
\begin{equation}\label{eq:betap_hat}
    (\hat{\beta}_{pz}^{l\prime},\hat{\beta}_{px}^{l\prime}) = 
    \arg\min\limits_{\beta_{pz},\beta_{px}}\frac{1}{J\cdot |I_l^c|}\sum_{t\in I_l^c}\sum_{j=1}^J(p_{jt}-z_{jt}'\beta_{pz}-x_{jt}'\beta_{px})^2 +\lambda_{\beta_{pz},\beta_{px}}\normm{(\beta_{pz}',\beta_{px}')'}_1 
\end{equation}

The penalty terms $\lambda_\Pi$ and $\lambda_{\beta_{pz},\beta_{px}}$ are chosen by the user.
The simulation applies cross-validation to select $\lambda_\Pi$ and $\lambda_{\beta_{pz},\beta_{px}}$.
Alternatively, the researcher may adopt the theoretically suggested values 
$\lambda_\Pi=C_{\Pi}\sqrt{\log(d_x\vee T)/T}$ and $\lambda_{\beta_{pz},\beta_{px}}=C_{\beta_{pz},\beta_{px}}\sqrt{\log(d_x\vee T)/T}$, where $C_{\Pi}$ and $C_{\beta_{pz},\beta_{px}}$ are constants.
Then define $\hat{\Pi}^l=(\hat{\Pi}_1^l,\dots,\hat{\Pi}_{d_z}^l)'$.
Consequently we have $\hat{f}_z^l(x_{jt})=\hat{\Pi}^lx_{jt}$ and $\hat{f}_p^l(x_{jt},z_{jt})=z_{jt}'\hat{\beta}_{pz}^l+x_{jt}'\hat{\beta}_{px}^l$.

\subsubsection{Estimation of $f_{u0}$}
Define $\theta=(\sigma,\alpha,\beta)$ and $\theta_0=(\sigma_0,\alpha_0,\beta_0)$ with $\beta_0$ as in \eqref{eq:highdimBLP_approx_u}.
This part combines an $\ell_1$-penalized minimum distance estimator, motivated by the Restricted Minimum Distance (RMD) estimator of \citet{belloni2018high}, with the LASSO.
Two features of the moment function $g(\theta)$ make the problem difficult.
First, $g(\theta)$ is nonlinear in $\sigma$, and the resulting criterion is nonconvex in $\sigma$.
Second, the coefficient vector $\beta$ is high-dimensional, with $d_x$ possibly exceeding $T$.
We therefore proceed in two steps.
Step A delivers a preliminary estimator $\tilde{\sigma}$ of the nonlinear parameter at a sufficiently fast rate.
Given $\tilde{\sigma}$, the mean utility $y_{jt}(\tilde{\sigma})$ is approximately linear in $x_{jt}$ and $p_{jt}$, and Step B estimates $f_{u0}$ by a LASSO regression.
More formally, for each $l=1,\dots,L$, the estimator of $f_{u0}$ is constructed from the observations in $I_l^c$ as follows:
\begin{itemize}
    \item Step A: An $\ell_1$-penalized minimum distance problem
    \begin{equation}\label{eq:theta_tilde}
        \tilde{\theta}^l=(\tilde{\sigma}^l,\tilde{\alpha}^l,\tilde{\beta}^l)=
        \arg\min_{\substack{\sigma\in\Theta_\sigma, \alpha\in\Theta_\alpha\\
        \beta\in \R^{d_x}}}\{\norm{\hat{g}^l(\theta)}_\infty+\lambda_{\tilde{\theta}}\norm{\beta}_1\}
    \end{equation}
    \item Step B: Run a LASSO regression of $y_{jt}(\tilde{\sigma}^l)$ on $x_{jt}$ and the ``predicted price'' $\hat{f}_{p}^{l}(x_{jt},z_{jt})$, which replaces the actual price $p_{jt}$:
    \begin{equation}\label{eq:theta_hat}
        (\hat{\alpha}^l,\hat{\beta}^l)=\arg\min_{\alpha,\beta}\{\frac{1}{J\cdot |I_l^c|}\sum_{t\in I_l^c}\sum_{j=1}^J[(y_{jt}(\tilde{\sigma}^l)-\alpha\hat{f}_{p}^{l}(x_{jt},z_{jt})-x_{jt}'\beta)^2]+\lambda_{\beta}\normm{\beta}_1\}
    \end{equation}
    Then define $\hat{\sigma}^l=\tilde{\sigma}^l$ and $\hat{\theta}^l=(\hat{\sigma}^l,\hat{\alpha}^l,\hat{\beta}^l)$.
\end{itemize}

Within Step A, the two difficulties are handled by different devices.
The nonconvexity is confined to $\sigma$. For each fixed $\sigma$, $y_{jt}(\sigma)$ is obtained from the share inversion and $\hat{g}^l(\theta)$ is affine in $(\alpha,\beta')'$, so Step A can be computed as the profiled problem
\begin{align*}
    \tilde{\theta}^l = \arg\min_{\sigma\in\Theta_\sigma}\ \min_{\alpha\in\Theta_\alpha,\beta\in\R^{d_x}}
    \Big\{\norm{\hat{g}^l(\sigma,\alpha,\beta)}_\infty
    +\lambda_{\tilde{\theta}}\norm{\beta}_1\Big\},
\end{align*}
in which the inner problem is a linear program in $(\alpha,\beta')'$ and the outer minimization is a grid search over the fixed-dimensional compact set $\Theta_\sigma$.
The high dimensionality is handled by the $\ell_1$ penalty on $\beta$ together with the $\ell_\infty$ norm of the moment vector.
The choice of the $\ell_\infty$ norm matters for the rate.
The sampling error in each of the $d_{\tilde{z}}$ components of $\hat{g}^l(\theta_0)$ is $O_p(T^{-1/2})$, so its $\ell_2$ norm is typically of order $\sqrt{d_{\tilde{z}}/T}$, which does not vanish when $d_{\tilde{z}}\gtrsim T$.
The $\ell_\infty$ norm tracks only the largest deviation, which for bounded moments is of order $\sqrt{\log(d_{\tilde{z}}\vee T)/T}$ (Lemma \ref{lemma:Hoeffding}).
Together with the approximation error $r^u_{jt}$, this bound determines both the order of the tuning parameter $\lambda_{\tilde{\theta}}$ and the rate of $\tilde{\sigma}$.
The penalty applies only to $\beta$: $(\sigma,\alpha)$ is of fixed dimension and requires no regularization.
Step B exploits the approximate linearity of $y_{jt}(\tilde\sigma^l)$ noted above. Because $p_{jt}$ is correlated with $\xi_{jt}$, it is replaced by its projection $\hat f_p^l(x_{jt},z_{jt})$ on the instruments and characteristics, as in the first stage of 2SLS, so that the LASSO regressors are uncorrelated with $\xi_{jt}$.

The tuning parameters $\lambda_{\tilde{\theta}}$ and $\lambda_{\beta}$ can be chosen by cross-validation in the simulation. 
Alternatively, the researcher may choose the theoretically suggested values
$\lambda_{\tilde{\theta}}=C_{\tilde{\theta}}\sqrt{d_0\log(d_{\tilde{z}}\vee T)/T}$ and $\lambda_{\beta}=C_{\beta}\sqrt{\log(d_x\vee T)/T}$, 
where $C_{\tilde{\theta}}$ and $C_{\beta}$ are constants.
Consequently, we have $\hat{f}_u^l(x_{jt})=x_{jt}'\hat{\beta}^l$.
It is worth noting that if $\alpha_0=0$ is assumed, then estimation of $f_{p0}$ becomes redundant, and Step B can be executed without $p_{jt}$. In this scenario, $y_{jt}(\sigma_0) = x_{jt}'\beta_0 + \xi_{jt}$, eliminating the endogeneity issue.

\subsection{Theory for the convergence rates of nuisance parameter estimators}
In this section, we provide regularity conditions that are sufficient for the desired convergence rate of the nuisance parameter estimators, $f_{z0}$ and $f_{u0}$.
First, we define the moment function for the estimation of $f_{u0}$.
\begin{align*}
    g(\theta)=\E_J \E\tilde{z}_{jt}(y_{jt}(\sigma)-\alpha p_{jt}-x_{jt}'\beta),\quad \tilde{z}_{jt}=(z_{jt}',x_{jt}')'\in\mathbb{R}^{d_{\tilde{z}}},\quad \theta=(\sigma,\alpha,\beta)
\end{align*}

\begin{condition}\label{cond:Regularity2}\leavevmode
    \begin{enumerate}[label=(\roman*), ref=\thecondition(\roman*)]
        \item Data $(w_t)_{t=1}^T$ obey the model (\ref{eq:highdimBLP_price}).\label{cond:Exogeneity2}
        \item The maximum eigenvalue of $\E x_{jt}x_{jt}'$ is bounded, i.e., $\lambda_{max}(\E x_{jt} x_{jt}')\leq M$.
            \label{cond:MaxEig}
    \end{enumerate}
\end{condition}
 
\begin{condition}[\textbf{Boundedness for Nuisance parameter estimators}]\label{cond:BoundednessNuisance}\leavevmode
    \begin{enumerate}[label=(\roman*), ref=\thecondition(\roman*)]
        \item $\forall$ $j,t$, random variables satisfy $|p_{jt}|$, $\norm{x_{jt}}_\infty$, $\normm{\epsilon^{z}_{jt}}_\infty$, $|\xi_{jt}|$, $|f_{u0}(x_{jt})|$, $\norm{f_{z0}(x_{jt})}_\infty, \ |f_{p0}(x_{jt},z_{jt})|\le M$ a.s. Parameters satisfy $\norm{\theta_{10}}_1$, $\normm{\beta_0}_1$, $\norm{\Pi_{0}}_1\leq M$. $d_z$ is fixed. $d_{\tilde z}=O(d_x)$.
            \label{cond:BoundednessVar2}
        \item Parameter spaces $\Theta_\alpha$ and $\Theta_\sigma$ are fixed compact sets. $\sigma_0$ and $\alpha_0$ lie in the interior of $\Theta_{\sigma}$ and $\Theta_{\alpha}$.
         \label{cond:compactnessNuisance}
    \end{enumerate}
\end{condition}

\begin{remark}
Conditions \ref{cond:MaxEig} and \ref{cond:BoundednessVar2} impose boundedness on the eigenvalues, parameters and random variables to ensure the convergence of the constructed estimator of $f_{u0}$.
One may replace the boundedness condition with conditions on moments, as stated in Conditions SE and SM in \citet{belloni2014inference}.
Here we choose the boundedness condition for simplicity.
This condition is used to replace Condition \ref{cond:BoundednessNO} for both the convergence of the preliminary estimator $\hat{f}_u$, and the asymptotic normality of the final estimator $\check{\theta}_1$.
\end{remark}

\begin{condition}[\textbf{Approximate sparsity}]\label{cond:ApproxSparsity}\leavevmode
    \begin{enumerate}[label=(\roman*), ref=\thecondition(\roman*)]
        \item \label{cond:Approx}Functions $f_{p0}$, $f_{z0}$, and $f_{u0}$ admit an approximately sparse form.
        Namely there exists $\beta_{px0}\in\mathbb{R}^{d_x}$, $\beta_{pz0}\in\mathbb{R}^{d_z}$, $\beta_{0}\in\mathbb{R}^{d_x}$, $\Pi_0\in\mathbb{R}^{d_z\times d_x}$, which depend on $T$, s.t. equations (\ref{eq:highdimBLP_approx_u}), (\ref{eq:highdimBLP_approx_z}) and (\ref{eq:highdimBLP_approx_p}) hold, where $r_{jt}^p$, $r_{jt}^u$ and $r_{jt}^z$ are the approximation errors.
        \item The sparsity index $d_0$ obeys $\sqrt{\frac{d_0\log(d_{x}\vee T)}{T}}=o(T^{-1/4})$.\label{cond:Sparsity}
    \end{enumerate}
\end{condition}

\begin{remark}
Condition \ref{cond:ApproxSparsity} is the key assumption ensuring that the $\ell_1$-penalized estimators $\hat{f}_z$, $\hat{f}_u$, and $\hat{f}_p$ converge at the desired rate.
Specifically, we require that the approximation errors vanish at the rate of $\sqrt{d_0/T}$, which matches the oracle convergence rate of the estimated coefficients that is achievable if the identities of the relevant controls were known a priori.
This rate ensures that approximation errors do not dominate the overall estimation error when employing LASSO to estimate $f_{z0}$ and $f_{p0}$.
Such conditions are standard in the high-dimensional econometrics literature (see, e.g., \citet{belloni2011high, belloni2014inference}).
\end{remark}

For a positive semi-definite matrix $M \in \R^{d\times d}$ and an integer $m \ge 1$, define the minimal and maximal $m$-sparse eigenvalues
\[
\phi_{\min}(m)[M] := \min_{\delta\neq0,\ \normm{\delta}_0\le m}\frac{\delta'M\delta}{\normm{\delta}_2^2},
\qquad
\phi_{\max}(m)[M] := \max_{\delta\neq0,\ \normm{\delta}_0\le m}\frac{\delta'M\delta}{\normm{\delta}_2^2}.
\]
For $l = 1,\dots,L$, let
\[
\hat\Sigma_l^{(1)} := \frac{1}{J\abs{I_l^c}}\sum_{j=1}^J\sum_{t\in I_l^c}\big(f_{p0}(x_{jt},z_{jt}), x_{jt}'\big)'\big(f_{p0}(x_{jt},z_{jt}), x_{jt}'\big),
\]
\[
\hat\Sigma_l^{(2)} := \frac{1}{J\abs{I_l^c}}\sum_{j=1}^J\sum_{t\in I_l^c}(z_{jt}', x_{jt}')'(z_{jt}', x_{jt}') .
\]
 
\begin{condition}[Sparse eigenvalues]\label{cond:SE}
There exist a sequence of integers $\ell_T \to \infty$ and constants $0 < \underline\phi \le \bar\phi < \infty$ such that, with probability approaching one, for all $l = 1,\dots,L$ and $k = 1, 2$,
\[
\underline\phi \le \phi_{\min}(\ell_T d_0)\big[\hat\Sigma_l^{(k)}\big] \le \phi_{\max}(\ell_T d_0)\big[\hat\Sigma_l^{(k)}\big] \le \bar\phi .
\]
\end{condition}
 
\begin{remark}\label{rem:SE}
Condition \ref{cond:SE} is the sparse eigenvalue condition (Condition SE) of \citet{belloni2014inference}; $\ell_T$ may grow arbitrarily slowly. Intuitively, it requires that every \emph{small} group of regressors, with at most $\ell_Td_0$ members, be well conditioned. The lower bound $\underline\phi$ rules out any small group that is (nearly) collinear, i.e.\ no regressor can be closely reproduced by a sparse linear combination of the others. The upper bound $\bar\phi$ rules out a small group of regressors that move together so strongly that their sum explodes. The condition says nothing about large groups. Hence it can hold when $d_x > T$, where the full Gram matrix is necessarily singular. It fails in three typical situations:
\begin{enumerate}[label=(\alph*)]
    \item Exact or near duplicates among $x_{jt}$, such as a full set of product dummies together with an intercept. Dropping a reference category restores the condition without changing $f_{u0}$.
    \item For $\hat\Sigma_l^{(1)}$: $f_{p0}(x_{jt},z_{jt})$ is a sparse linear combination of $x_{jt}$. This happens when the excluded instruments carry no information about price beyond $x_{jt}$. In that case $\alpha$ is not separately identified from $\beta$, so the failure of the condition reflects weak instruments rather than a technicality. Similarly, for $\hat\Sigma_l^{(2)}$ the condition fails if all instruments in $z_{jt}$ are sparse combinations of $x_{jt}$.
    \item The upper bound fails if many regressors share a strong common component. For instance, if each pair of regressors has the same correlation $\rho > 0$, $\phi_{\max}(m) = 1 + (m-1)\rho \to \infty$ as $m$ grows.
\end{enumerate}

Since $\phi_{\min}(m)$ is nonincreasing and $\phi_{\max}(m)$ is nondecreasing in $m$, the bounds hold for every $m \le \ell_T d_0$. They also hold for the submatrix $\frac{1}{J\abs{I_l^c}}\sum_{j,t}x_{jt}x_{jt}'$, whose sparse quadratic forms are those of $\hat\Sigma_l^{(k)}$ restricted to vectors that vanish on the first block.
A primitive sufficient condition is the following. Let $\Sigma^{(k)} := \E\,\E_J[\cdot]$ denote the population counterpart of $\hat\Sigma_l^{(k)}$, and suppose $2\underline\phi \le \phi_{\min}(m)[\Sigma^{(k)}] \le \phi_{\max}(m)[\Sigma^{(k)}] \le \bar\phi/2$ with $m = \ell_Td_0$, and $\ell_T d_0\sqrt{\log(d_x\vee T)/T} \to 0$. Then Condition \ref{cond:SE} holds w.p.a.1.
\footnote{
For any $\delta$ with $\normm{\delta}_0 \le m$, $\normm{\delta}_1^2 \le m\normm{\delta}_2^2$, so
\[
\abs{\delta'\big(\hat\Sigma_l^{(k)} - \Sigma^{(k)}\big)\delta} \le \max_{i,i'}\abs{\big(\hat\Sigma_l^{(k)} - \Sigma^{(k)}\big)_{ii'}}\,\normm{\delta}_1^2 \le m\max_{i,i'}\abs{\big(\hat\Sigma_l^{(k)} - \Sigma^{(k)}\big)_{ii'}}\,\normm{\delta}_2^2 .
\]
Under Condition \ref{cond:BoundednessNuisance}(i) and Lemma \ref{lemma:Hoeffding}, the maximum entry is $\max_{i,i'}\abs{\big(\hat\Sigma_l^{(k)} - \Sigma^{(k)}\big)_{ii'}}\lesssim \sqrt{\log(d_x\vee T)/T}$ w.p.a.1. Hence Condition \ref{cond:SE} holds whenever the population bounds hold for some $\ell_T \to \infty$ with $\ell_T d_0\sqrt{\log(d_x\vee T)/T} \to 0$. Since $d_0\sqrt{\log(d_x\vee T)/T} \to 0$, such a slowly growing $\ell_T$ can always be chosen.}
\end{remark}

The following theorem establishes the desired convergence rate of the preliminary estimators of the nuisance parameters under the suitable conditions.

\begin{theorem}
    \label{thm:Nuisance}
    If Conditions \ref{cond:Regularity1}, \ref{cond:IdentificationMoment}, \ref{cond:IdentificationGWG}, \ref{cond:Regularity2}, \ref{cond:BoundednessNuisance}, \ref{cond:ApproxSparsity}, \ref{cond:SE} 
    hold, and
    \[
    \lambda_{\Pi}=C_{\Pi}\sqrt{\frac{\log(d_x\vee T)}{T}},\quad
    \lambda_{\beta_{pz},\beta_{px}}=C_{\beta_{pz},\beta_{px}}\sqrt{\frac{\log(d_x\vee T)}{T}},\quad
    \lambda_{\tilde{\theta}}=C_{\tilde{\theta}}\sqrt{\frac{d_0\log(d_{\tilde{z}}\vee T)}{T}},\quad
    \lambda_{\beta}=C_{\beta}\sqrt{\frac{\log(d_x\vee T)}{T}}
    \]
    with sufficiently large constants, then the following holds for each $l=1,\dots,L$:
    \begin{align*}
     \normm{\hat{f}_z^l (x)-f_{z0}(x)}_{L^2}&\lesssim_P \sqrt{\frac{d_0\log(d_x\vee T)}{T}}=o(T^{-1/4}), \\
     \normm{\hat{f}_{u}^l(x)-f_{u0}(x)}_{L^2} &\lesssim_P \sqrt{\frac{d_0\log(d_x\vee T)}{T}}=o(T^{-1/4}).
    \end{align*}
\end{theorem}

\begin{remark}
Theorem \ref{thm:Nuisance} asserts that the preliminary estimators $\hat{f}_z$ and $\hat{f}_u$ converge at desired rates. 
This result plays the same role as Lemma 1 does in \citet{belloni2014inference} and is crucial for the Neyman orthogonal estimator to achieve the desired asymptotic normality.
\end{remark}

\subsection{Inference}
We now write the estimator in equation (\ref{eq:orthogonal_estimator}) explicitly for the two weighting matrices used in practice. Throughout, $\hat f_z^l$ and $\hat f_u^l$ are the LASSO-based estimators of Section~\ref{sec:prelim_estimation} constructed from the observations not in $I_l$. The first estimator uses the identity weighting matrix $\hat W=I_{d_z}$,
\begin{equation}\label{eq:theta_I}
\check{\theta}_{1,I}=(\check{\sigma}_I',\check{\alpha}_I)'
=\arg\min_{\theta_1\in\Theta_\sigma\times\Theta_\alpha}
\Big(\E_{T,L}\big[\psi_t(\theta_1,\hat f_z^l,\hat f_u^l)\big]\Big)'\, I_{d_z}\,
\Big(\E_{T,L}\big[\psi_t(\theta_1,\hat f_z^l,\hat f_u^l)\big]\Big).
\end{equation}
Since $\psi(w_t,\theta_{10},f_{z0},f_{u0})=\E_J\epsilon^z_{jt}\xi_{jt}$ has mean zero, $\Omega=\Var(\E_J\epsilon^z_{jt}\xi_{jt})$ is estimated by the plug-in estimator
\begin{equation}\label{eq:Omega_hat}
\hat{\Omega}
=\E_{T,L}\,\Big[
\psi_t(\check{\theta}_{1,I},\hat f_z^l,\hat f_u^l)\,
\psi_t(\check{\theta}_{1,I},\hat f_z^l,\hat f_u^l)'\Big]
\in\R^{d_z\times d_z}.
\end{equation}
The second estimator uses the efficient weighting matrix $\hat W=\hat{\Omega}^{-1}$,
\begin{equation}\label{eq:theta_opt}
\check{\theta}_{1,\mathrm{opt}}=(\check{\sigma}_{\mathrm{opt}}',\check{\alpha}_{\mathrm{opt}})'
=\arg\min_{\theta_1\in\Theta_\sigma\times\Theta_\alpha}
\Big(\E_{T,L}\big[\psi_t(\theta_1,\hat f_z^l,\hat f_u^l)\big]\Big)'\, \hat{\Omega}^{-1}\,
\Big(\E_{T,L}\big[\psi_t(\theta_1,\hat f_z^l,\hat f_u^l)\big]\Big).
\end{equation}
 
To estimate $G$, for $\sigma\in\Theta_\sigma$ let
\begin{equation}\label{eq:G_hat}
\hat G(\sigma)=\E_{T,L}
\E_J\left[
\begin{bmatrix}
\frac{\partial}{\partial\sigma}y_{jt}(\sigma)\\
-p_{jt}
\end{bmatrix}
\big(z_{jt}-\hat f_z^l(x_{jt})\big)'
\right]\in\R^{d_{\theta_1}\times d_z},
\end{equation}
which does not depend on $\alpha$. Write $\hat G_I=\hat G(\check\sigma_I)$ and $\hat G_{\mathrm{opt}}=\hat G(\check\sigma_{\mathrm{opt}})$. The asymptotic variances of the two estimators are $P_I=(GG')^{-1}G\Omega G'(GG')^{-1}$ and $P_{\mathrm{opt}}=(G\Omega^{-1}G')^{-1}$, i.e.\ $P$ in Theorem~\ref{thm:AsympNorm} evaluated at $W=I_{d_z}$ and $W=\Omega^{-1}$, respectively, and we estimate them by
\begin{equation*}
\hat P_I=(\hat G_I\hat G_I')^{-1}\hat G_I\hat{\Omega}\hat G_I'(\hat G_I\hat G_I')^{-1},
\qquad
\hat P_{\mathrm{opt}}=(\hat G_{\mathrm{opt}}\hat{\Omega}^{-1}\hat G_{\mathrm{opt}}')^{-1}.
\end{equation*}
Both matrices are symmetric and, by Condition~\ref{cond:IdentificationG}, positive definite with probability approaching one. Let $\hat P_I=\hat U_I\hat\Lambda_I\hat U_I'$ be the spectral decomposition of $\hat P_I$, where $\hat U_I$ is orthogonal and $\hat\Lambda_I$ is the diagonal matrix of eigenvalues, and define the symmetric square root and its inverse as
\begin{equation*}
\hat P_I^{1/2}=\hat U_I\hat\Lambda_I^{1/2}\hat U_I',
\qquad
\hat P_I^{-1/2}=\hat U_I\hat\Lambda_I^{-1/2}\hat U_I'.
\end{equation*}
$\hat P_{\mathrm{opt}}^{1/2}$ and $\hat P_{\mathrm{opt}}^{-1/2}$ are defined in the same way from $\hat P_{\mathrm{opt}}=\hat U_{\mathrm{opt}}\hat\Lambda_{\mathrm{opt}}\hat U_{\mathrm{opt}}'$.\footnote{In practice, one can instead use the Cholesky decomposition $\hat P=\hat K\hat K'$, and replace $\hat P^{-1/2}$ by $\hat K^{-1}$.}
 
\begin{corollary}\label{cor:inference}
Suppose Conditions \ref{cond:Regularity1}, \ref{cond:IdentificationMoment}, \ref{cond:IdentificationG}, \ref{cond:Regularity2}, \ref{cond:BoundednessNuisance}, \ref{cond:ApproxSparsity}, \ref{cond:SE} hold. Tuning parameters are chosen as in Theorem~\ref{thm:Nuisance}.
Additionally, assume that
\begin{align*}
    \Omega\to\Omega_*,\qquad G\to G_*,\qquad |r_{jt}^p|,|r_{jt}^u|,\|r_{jt}^z\|_2\lesssim \sqrt{d_0/T}.
\end{align*}    
Then $\hat{\Omega}-\Omega\xrightarrow{p}0$, $\hat P_I-P_I\xrightarrow{p}0$, $\hat P_{\mathrm{opt}}-P_{\mathrm{opt}}\xrightarrow{p}0$, and the estimators in \eqref{eq:theta_I} and \eqref{eq:theta_opt} are asymptotically normally distributed, i.e.,
\begin{align*}
\hat P_I^{-1/2}\sqrt{T}\big(\check{\theta}_{1,I}-\theta_{10}\big)&\xrightarrow{d}N(0,I_{d_{\theta_1}}),\\
\hat P_{\mathrm{opt}}^{-1/2}\sqrt{T}\big(\check{\theta}_{1,\mathrm{opt}}-\theta_{10}\big)&\xrightarrow{d}N(0,I_{d_{\theta_1}}).
\end{align*}
\end{corollary}

\section{Monte Carlo Simulation}
The purpose of this section is to provide a comprehensive examination of the finite-sample performance of the proposed Neyman orthogonal estimator. 
We achieve this through a Monte Carlo simulation, which is designed to evaluate key statistical properties, such as bias, variance, 
and rejection rates. 
We emphasize that while our theoretical results establish the estimator's asymptotic properties, these simulations are crucial for validating its practical utility in the realistic setting of limited sample sizes. 
The data-generating process is based on the core model outlined in Section 2 and the parameter settings are inspired by the Monte Carlo studies in \citet{belloni2014high}.

\subsection{Monte Carlo examples}
The data generating process is 
\begin{align*}
    y_{jt}(\sigma_0)&= \alpha_0 p_{jt} + x_{jt}'\beta_0 + \xi_{jt}, &\E[\xi_{jt}\, x_{jt}]&=0,\\
    z_{jt}&= \Pi_0 x_{jt}+\epsilon_{jt}^z, &\E[\epsilon_{jt}^z x_{jt}']&=0,\\
    p_{jt}&= x_{jt}'\beta_{px0} + z_{jt}'\beta_{pz0} + \epsilon_{jt}^p, &\E[\epsilon_{jt}^p (z_{jt}',x_{jt}')]&=0.
\end{align*}
We set $J=4$, $T=50$, $d_x=300$ and $d_z=12$.
The parameters of interest are $\theta_{10}=(\sigma_0,\alpha_0)=(1,-1)$, and the random coefficient on price is $\alpha_i\sim N(\alpha_0,\sigma_0^2)$.

The first element of $x_{jt}$ is $x_{jt,1}=1$. For $k\geq 2$, $x_{jt,k}=\sqrt{\rho}\,u_{t,k}+\sqrt{1-\rho}\,v_{jt,k}$ with $\rho=0.9999$, where $u_{t,k}$ and $v_{jt,k}$ are independent $\sqrt{3}\cdot U[-1,1]$ variables. Each control thus has unit variance, controls are independent across $k$, and $\Corr(x_{jt,k},x_{j't,k})=\rho$ for $j\neq j'$.
For $k=6,\dots,15$ we replace $x_{jt,k}$ by $\sqrt{0.9}\,x_{jt,5}+\sqrt{0.1}\,x_{jt,k}$, so these ten controls are highly correlated proxies for $x_{jt,5}$.
The nuisance parameters are sparse with five active elements:
\begin{align*}
    \beta_0&=(2,\,1,\,0.6,\,0.4,\,0.3,\,0,\dots,0)', \qquad
    \beta_{px0}=(1.2,\,0,\,0,\,0,\,0.5,\,0,\dots,0)',\\
    \beta_{pz0}&=\big(\bar{\beta}_{pz}',\underbrace{0,\dots,0}_{8}\big)', \qquad \bar{\beta}_{pz}=(0.95,\,0.70,\,0.60,\,0.45)',\\
    \Pi_{0,i}&=\big(\underbrace{0,\dots,0}_{k_i},\,0.06,\,0.045,\,0.03,\,0.02,\,0.01,\,0,\dots,0\big), \quad k_i=i-4\lfloor (i-1)/4\rfloor,\quad i=1,\dots,12,
\end{align*}
where $\Pi_{0,i}$ is the $i$-th row of $\Pi_0$.

Let $\xi_{jt}$ and $\eta_{jt,k}$, $k=2,3,4$, be independent $U[-1,1]$ variables, and let $q_t$ and $\nu_{jt}$ be independent $\sqrt{3}\cdot U[-1,1]$ variables, all independent of $x_{jt}$.
Define four primitive excluded shocks with mean zero, unit variance and mutual orthogonality,
\begin{align*}
    e_{jt}=\Big(q_t,\ \sqrt{\tfrac{3J}{J-1}}\big(\eta_{jt,2}-\bar{\eta}_{t,2}\big),\ \sqrt{3}\,x_{jt,2}\,\eta_{jt,3},\ \sqrt{3}\,\eta_{jt,4}\Big)',
\end{align*}
where $\bar{\eta}_{t,2}=J^{-1}\sum_{j'=1}^{J}\eta_{j't,2}$ is the within-market average, so that the second component sums to zero within each market, and let $\kappa_{jt}=\bar{\beta}_{pz}'e_{jt}$.
The instrument disturbance stacks the primitive shocks and their interactions with $\kappa_{jt}$ and $x_{jt,2}$:
\begin{align*}
    \epsilon^z_{jt}=\Big(e_{jt}',\ \big(e_{jt,k}\,\kappa_{jt}-\bar{\beta}_{pz,k}\big)_{k=1}^{4},\ \big(e_{jt,k}\,x_{jt,2}\big)_{k=1}^{4}\Big)',
\end{align*}
where the centering uses $\E[e_{jt,k}\kappa_{jt}]=\bar{\beta}_{pz,k}$. The price residual is
\begin{align*}
    \epsilon^p_{jt}=\xi_{jt}+\tfrac{1}{5}(x_{jt,2}^2-1)+\tfrac{2}{5}(x_{jt,3}^2-1)+\tfrac{1}{2}\nu_{jt}.
\end{align*}
By the symmetry of the uniform distributions, $\epsilon^p_{jt}$ is uncorrelated with $(z_{jt}',x_{jt}')'$. Hence $\beta_{px0}$ and $\beta_{pz0}$ coincide with the linear projection coefficients of $p_{jt}$ on $(z_{jt}',x_{jt}')'$, while $\xi_{jt}$ makes price endogenous.
Market shares are computed by Gauss--Hermite quadrature of order $9$.
Define $\tilde{z}_{jt}=(z_{jt}',x_{jt}')'$.
Tuning parameters are chosen by cross-validation.
For cross-fitting, the $50$ markets are split into $L=10$ folds of $5$ markets each.

We compare six estimators of $\theta_{10}$: four feasible procedures and two infeasible oracles that use the true nuisance functions.
Since $d_x>T$, the GMM estimator of Section~2 cannot be computed directly, so the first feasible procedure combines it with variable selection.
In this design the nuisance functions are linear, $f_{z0}(x)=\Pi_0x$ and $f_{u0}(x)=x'\beta_0$, and the orthogonal procedures use the cross-fitted estimators $\hat f_z^l(x)=\hat\Pi^lx$ and $\hat f_u^l(x)=x'\hat\beta^l$ of Section~\ref{sec:prelim_estimation}.
Oracle 1 is more precise than Oracle 2 because it also exploits the variation in the instruments explained by $x_{jt}$, which the orthogonal moment discards.

\begin{enumerate}
    \item Estimator 1 (GMM with optimal instruments): a selection-based version of the demand-side optimal-instrument GMM of \citet{reynaert2014improving}.
    Two steps tuned by cross-validation over markets give the inputs.
    A preliminary GMM with an $\ell_1$ penalty on the non-constant elements of $\beta$ gives $(\bar\sigma,\bar\alpha,\bar\beta)$ and a selected support $\hat S_1$.
    A Lasso regression of $p_{jt}$ on $(z_{jt}',x_{jt}')'$ gives fitted prices $\hat p_{jt}$ and a second support $\hat S_2$.
    With $\hat S=\hat S_1\cup\hat S_2$, the instruments are
    \[
        h_{jt}=\hat\omega^{-1}\Big(\partial_\sigma y_{jt}(\hat s_t,\bar\sigma),\ -\hat p_{jt},\ -x_{jt,\hat S}'\Big)',
    \]
    where $\hat s_t$ are the shares predicted at $(\bar\sigma,\bar\alpha,\bar\beta)$ with $p_{jt}$ replaced by $\hat p_{jt}$, and $\hat\omega$ is the sample second moment of the implied $\hat\xi_{jt}$.
    The estimator $(\sigma,\alpha,\beta_{\hat S})$ then solves the exactly identified moments $\E_{JT}\,h_{jt}\big(y_{jt}(\sigma)-\alpha p_{jt}-x_{jt,\hat S}'\beta_{\hat S}\big)=0$.
    \item Estimator 2 (Preliminary): the Step A estimator in \eqref{eq:theta_tilde}, computed on the full sample,
    \[
        \tilde{\theta}=(\tilde{\sigma},\tilde{\alpha},\tilde{\beta})
        =\arg\min_{\sigma\in\Theta_\sigma,\,\alpha\in\Theta_\alpha,\,\beta\in\R^{d_x}}
        \norm{\hat{g}(\theta)}_\infty+\lambda_{\tilde{\theta}}\normm{\beta}_1,
        \qquad
        \hat g(\theta)=\E_{JT}\,\tilde z_{jt}\big(y_{jt}(\sigma)-\alpha p_{jt}-x_{jt}'\beta\big);
    \]
    we report $\tilde{\theta}_1=(\tilde{\sigma},\tilde{\alpha})$.
    \item Estimator 3 (No Neyman orthogonality): plug $\tilde f_u(x)=x'\tilde\beta$ into the original, non-orthogonal moment function,
    \[
        \hat{\theta}_1=\arg\min_{\theta_1\in\Theta_\sigma\times\Theta_\alpha}
        \norm{\E_{JT}\,\tilde z_{jt}\big(y_{jt}(\sigma)-\alpha p_{jt}-\tilde f_u(x_{jt})\big)}_2 .
    \]
    \item Estimator 4 (Neyman orthogonal): $\check{\theta}_{1,I}$ in \eqref{eq:theta_I}.
    \item Oracle 1 (support): plug the true $f_{u0}$ into the original moment function, with instruments $z_{jt}$ and the subvector $\bar x_{jt}$ of characteristics that have nonzero coefficients in $\beta_0$,
    \[
        \hat{\theta}_{1,o1}=\arg\min_{\theta_1\in\Theta_\sigma\times\Theta_\alpha}
        \norm{\E_{JT}\begin{pmatrix} z_{jt}\\ \bar x_{jt}\end{pmatrix}\big(y_{jt}(\sigma)-\alpha p_{jt}-f_{u0}(x_{jt})\big)}_2 .
    \]
    \item Oracle 2 (shock): the infeasible counterpart of Estimator 4, which plugs the true $(f_{z0},f_{u0})$ into the orthogonal moment \eqref{eq:NOmoment}, so that the instruments are $z_{jt}-f_{z0}(x_{jt})=\epsilon^z_{jt}$,
    \[
        \hat{\theta}_{1,o2}=\arg\min_{\theta_1\in\Theta_\sigma\times\Theta_\alpha}
        \norm{\E_T\,\psi(w_t,\theta_1,f_{z0},f_{u0})}_2 .
    \]
\end{enumerate}

Tables~\ref{tab:simulation_sigma} and \ref{tab:simulation_alpha} report the bias, standard deviation, RMSE and 5\% rejection rate for $\sigma_0$ and $\alpha_0$.
The two oracle benchmarks are essentially unbiased, with rejection rates between 0.02 and 0.04, and Oracle 1 is more precise than Oracle 2, as expected.
Among the feasible procedures, the Neyman orthogonal estimator attains the smallest RMSE for both parameters (0.121 for $\sigma$ and 0.075 for $\alpha$), and its rejection rates (0.040 and 0.070) are close to the nominal level.
In contrast, the Preliminary and No Neyman orthogonal estimators have biases between 0.09 and 0.12 in absolute value and over-reject, with rejection rates of 0.520 and 0.090 for $\sigma$ and 0.790 and 0.670 for $\alpha$.
The GMM estimator with optimal instruments is nearly unbiased but has the largest dispersion, and its rejection rates (0.010 and 0.040) are at or below the nominal level.
These results show that Neyman orthogonalization reduces bias and gives valid inference when the nuisance parameter is high-dimensional.

\begin{table}[t]
\centering
\begin{threeparttable}
\caption{Simulation Results for $\sigma$}
\renewcommand{\arraystretch}{1.5}
\label{tab:simulation_sigma}
\begin{tabular}{@{\extracolsep{5pt}} lcccc}
\\[-1.8ex]\hline
\hline \\[-1.8ex]
& Bias & Standard deviation & RMSE & Rej.\ rate \\
\hline \\[-1.8ex]
GMM optimal IV
& $0.017$ & $0.186$ & $0.186$ & $0.010$ \\
Preliminary
& $-0.113$ & $0.110$ & $0.157$ & $0.520$ \\
No Neyman orthogonal
& $-0.102$ & $0.089$ & $0.135$ & $0.090$ \\
Neyman orthogonal
& $-0.021$ & $0.120$ & $0.121$ & $0.040$ \\
Oracle 1
& $0.003$ & $0.077$ & $0.077$ & $0.030$ \\
Oracle 2
& $-0.004$ & $0.093$ & $0.093$ & $0.020$ \\
\hline \\[-1.8ex]
\end{tabular}
\begin{tablenotes}
\small
\item Note: This table presents simulation results for $\sigma$ based on
100 Monte Carlo replications, with $T=50$, $J=4$, and $d_x=300$.
Bias is the mean of $\hat{\sigma}-\sigma_0$ across replications.
Standard deviation is the standard deviation of $\hat{\sigma}$ across
replications, and RMSE is the root mean squared error.
Rej.\ rate is based on standard errors calculated by bootstrap methods.
Oracle 1 uses the true support of the nuisance coefficients, while Oracle 2
uses the true residualized instruments and the true nuisance function $f_{u0}$.
Both oracle estimators reported in this table use the identity weighting matrix.
\end{tablenotes}
\end{threeparttable}
\end{table}


\begin{table}[t]
\centering
\begin{threeparttable}
\caption{Simulation Results for $\alpha$}
\renewcommand{\arraystretch}{1.5}
\label{tab:simulation_alpha}
\begin{tabular}{@{\extracolsep{5pt}} lcccc}
\\[-1.8ex]\hline
\hline \\[-1.8ex]
& Bias & Standard deviation & RMSE & Rej.\ rate \\
\hline \\[-1.8ex]
GMM optimal IV
& $-0.016$ & $0.140$ & $0.141$ & $0.040$ \\
Preliminary
& $0.093$ & $0.081$ & $0.123$ & $0.790$ \\
No Neyman orthogonal
& $0.116$ & $0.068$ & $0.135$ & $0.670$ \\
Neyman orthogonal
& $0.030$ & $0.069$ & $0.075$ & $0.070$ \\
Oracle 1
& $0.001$ & $0.032$ & $0.032$ & $0.040$ \\
Oracle 2
& $0.005$ & $0.051$ & $0.051$ & $0.040$ \\
\hline \\[-1.8ex]
\end{tabular}
\begin{tablenotes}
\small
\item Note: This table presents simulation results for $\alpha$ based on
100 Monte Carlo replications, with $T=50$, $J=4$, and $d_x=300$.
Bias is the mean of $\hat{\alpha}-\alpha_0$ across replications.
Standard deviation is the standard deviation of $\hat{\alpha}$ across
replications, and RMSE is the root mean squared error.
Rej.\ rate is based on standard errors calculated by bootstrap methods.
Oracle 1 uses the true support of the nuisance coefficients, while Oracle 2
uses the true residualized instruments and the true nuisance function $f_{u0}$.
Both oracle estimators reported in this table use the identity weighting matrix.
\end{tablenotes}
\end{threeparttable}
\end{table}

 
\section{Empirical Application: Demand for Canned Tuna}\label{sec:empirical}\label{sec:empirical-DFF}
\paragraph{Motivation.}
In this section, we apply our estimator to the Dominick's Finer Foods scanner data, focusing on canned tuna, a category widely used in the demand estimation literature. \citet{chevalier2003don} document that the price index for tuna falls during Lent, a period of high demand, and interpret this as evidence of loss-leader models of retailer competition. \citet{nevo2006does} attribute the decline largely to a shift toward cheaper brands. \citet{gandhi2023estimating} show that zero market shares are pervasive in these data and that their treatment materially affects estimated price elasticities.

\subsection{Data}\label{sec:emp-DFF-data}
The Dominick's Finer Foods data set comes from University of Chicago Graduate School of Business.
The public DFF database contains scanner data collected from a roughly 100-store supermarket
chain in the Chicago metropolitan area between 1989 and 1997. The replication input used here contains
86 stores and 300 weeks, corresponding approximately to September 1989 through June 1995.
In the cleaned replication input of \citet{gandhi2023estimating}, UPCs have already been combined into 49 tuna
product identifiers, denoted by \texttt{prod\_id}. For each store, week, and product identifier, the data
contain quantity sold, 
price,
package size, and an indicator for whether
the product was subject to promotional activity. The variable \texttt{profit} is the gross-margin
percentage, measured on a 0--100 scale. Thus, the implied wholesale acquisition cost can be recovered. \footnote{We follow \citet{gandhi2023estimating} and define the wholesale cost as
\[
  \text{wholesale cost}
  =
  \text{retail price}
  \left(1-\frac{\texttt{profit}}{100}\right).
\]
}
It also reports the number of customers who visited a store and completed a purchase.
We use this count as the measure of market size. 
The DFF database also supplies store-level demographic
variables constructed from neighbourhood characteristics around each store. These variables are
time invariant and are merged into the demand sample by store.

\subsubsection*{Markets, Market Shares, Choice Sets and High-Dimensional Controls}\label{sec:emp-processing}

We define a market $t$ as a store--four-week period, which yields 75 periods per store and 6,232 candidate markets. Quantities and revenues are summed over the weeks in a market, and the price of alternative $j$ is the quantity-weighted average price in dollars per package, adjusted for multipack bundles.\footnote{Specifically,
$p_{jt}=\frac{\sum_{w\in t}p_{jw}q_{jw}}{\sum_{w\in t}q_{jw}}$, where $q_{jw}$ is the quantity of alternative $j$ sold in week $w$.}
Market size is the number of customers who visited the store during the four-week period, computed from the store--week customer count \texttt{custcoun}. Following \citet{gandhi2023estimating}, the outside share is one minus the sum of inside shares.\footnote{
$M_t=\sum_{w\in t}\texttt{custcoun}_{w}.$ 
The inside and outside shares are
$s_{jt}=\frac{q_{jt}}{M_t},\ s_{0t}=1-\sum_{j=1}^{J}s_{jt}$.}

To keep every alternative's share strictly positive in every market, we retain the eight highest-volume product identifiers as standalone alternatives and aggregate the remaining ones into a composite alternative, \emph{Other tuna}, whose quantity is the sum and whose price is the quantity-weighted average of its constituents. This gives $J=9$ alternatives, with \emph{Other tuna} accounting for 37.40\% of tuna volume. Dropping markets in which any alternative has a zero or missing share or a nonpositive or missing price leaves $T=6{,}138$ markets and $55{,}242$ product--market observations.
Table~\ref{tab:tuna-coverage} summarizes the sample; a mean inside share of 6.5\% is plausible for a narrow category within a full-line grocery basket. Table~\ref{tab:tuna-products} describes the alternatives: seven of the eight standalone products have mean prices of about \$0.74--\$0.77 per package, whereas product 8 (\$1.44) is close in price to \emph{Other tuna} (\$1.52).

\begin{table}[t]
\centering
\caption{Summary Statistics of Market Coverage}\label{tab:tuna-coverage}
\begin{tabular}{lrrrrrr}
\toprule
 & Min. & 1st Qu. & Median & Mean & 3rd Qu. & Max. \\
\midrule
Inside share (\%)        & 1.815 & 4.545  & 5.586  & 6.511  & 7.224  & 44.302  \\
Market size (000 visits) & 4.368 & 55.975 & 69.229 & 71.498 & 87.688 & 138.061 \\
Tuna units sold (000)    & 0.300 & 2.938  & 3.980  & 4.638  & 5.431  & 37.667  \\
Markets per store        & 37    & 71     & 74     & 71.37  & 75     & 75      \\
\bottomrule
\end{tabular}

\vspace{0.5em}
\parbox{0.92\textwidth}{\footnotesize
\textit{Notes:} The first three rows report statistics across the 6,138 retained
store--four-week markets. The last row reports statistics across the 86 retained stores.
Inside share is total tuna units sold divided by market size, where market size is the four-week sum of weekly purchasing-store visit counts.
 Every retained market contains all nine alternatives.
}
\end{table}

The DFF data contain 45 store-level demographic variables.\footnote{Three pairs are exact duplicates in our sample, leaving 42 distinct demographic variables. Some stores have missing values for some of these variables. For each store, we replace missing values with the cross-store median and create an indicator for missingness.}
Interacting them with the product indicators and adding store, year and four-week-period fixed effects generates 456 candidate characteristics, of which 412 remain after removing exact collinearity.\footnote{The sample contains nine product groups, with ``Other tuna'' serving as the omitted reference category. Thus, 42 standardized demographic variables plus one missing-demographics indicator are interacted with eight product indicators, yielding \(43\times8=344\) interaction terms. The remaining candidate regressors consist of eight product indicators, 85 store fixed effects, five year fixed effects, 12 within-year four-week-period fixed effects (a ``year'' consists of 13 consecutive four-week periods (52 DFF weeks) starting from DFF week~1, rather than a calendar year), and two Lent-related controls, for a total of \(344+8+85+5+12+2=456\).}

\begin{table}[t]
\centering
\caption{Summary Statistics of Alternatives}\label{tab:tuna-products}
\begin{tabular}{clrrrr}
\toprule
Index & Product ID & Volume (\%) & Mean share (\%) &
Mean price (\$) & S.D. \ price (\$) \\
\midrule
1 & 50330      & 17.19 & 1.126 & 0.765 & 0.114 \\
2 & 50205      & 13.71 & 0.885 & 0.753 & 0.120 \\
3 & 50045      & 11.11 & 0.716 & 0.744 & 0.126 \\
4 & 50400      &  5.33 & 0.336 & 0.740 & 0.128 \\
5 & 50280      &  5.15 & 0.327 & 0.761 & 0.097 \\
6 & 50200      &  3.68 & 0.229 & 0.757 & 0.119 \\
7 & 50050      &  3.40 & 0.213 & 0.742 & 0.124 \\
8 & 50460      &  3.02 & 0.200 & 1.441 & 0.078 \\
9 & Other tuna & 37.40 & 2.479 & 1.519 & 0.184 \\
\bottomrule
\end{tabular}

\vspace{0.5em}
\parbox{0.92\textwidth}{\footnotesize
\textit{Notes:} Product ID is the product-group identifier constructed in the
original data-cleaning code as $\lfloor\texttt{nitem}/10\rfloor$ and is not a
full UPC. \emph{Other tuna} aggregates the remaining 41 product identifiers.
Volume is an alternative's share of total tuna units sold in the estimation
sample and sums to 100\% across alternatives. Mean share is the average number
of units sold per purchasing visit across the 6,138 store--four-week markets;
it is not a within-tuna volume share. Mean price and its standard deviation are
also computed across markets and refer to dollars per package after adjustment
for multipack bundles.
}
\end{table}

\subsection{Empirical Results}\label{sec:emp-results}
We adopt the lognormal specification, writing
$\alpha_i=-\exp(\sigma_{10}+\sigma_{20}\nu_i)$ with $\nu_i\sim N(0,1)$, so that $\alpha_0=0$ and
$\theta_{10}=(\sigma_{10},\sigma_{20})$. Utility is given by \eqref{eq:standardBLP_utility}, with
$x_{jt}$ the 412 controls described above. 
Following \citet[Footnote 15]{gandhi2023estimating}, we construct excluded instruments from the
wholesale cost $c_{jt}$,\footnote{The exclusion restriction rests on the finding that retail promotions are funded
largely through reductions in retailer margins rather than in wholesale costs
\citep{chevalier2003don,hosken2004patterns}. It fails if manufacturers set
wholesale prices in response to demand shocks; see also
\citet[Footnote 30]{gandhi2023estimating}.}
\begin{equation*}
    z_{jt}=\big(q_1(\varphi_{jt}),\,q_2(\varphi_{jt}),\,q_3(\varphi_{jt})\big)',\qquad
    \varphi_{jt}=\Phi\!\left(\frac{c_{jt}-\bar c}{s_c}\right),
\end{equation*}
where $\Phi(\cdot)$ is the standard normal distribution function and $q_k(\cdot)$ is a degree-$k$
polynomial, the three being orthonormal under the empirical distribution of $\varphi_{jt}$.\footnote{That is, $\frac{1}{N}\sum_{j,t} q_k(\varphi_{jt})=0$ and $\frac{1}{N}\sum_{j,t} q_k(\varphi_{jt})q_{k'}(\varphi_{jt})=\mathbb{I}\{k=k'\}$ for $k,k'=1,2,3$, with $\bar c=\frac{1}{N}\sum_{j,t}c_{jt}$ and $s_c=\big[\frac{1}{N-1}\sum_{j,t}(c_{jt}-\bar c)^2\big]^{1/2}$, where the sums run over all $N=JT=55{,}242$ product--market observations.}
The full instrument vector is $\tilde z_{jt}=(z_{jt}',1,x_{jt}')'$.
Table~\ref{tab:tuna-estimates} compares the Neyman orthogonal (NO) estimator with unpenalized GMM and optimal-IV GMM.
NO gives $\hat\sigma_1=2.034$ and $\hat\sigma_2=0.700$, whereas unpenalized GMM gives larger values,
$\hat\sigma_1=3.368$ and $\hat\sigma_2=1.011$, and optimal-IV GMM gives $\hat\sigma_1=3.552$ and $\hat\sigma_2=1.923$.



\begin{table}[t]
\centering
\caption{Demand Parameter Estimates}
\label{tab:tuna-estimates}
\begin{tabular}{lccc}
\toprule
 & NO estimator & GMM, penalty $=0$ & GMM, optimal IV \\
\midrule
$\hat\sigma_1$        & 2.034   & 3.368   & 3.552 \\
Std.\ error          & (0.967) & (0.746) & (0.090)    \\
$\hat\sigma_2$       & 0.700   & 1.011   & 1.923 \\
Std.\ error          & (0.435) & (0.488) & (0.069)  \\[2pt]
Step A penalty       & 0.06    & --      & --    \\
Alternatives         & \multicolumn{3}{c}{9}      \\
Observations         & \multicolumn{3}{c}{55,242} \\
Markets              & \multicolumn{3}{c}{6,138}  \\
\bottomrule
\end{tabular}

\vspace{0.5em}
\parbox{0.96\textwidth}{\footnotesize
\textit{Notes:} The table reports estimates of $\sigma_{10}$ and
$\sigma_{20}$ in
$\alpha_i=-\exp(\sigma_{10}+\sigma_{20}\nu_i)$.
NO standard errors use the cross-fitted orthogonal moments.
Unpenalized GMM uses all variables, including the constant. 
Optimal-IV GMM uses the union of characteristics selected by preliminary
penalized GMM and price LASSO.
Unpenalized GMM and optimal-IV GMM have standard errors computed using the sandwich formula, using a generalized inverse to avoid singularity of the Jacobian.
}
\end{table}

\paragraph{Own-Price Elasticities}

For each product--market observation, we calculate
\[
  \eta_{jjt}
  =
  \frac{\partial s_{jt}}{\partial p_{jt}}
  \frac{p_{jt}}{s_{jt}},
  \qquad
  \frac{\partial s_{jt}}{\partial p_{jt}}
  =
  \int
  \alpha_i s_{ijt}(1-s_{ijt})\,dF(\nu_i),
\]
where $s_{ijt}$ is consumer $i$'s choice probability for alternative $j$ in market $t$.
Table~\ref{tab:tuna-elasticities} summarizes their distribution.
Under NO, every elasticity lies between $-4.461$ and $-1.213$ (mean $-2.629$), so demand is
elastic throughout the sample without an extreme lower tail. Unpenalized GMM, consistent with its
larger price coefficient, produces larger and more dispersed elasticities (mean $-4.005$, standard deviation $1.186$), with a heavy lower tail reaching $-13.549$.
Optimal-IV GMM produces a mean elasticity of $-1.710$, with a lower tail reaching $-16.952$ and a standard deviation of $0.592$, demonstrating a similar pattern.



\begin{table}[t]
\centering
\caption{Summary Statistics of Own-Price Elasticities}
\label{tab:tuna-elasticities}
\begin{tabular}{lccc}
\toprule
Statistic & NO estimator & GMM, penalty $=0$ & GMM, optimal IV \\
\midrule
Mean            & $-2.629$ & $-4.005$  & $-1.710$  \\
Std.\ deviation & 0.375    & 1.186     & 0.592     \\
Minimum         & $-4.461$ & $-13.549$ & $-16.952$ \\
1st quartile    & $-2.819$ & $-4.663$  & $-2.052$  \\
Median          & $-2.556$ & $-3.684$  & $-1.675$  \\
3rd quartile    & $-2.374$ & $-3.169$  & $-1.322$  \\
Maximum         & $-1.213$ & $-0.820$  & $-0.436$  \\
\midrule
Alternatives    & \multicolumn{3}{c}{9}      \\
Observations    & \multicolumn{3}{c}{55,242} \\
\bottomrule
\end{tabular}

\vspace{0.5em}
\parbox{0.96\textwidth}{\footnotesize
\textit{Notes:} The table summarizes own-price elasticities for nine
alternatives in 6,138 store--four-week markets.
}
\end{table}

\paragraph{Held-Out Prediction}

We split the 6,138 markets into two equal folds, estimate parameters on one fold, predict shares
in the other, and then swap the folds. Table~\ref{tab:tuna-heldout} reports the pooled results.

NO predicts held-out shares substantially better than unpenalized GMM: its MSE is 13\% of the
unpenalized GMM MSE, and the correlation between predicted and observed shares is 0.58 rather than 0.21.
The optimal IV GMM estimator gives a relative MSE of 1.05 and a correlation of 0.18.
All estimators overpredict the mean share of 0.0072, but NO much less so (0.0102 against 0.0177 for unpenalized GMM and 0.0176 for the optimal IV GMM).



\begin{table}[t]
\centering
\small
\setlength{\tabcolsep}{3pt}
\caption{Held-Out Prediction of Market Shares}
\label{tab:tuna-heldout}
\begin{tabular}{lccccccccc}
\toprule
Estimator & MSE & RMSE & MAE & Bias & Max.\ $|e|$ &
Corr. & Mean $\hat s$ & Mean $s$ & Rel.\ MSE \\
\midrule
NO
 & 0.00045 & 0.02125 & 0.00827 & 0.00292 & 0.33408
 & 0.58447 & 0.01015 & 0.00723 & 0.12629 \\
GMM, penalty $=0$
 & 0.00358 & 0.05980 & 0.01718 & 0.01044 & 0.69608
 & 0.20526 & 0.01767 & 0.00723 & 1.00000 \\
GMM, optimal IV
 & 0.00374 & 0.06116 & 0.01726 & 0.01035 & 0.70039
 & 0.18261 & 0.01758 & 0.00723 & 1.04597 \\
\bottomrule
\end{tabular}

\vspace{0.5em}
\parbox{0.96\textwidth}{\footnotesize
\textit{Notes:} Statistics are pooled over 55,242 held-out
product--market observations. Shares, prediction errors, RMSE, MAE,
and bias are reported in share units; MSE is reported in squared
share units. Bias is the mean of $\hat s_{jt}-s_{jt}$.
Relative MSE is normalized by the MSE of unpenalized GMM.
Folds are constructed at the market level, although markets from
the same store may appear in different folds.
}
\end{table}

\section{Conclusion}

This paper develops inference for the BLP model when mean utility depends on high-dimensional product characteristics, possibly more than the number of markets.
We propose a Neyman orthogonal score whose only high-dimensional inputs are two regression functions, so that no high-dimensional matrix needs to be estimated, together with an $\ell_1$-penalized procedure that estimates these functions at the required rates even though one of them depends on the nonlinear parameters through the share inversion.
The resulting estimator of the parameters of interest is $\sqrt{T}$-asymptotically normal.
In simulations, it attains rejection rates close to the nominal level, whereas plug-in estimators that ignore orthogonality over-reject severely.
In an application to canned tuna from Dominick's Finer Foods, the estimated own-price elasticities indicate elastic demand throughout the sample without extreme values, whereas the GMM benchmarks yield implausibly large elasticities for some observations. The estimator also predicts held-out market shares with roughly one-eighth of the GMM benchmarks' mean squared error.
Natural extensions include high-dimensional random coefficients, a supply side, and the data-driven construction of approximately optimal instruments from many candidates.
\addcontentsline{toc}{chapter}{Appendices}
\appendix
\section{Appendix: Lemmas}

\begin{lemma}\label{lemma:Consistency}
    (Extension of Lemma 2.1 in \citet{newey1994large})
    For simplicity, we suppress the dependence of $\hat{\theta}$, $\theta_0$ and $\Theta$ on $T$.
    Let $\hat{\theta}$ be the maximizer of $\hat{Q}_T(\theta)$, and $Q_T(\theta)$ be a non-stochastic function.
    If (i) $Q_T(\theta)$ is uniquely maximized at $\theta_0\in \Theta$, where $\Theta$ is compact;
    (ii) $\forall$ $\epsilon>0$, there exists $\eta>0$ such that $\inf_{T,\normm{\theta-\theta_0}_2>\epsilon}Q_T(\theta_0)-Q_T(\theta)\geq\eta$ holds for any $T$; 
    (iii) $\sup_\Theta|\hat{Q}_T(\theta)-Q_T(\theta)|=o_p(1)$,
    then $\normm{\hat{\theta}-\theta_0}_2\overset{p}{\to}0$.
\end{lemma}
This lemma extends the consistency result of \citet{newey1994large} (Lemma 2.1) to a setting where the parameter space $\Theta$ and functions $Q_T$ and $\hat{Q}_T$ are allowed to depend on the sample size $T$.
Since $Q_T$ varies with $T$, we introduce condition (ii) to ensure consistency; this serves as a necessary replacement for the standard continuity condition on $Q$ required in Lemma 2.1 of \citet{newey1994large}.

\begin{lemma}\label{lemma:UniformConvergence}
    Let $\Theta\subset\R^{d_\theta}$ be compact, where $d_\theta$ and $\Theta$ do not depend on $T$.
    For each $T$, let $w_1,\dots,w_T$ be i.i.d., where the distribution of $w_t$ may depend on $T$.
    Suppose that $\sup_T\E d(w_t)^2<\infty$, and that there exists a constant $C<\infty$,
    independent of $T$, such that almost surely, for all $\theta,\theta'\in\Theta$,
    \begin{equation*}
        \norm{a(w_t,\theta)}_2\leq d(w_t),
        \qquad
        \norm{a(w_t,\theta)-a(w_t,\theta')}_2\leq C\norm{\theta-\theta'}_2 .
    \end{equation*}
    Then
    $\sup_{\theta\in\Theta}\norm{\E_T a(w_t,\theta)-\E a(w_t,\theta)}_2 = o_p(1)$.
\end{lemma}

\begin{proof}[Proof of Lemma \ref{lemma:UniformConvergence}]
    Write $\hat{A}(\theta):=\E_T a(w_t,\theta)$ and $A(\theta):=\E a(w_t,\theta)$.
    By the Lipschitz condition, both $\hat{A}$ and $A$ are $C$-Lipschitz on $\Theta$.
    Fix $\epsilon>0$ and let $\delta:=\epsilon/(3C)$. Since $\Theta$ is compact, it is covered by
    finitely many balls $\{\theta:\norm{\theta-\theta_k}_2\leq\delta\}$, $k=1,\dots,K$,
    where $K$ depends only on $\Theta$ and $\delta$, not on $T$.
    For $\theta$ in the $k$-th ball,
    \begin{equation*}
        \norm{\hat{A}(\theta)-A(\theta)}_2
        \leq \norm{\hat{A}(\theta)-\hat{A}(\theta_k)}_2
        + \norm{\hat{A}(\theta_k)-A(\theta_k)}_2
        + \norm{A(\theta_k)-A(\theta)}_2
        \leq \frac{2\epsilon}{3} + \max_{1\leq k\leq K}\norm{\hat{A}(\theta_k)-A(\theta_k)}_2 .
    \end{equation*}
    Since the $w_t$ are i.i.d. and $\norm{a(w_t,\theta_k)}_2\leq d(w_t)$,
    $\E\norm{\hat{A}(\theta_k)-A(\theta_k)}_2^2
    = T^{-1}\E\norm{a(w_t,\theta_k)-A(\theta_k)}_2^2
    \leq T^{-1}\E d(w_t)^2$.
    By Chebyshev's inequality and the union bound,
    \begin{equation*}
        P\Big(\sup_{\theta\in\Theta}\norm{\hat{A}(\theta)-A(\theta)}_2>\epsilon\Big)
        \leq \sum_{k=1}^{K} P\Big(\norm{\hat{A}(\theta_k)-A(\theta_k)}_2>\frac{\epsilon}{3}\Big)
        \leq \frac{9K\sup_T\E d(w_t)^2}{T\epsilon^2}\to 0 .
    \end{equation*}
\end{proof}

\begin{lemma}\label{lemma:ConditionalPopulation}
    Observations $(w_t)_{t=1}^T$ are i.i.d. across $t$. 
    Partition the sample indices \{1,\dots,T\} into $L$ groups, and denote the indices in the $l$-th group by $I_l$, and the number of observations in $I_l$ by $T_l$.
    Let $\hat{\eta}_l$ be an estimated function of $w_t$ using observations NOT in $I_l$, where $t\in I_l$. 
    If $\normm{a(w_t,\hat{\eta}_l)}_{L^2}\cdot \normm{b(w_t,\hat{\eta}_l)}_{L^2}=o_p(1)$, for $l=1,\dots,L$, where $a(w_t,\hat{\eta}_l)$ is a vector and $b(w_t,\hat{\eta}_l)$ is a scalar, then 
    \begin{equation}\label{eq:lemmaCP_eq1}
        \norm{\E_{T,L}a(w_t,\hat{\eta}_l)b(w_t,\hat{\eta}_l)}_2=\norm{\frac{1}{T}\sum_{l=1}^L\sum_{t\in I_l}a(w_t,\hat{\eta}_l)b(w_t,\hat{\eta}_l)}_2=o_p(1)
    \end{equation}
    If $\normm{\frac{1}{T_l}\sum_{t\in I_l}a\left(w_t,\hat{\eta}_l\right)}_{L^2}=o_p(1)$, for $l=1,\dots,L$, where $a(w_t,\hat{\eta}_l)$ is a vector, then
    \begin{equation}\label{eq:lemmaCP_eq2}
        \norm{\E_{T,L}a(w_t,\hat{\eta}_l)}_2=o_p(1)
    \end{equation}
\end{lemma}

The next lemma ensures uniform boundedness of the first and second derivatives of $y_{t}(\sigma)$ with respect to $\sigma$ under the specified regularity conditions.
\begin{lemma}\label{lemma:UniformBoundedness}
    Suppose Condition~\ref{cond:Regularity1} and Condition~\ref{cond:BoundednessNO} hold. Let $\sigma=(\sigma_1,\dots,\sigma_{d_\sigma})'$. Then $y_{t}(\sigma)$ is $C^{2}(\Theta_\sigma)$ a.s. and for each $i,j=1,\dots,d_\sigma$,
    \begin{align*}
        \sup_{\sigma\in\Theta_\sigma}\norm{y_t(\sigma)}_2<C<\infty,\; a.s.\qquad
        \sup_{\sigma\in\Theta_\sigma}\norm{\partial_{\sigma_i} y_{t}(\sigma)}_{2}<C<\infty,\; a.s.\qquad
        \sup_{\sigma\in\Theta_\sigma}\norm{\partial_{\sigma_i,\sigma_j}^{2}y_{t}(\sigma)}_{2}<C<\infty,\; a.s.
    \end{align*}
\end{lemma}

\begin{lemma}[Identification of $\theta_{10}$ under the $\ell_\infty$ criterion]
\label{lemma:ell-infty-id}
Let $\theta_1=(\sigma,\alpha)$.
Suppose Conditions~\ref{cond:Regularity1}, \ref{cond:IdentificationMoment}, \ref{cond:IdentificationGWG}, \ref{cond:Regularity2}, \ref{cond:BoundednessNuisance} and \ref{cond:ApproxSparsity} hold.
$\tilde z_{jt}$ contains components $(z_{jt}',x_{jt}')'$ with
$d_z$ fixed. Then there exists $C>0$, independent of $T$, such that for every $\theta\in
\Theta_\sigma\times\Theta_{\alpha}\times\{\beta\in\R^{d_x}:\|\beta\|_1\le M_{\beta}\}$,
\begin{align*}
    \|\theta_1-\theta_{10}\|_2\;\le\; C\Big(\|g(\theta)\|_\infty\vee \sqrt{d_0/T}\Big).
\end{align*}
\end{lemma}

The following are Hoeffding's inequality and the results from LASSO theory.

\begin{lemma}[corollary of Proposition 2.5 in \citet{wainwright2019high}]\label{lemma:Hoeffding}
    If $w_1,\dots,w_T\in\mathbb{R}^m$ are independent and $\norm{w_i}_\infty\leq M$, then
\[
    P\left(\norm{\frac{1}{T}\sum_{i=1}^T(w_i-\E w_i)}_\infty\geq t\right)\leq 2m\exp\left(-\frac{Tt^2}{2M^2}\right)
\]
In particular if $t=CT^{-1/2}\sqrt{\log(m\vee T)}$, then
\[
    P\left(\norm{\frac{1}{T}\sum_{i=1}^T(w_i-\E w_i)}_\infty\geq CT^{-1/2}\sqrt{\log(m\vee T)} \right)\leq 2m(\frac{1}{m\vee T})^{\frac{C^2}{2M^2}}
\]
When $C$ is large enough, for example $C>\sqrt{2M^2}$, we have
\[
    2m(\frac{1}{m\vee T})^{\frac{C^2}{2M^2}}=o(1)
\]
\end{lemma}
Lemma \ref{lemma:Hoeffding} is a direct corollary of Proposition 2.5 in \citet{wainwright2019high}.
It provides a tail bound for the deviation of the sample mean of vectors from its expectation in the infinity norm.
Given lemma \ref{lemma:Hoeffding}, the following two lemmas can be proved.
\begin{lemma}\label{lemma:empbound}
    Suppose $\hat{g}(\theta_0)=\E_{JT}\tilde{z}_{jt}(y_{jt}(\sigma_0)-\alpha_0p_{jt}-x_{jt}'\beta_0)=\E_{JT}\tilde{z}_{jt}(\xi_{jt}+r_{jt}^u)$, where $\tilde{z}_{jt}\in\R^{d_{\tilde{z}}}$ and $\xi_{jt}\in\R$. If $\normm{\tilde{z}_{jt}}_\infty$ and $\normm{\xi_{jt}}_\infty\leq M$, then 
    \begin{align*}
        P\left(\norm{\hat{g}(\theta_0)}_\infty\geq C_1T^{-1/2}\sqrt{d_0\log(d_{\tilde{z}}\vee T)}\right)&\leq 2d_{\tilde{z}}(\frac{1}{d_{\tilde{z}}\vee T})^{\frac{C_1^2d_0}{8M^4}} + \frac{2C}{C_1\sqrt{\log(d_{\tilde{z}}\vee T)}}\\
        &=o(1) \text{ when } C_1 \text{ is large enough}
    \end{align*}
    
\end{lemma}

\begin{lemma}\label{lemma:supdiff}
    Let $\tilde{z}_{jt}\in\R^{d_{\tilde{z}}}$ and $M_\sigma = \sup\{\|\sigma\|_2:\sigma\in\Theta_\sigma\}$.
    Suppose Condition \ref{cond:Regularity2} and \ref{cond:BoundednessNuisance} hold, then
    \begin{align*}
        P\left(\sup_{\sigma\in\Theta_\sigma}\norm{(\E_{JT}-\E)\tilde{z}_{jt}y_{jt}(\sigma)}_\infty\geq C_2 T^{-1/2}\sqrt{\log(d_{\tilde{z}}\vee T)}\right)&\leq 
        CM_{\sigma}^{d_\sigma}T^{d_\sigma/2}d_{\tilde{z}}(\frac{1}{d_{\tilde{z}}\vee T})^{\frac{C_2^2}{18M^4}}\\
        & = o(1) \text{ when } C_2 \text{ is large enough}\\
        P\left(\norm{(\E_{JT}-\E)\tilde{z}_{jt}\cdot (p_{jt},x_{jt}')}_\infty\geq C_3 T^{-1/2}\sqrt{\log(d_{\tilde{z}}\vee T)}\right)&\leq
        2d_{\tilde{z}}^2(\frac{1}{d_{\tilde{z}}\vee T})^{\frac{C_3^2}{2M^4}}\\
        & = o(1) \text{ when } C_3 \text{ is large enough }
    \end{align*}
\end{lemma}

We define a matrix $M$ to satisfy the restricted eigenvalue condition over a set $S \subset \{1,\dots,d\}$ with parameters $(\kappa, \nu)$ if $\Delta'M\Delta \ge \kappa\normm{\Delta}_2^2$ for all $\Delta \in \mc{C}_\nu(S) := \{\Delta \in \R^d : \normm{\Delta_{S^c}}_1 \le \nu\normm{\Delta_S}_1\}$.
By Lemma \ref{lem:SE_RE} below, Condition \ref{cond:SE} implies the restricted eigenvalue conditions used for the LASSO estimators of $f_{p0}$, $f_{z0}$ and $f_{u0}$.
\begin{lemma}[Sparse eigenvalues imply restricted eigenvalues]\label{lem:SE_RE}
Let $M \in \R^{d\times d}$ be positive semi-definite, $S \subset \{1,\dots,d\}$ with $\abs{S} \le d_0$, and $\nu > 0$. Suppose that for some integer $m \ge 1$,
\[
\phi_{\min}(d_0 + m)[M] \ge \underline\phi > 0,
\qquad
\phi_{\max}(m)[M] \le \bar\phi,
\qquad
m \ge 4\nu^2\frac{\bar\phi}{\underline\phi}\,d_0 .
\]
Then $M$ satisfies the restricted eigenvalue condition over $S$ with parameters $(\underline\phi/9, \nu)$, i.e.\ $\Delta'M\Delta \ge \frac{\underline\phi}{9}\normm{\Delta}_2^2$ for all $\Delta \in \mc{C}_\nu(S)$.
\end{lemma}

Define $S_1 = \{1\} \cup \{j+1 : \beta_{0,j} \neq 0\}$ with $\abs{S_1} \le d_0$ \footnote{Rigorously, the bound is $\abs{S_1} \le d_0 + 1$. Since $d_0+1\lesssim d_0$, we assume the bound is $d_0$ for simplicity.}.
\begin{lemma}\label{lem:SE_RE2}
Suppose Condition \ref{cond:SE} holds, and fix $\nu_1, \nu_2 > 1$. Then, with probability approaching one, for all $l = 1,\dots,L$:
\begin{enumerate}[label=(\roman*)]
    \item $\hat\Sigma_l^{(1)}$ and $\hat\Sigma_l^{(2)}$ satisfy $\phi_{\min}(m + (d_z+1)d_0) \ge \underline\phi$ and $\phi_{\max}(m) \le \bar\phi$, with $m := (\ell_T - d_z - 1)d_0$;
    \item $\hat\Sigma_l^{(1)}$ satisfies the restricted eigenvalue condition over $S_1$ with parameters $(\kappa_1, \nu_1) = (\underline\phi/9, \nu_1)$;
    \item $\hat\Sigma_l^{(2)}$ satisfies the restricted eigenvalue condition with parameters $(\kappa_2, \nu_2) = (\underline\phi/9, \nu_2)$ over every $S \subset \{1,\dots,d_z+d_x\}$ with $\abs{S} \le d_0$.
\end{enumerate}
\end{lemma}

The following lemma gives the convergence rates of LASSO with approximation error.
\begin{lemma}[LASSO with approximation error]\label{lem:lasso_approx}
Consider the model $Y = X\beta_0 + a + e$, where $Y, a, e \in \R^n$, $X \in \R^{n\times p}$ and $\beta_0 \in \R^p$. Let $c_a := \normm{a}_n$. Let $\mc U \subset \{1,\dots,p\}$ be a (possibly empty) set of unpenalized coordinates, $\mc P := \{1,\dots,p\}\setminus \mc U$, and let $S \supseteq \{j : \beta_{0,j}\neq 0\}\cup \mc U$ with $|S| \le d_0$. Define the tuning parameter $\lambda > 0$ and the estimator
\[
\hat\beta \in \arg\min_{\beta\in\R^p}\Big\{\frac{1}{n}\normm{Y - X\beta}_2^2 + \lambda\normm{\beta_P}_1\Big\},
\qquad \hat\Delta := \hat\beta - \beta_0 .
\]
Suppose $\frac{1}{n}X'X$ satisfies the restricted eigenvalue condition over $S$ with parameters $(\kappa, \nu)$, where $\kappa > 0$ and $\nu > 1$. Fix any $\nu_0 \in (1, \nu)$ (e.g.\ $\nu_0 = (1+\nu)/2$). Then under the event
\[
\mc{E} := \Big\{\lambda \ge 2\,\frac{\nu_0+1}{\nu_0-1}\,\norm{\frac{X'e}{n}}_\infty\Big\}
\]
we have
\begin{enumerate}[label=(\roman*)]
    \item (prediction norm)
    \[
    \normm{X\hat\Delta}_n \le 2c_a + \frac{2\nu_0}{\nu_0+1}\cdot\frac{\sqrt{d_0}\,\lambda}{\sqrt{\kappa}}\lesssim c_a + \sqrt{d_0}\lambda;
    \]
    \item ($\ell_1$ norm)
    \[
    \normm{\hat\Delta}_1 \le (1+\nu)\Big(\frac{2\sqrt{d_0}\,c_a}{\sqrt{\kappa}} + \frac{2\nu_0}{\nu_0+1}\cdot\frac{d_0\lambda}{\kappa}\Big) + K_{\nu,\nu_0}\,\frac{c_a^2}{\lambda}\lesssim \sqrt{d_0}\,c_a + d_0\lambda + \frac{c_a^2}{\lambda};
    \]
    \item ($\ell_2$ norm)
    \[
    \normm{\hat\Delta}_2 \le \frac{2c_a}{\sqrt{\kappa}} + \frac{2\nu_0}{\nu_0+1}\cdot\frac{\sqrt{d_0}\,\lambda}{\kappa} + K_{\nu,\nu_0}\,\frac{c_a^2}{\lambda}\lesssim c_a + \sqrt{d_0}\lambda + \frac{c_a^2}{\lambda},
    \]
\end{enumerate}
where $K_{\nu,\nu_0} := \frac{(\nu_0+1)(\nu+1)}{2(\nu-\nu_0)}$
\end{lemma}
\begin{remark}
    Note that the vector $a$ may be correlated with $X$ and $e$.
\end{remark}

The following lemma borrows the idea from \citet[Theorem 3]{belloni2013least}.
\begin{lemma}[Sparsity and $\ell_2$ rate under sparse eigenvalues]\label{lem:lasso_sparse}
Suppose the conditions of Lemma \ref{lem:lasso_approx} hold, and write (with the sparse eigenvalues defined before Condition \ref{cond:SE}) $\phi_{\min}(m) := \phi_{\min}(m)[\frac1nX'X]$ and $\phi_{\max}(m) := \phi_{\max}(m)[\frac1nX'X]$. Let $\hat M := \{j \in \mc P : \hat\beta_j \neq 0\}$, $\hat m := \abs{\hat M}$ and
\[
B := (\nu_0+1)^2\Big(\frac{3c_a}{\lambda} + \frac{2\nu_0}{\nu_0+1}\cdot\frac{\sqrt{d_0}}{\sqrt{\kappa}}\Big)^2 \lesssim \frac{c_a^2}{\lambda^2} + d_0 .
\]
Suppose that for some integer $m \ge 1$ and some constant $\underline\phi > 0$,
\[
m > 2B\phi_{\max}(m)
\qquad\text{and}\qquad
\phi_{\min}(m + d_0) \ge \underline\phi .
\]
Then under the event $\mc{E}$,
\begin{enumerate}[label=(\roman*)]
    \item (sparsity) $\hat m \le B\phi_{\max}(m)$ and $\normm{\hat\Delta}_0 \le d_0 + m$;
    \item ($\ell_2$ norm)
    \[
    \normm{\hat\Delta}_2 \le \frac{\normm{X\hat\Delta}_n}{\sqrt{\underline\phi}} \le \frac{1}{\sqrt{\underline\phi}}\Big(2c_a + \frac{2\nu_0}{\nu_0+1}\cdot\frac{\sqrt{d_0}\,\lambda}{\sqrt{\kappa}}\Big) \lesssim c_a + \sqrt{d_0}\lambda ,
    \]
    where $\lesssim$ hides constants depending only on $(\kappa, \nu_0, \underline\phi)$.
\end{enumerate}
\end{lemma}

\begin{lemma}[LASSO with a generated regressor]\label{lem:lasso_endogeneity}
Let $Y, F_p, \hat F_p, R, e \in \R^n$ and $X \in \R^{n\times d_x}$, and set $\tilde X := (F_p, X)$ and $\hat X := (\hat F_p, X)$. Suppose $Y = F_p\alpha_0 + X\beta_0 + R + e$ and let $\gamma_0 := (\alpha_0, \beta_0')'$. Define
\[
\hat\gamma = (\hat\alpha, \hat\beta')' \in \arg\min_{\alpha,\beta}\Big\{\frac1n\normm{Y - \hat F_p\alpha - X\beta}_2^2 + \lambda_\beta\normm{\beta}_1\Big\}.
\]
Recall that $S_1 = \{1\} \cup \{j+1 : \beta_{0,j} \neq 0\}$ with $\abs{S_1} \le d_0$. Let $c_p := \normm{\hat F_p - F_p}_n$ and $\bar c_a := \abs{\alpha_0}c_p + \normm{R}_n$. Suppose that, for constants $\kappa_1 > 0$, $\nu_1 > 1$ and $0 < \underline\phi \le \bar\phi$ and an integer $m \ge 1$,
\begin{enumerate}[label=(\alph*)]
    \item $\frac1n\tilde X'\tilde X$ satisfies the restricted eigenvalue condition over $S_1$ with parameters $(\kappa_1, \nu_1)$;
    \item $\phi_{\min}(m + d_0)\big[\frac1n\tilde X'\tilde X\big] \ge \underline\phi$ and $\phi_{\max}(m)\big[\frac1n\tilde X'\tilde X\big] \le \bar\phi$, where $m > \frac92\bar\phi B_\beta$ and $B_\beta := (\nu_0+1)^2\big(\frac{3\bar c_a}{\lambda_\beta} + \frac{2\nu_0}{\nu_0+1}\cdot\frac{2\sqrt{d_0}}{\sqrt{\kappa_1}}\big)^2$
    \item $c_p \le \frac12\min\big(\sqrt{\kappa_1}, \sqrt{\underline\phi}\big)$.
\end{enumerate}
Fix $\nu_0 \in (1, \nu_1)$. Then under the event
\[
\mc{E}_\beta := \Big\{\lambda_\beta \ge 2\,\frac{\nu_0+1}{\nu_0-1}\,\norm{\frac{\hat X'e}{n}}_\infty\Big\}
\]
we have
\[
\normm{\hat\gamma - \gamma_0}_2 \le \frac{2}{\sqrt{\underline\phi}}\Big(2\bar c_a + \frac{2\nu_0}{\nu_0+1}\cdot\frac{2\sqrt{d_0}\,\lambda_\beta}{\sqrt{\kappa_1}}\Big) \lesssim \bar c_a + \sqrt{d_0}\lambda_\beta .
\]
\end{lemma}

\section{Appendix: Proofs of Lemmas}

\begin{proof}[Proof of Proposition~\ref{prop:equivalence}]
Since $z_{jt}=\Pi_0x_{jt}+\epsilon^z_{jt}$ with $\mathbb{E}[\epsilon^z_{jt}\mid x_{jt}]=0$, we have
$\mathbb{E}\,\mathbb{E}_J[z_{jt}x_{jt}']=\Pi_0\Sigma_{xx}$, and hence
$G_\beta=-\begin{pmatrix}\Pi_0\\ I\end{pmatrix}\Sigma_{xx}$. Let
\[
  S=[\,I,\,-\Pi_0\,],\qquad
  T=\begin{pmatrix} I & -\Pi_0\\ 0 & I\end{pmatrix},\qquad
  E_2=\begin{pmatrix} 0\\ I\end{pmatrix},
\]
so that $S$ is the first block row of $T$ and $TG_\beta=-E_2\Sigma_{xx}$. Set $\tilde V=TVT'$ and
partition $\tilde V^{-1}=\begin{pmatrix}P_{11}&P_{12}\\P_{21}&P_{22}\end{pmatrix}$ conformably with
$(d_z,d_x)$. Substituting $V^{-1}=T'\tilde V^{-1}T$ into the definition of $\mu_0$, and noting
that $\Sigma_{xx}$ cancels, gives
\[
  \mu_0=(TA)'\Big[\tilde V^{-1}-\tilde V^{-1}E_2\big(E_2'\tilde V^{-1}E_2\big)^{-1}E_2'\tilde V^{-1}\Big]T
       =(TA)'\begin{pmatrix}P_{11}-P_{12}P_{22}^{-1}P_{21}&0\\0&0\end{pmatrix}T .
\]

We now simplify the middle matrix. Since $V$ is positive definite and $T$ is nonsingular,
$\tilde V=TVT'$ is positive definite, and so is $\tilde V^{-1}$. Hence its principal submatrix
$P_{22}$ is positive definite and invertible. Partition $\tilde V$ conformably as
\[
  \tilde V=\begin{pmatrix}\tilde V_{11}&\tilde V_{12}\\ \tilde V_{21}&\tilde V_{22}\end{pmatrix}.
\]
The identity $\tilde V\tilde V^{-1}=I$ implies, from its first block row,
\[
  \tilde V_{11}P_{11}+\tilde V_{12}P_{21}=I,
  \qquad
  \tilde V_{11}P_{12}+\tilde V_{12}P_{22}=0 .
\]
The second equation gives $\tilde V_{12}=-\tilde V_{11}P_{12}P_{22}^{-1}$. Substituting this
into the first equation yields
\[
  \tilde V_{11}\big(P_{11}-P_{12}P_{22}^{-1}P_{21}\big)=I,
  \qquad\text{that is,}\qquad
  P_{11}-P_{12}P_{22}^{-1}P_{21}=\tilde V_{11}^{-1}.
\]
Moreover, because the first block row of $T$ is $S$, the upper-left block of $\tilde V=TVT'$ is
$\tilde V_{11}=SVS'$. Writing $D=(SVS')^{-1}$, we therefore have
\[
  \mu_0=(TA)'\begin{pmatrix}D&0\\0&0\end{pmatrix}T .
\]
Finally, the first block row of $TA$ is $SA$, so
$(TA)'=\big[(SA)',\,(E_2'A)'\big]$, while
$\begin{pmatrix}D&0\\0&0\end{pmatrix}T=\begin{pmatrix}DS\\0\end{pmatrix}$. Multiplying out gives
\[
  \mu_0=(SA)'\,D\,S=(SA)'(SVS')^{-1}S .
\]

Finally, $z_{jt}-\Pi_0x_{jt}=\epsilon^z_{jt}$ gives
$SA=\mathbb{E}\,\mathbb{E}_J\big[\epsilon^z_{jt}(\partial_{\sigma'}y_{jt}(\sigma_0),\,-p_{jt})\big]=G'$,
$S\,m(w_t,\theta_1,\beta)=\mathbb{E}_J\big[(z_{jt}-f_{z0}(x_{jt}))(y_{jt}(\sigma)-\alpha p_{jt}-x_{jt}'\beta)\big]
=\psi(w_t,\theta_1,f_{z0},f_u)$, and
$SVS'=\mathrm{Var}\big(\mathbb{E}_J\epsilon^z_{jt}\xi_{jt}\big)=\Omega$.
\end{proof}

\begin{proof}[Proof of Lemma \ref{lemma:Consistency}]
    By definition of $\hat{\theta}$, we have
    \begin{align*}
        0&\leq \hat{Q}_T(\hat{\theta})-\hat{Q}_T(\theta_0) \\
        & = \hat{Q}_T(\hat{\theta})-Q_T(\hat{\theta})+Q_T(\hat{\theta})-Q_T(\theta_0)+Q_T(\theta_0)-\hat{Q}_T(\theta_0)\\
        Q_T(\theta_0)-Q_T(\hat{\theta})&\leq 2\sup_\Theta|\hat{Q}_T(\theta)-Q_T(\theta)|=o_p(1)
    \end{align*}
    Thus
    \begin{align*}
        P(\normm{\hat{\theta}-\theta_0}_2>\epsilon)&\leq P(Q_T(\theta_0)-Q_T(\hat{\theta})\geq\eta)\\
        &=o(1)
    \end{align*}
    Consequently, $\normm{\hat{\theta}-\theta_0}_2\overset{p}{\to}0$.
\end{proof}

\begin{proof}[Proof of Lemma \ref{lemma:ConditionalPopulation}]
    Denote all observations in $I_l$ by $\mathcal{W}_l$, all observations not in $I_l$ by $\mathcal{W}_l^c$.
    Note that $\hat{\eta}_l$ is an estimated function using observations in $\mathcal{W}_l^c$, and $w_t$ is an observation in $\mathcal{W}_l$.
    Then we have $w_t$ is independent of $\hat{\eta}_l$ for $t\in I_l$, and thus $\normm{a(w_t,\hat{\eta}_l)}_{L^2}$ and $\normm{b(w_t,\hat{\eta}_l)}_{L^2}$ are 
    functions of $\mathcal{W}_l^c$, since by notation $\normm{a(w_t,\hat{\eta}_l)}_{L^2}:= \left(\int \normm{a(w,\hat{\eta}_l(w))}_2^2 F(dw)\right)^{1/2}$.

    Given fixed $L$, we only need to prove 
    \[
        \norm{\frac{1}{T_l}\sum_{t\in I_l}a(w_t,\hat{\eta}_l)b(w_t,\hat{\eta}_l)}_2=o_p(1)
    \] 
    and 
    \[
        \norm{\frac{1}{T_l}\sum_{t\in I_l}a(w_t,\hat{\eta}_l)}_2=o_p(1)
    \]
    For any $\epsilon>0$, define the event set 
    \[
        K_\epsilon=\{\mathcal{W}_l^c|\normm{a(w_t,\hat{\eta}_l)}_{L^2}\cdot \normm{b(w_t,\hat{\eta}_l)}_{L^2}>\epsilon\}.
    \] 
    By the condition in the lemma, $P(K_\epsilon)=o(1)$, as $T \to \infty$.
    Then
    \begin{align*}
        &P(\normm{\frac{1}{T_l}\sum_{t\in I_l}a(w_t,\hat{\eta}_l)b(w_t,\hat{\eta}_l)}_2>M)\\
        \leq &P(\normm{\frac{1}{T_l}\sum_{t\in I_l}a(w_t,\hat{\eta}_l)b(w_t,\hat{\eta}_l)}_2>M, K_\epsilon^c)+P(K_\epsilon)\\
        =& \int_{\mathcal{W}_l^c\in K_\epsilon^c}P(\normm{\frac{1}{T_l}\sum_{t\in I_l}a(w_t,\hat{\eta}_l)b(w_t,\hat{\eta}_l)}_2>M, \mathcal{W}_l^c)+P(K_\epsilon)\\
        =&\int_{\mathcal{W}_l^c\in K_\epsilon^c}P(\normm{\frac{1}{T_l}\sum_{t\in I_l}a(w_t,\hat{\eta}_l)b(w_t,\hat{\eta}_l)}_2>M|\mathcal{W}_l^c)P(\mathcal{W}_l^c)+o(1)\\
        \leq &\int_{\mathcal{W}_l^c\in K_\epsilon^c}\frac{\E[\normm{\frac{1}{T_l}\sum_{t\in I_l}a(w_t,\hat{\eta}_l)b(w_t,\hat{\eta}_l)}_2|\mathcal{W}_l^c]}{M}P(\mathcal{W}_l^c)+o(1)\\
        \leq &\int_{\mathcal{W}_l^c\in K_\epsilon^c}\frac{\E[\normm{a(w_t,\hat{\eta}_l)}_2^2|\mathcal{W}_l^c]^{1/2}\cdot \E[b(w_t,\hat{\eta}_l)^2|\mathcal{W}_l^c]^{1/2}}{M}P(\mathcal{W}_l^c)+o(1)\\
        = & \int_{\mathcal{W}_l^c\in K_\epsilon^c}
        \frac{\left(\int \normm{a(w,\hat{\eta}_l(w))}_2^2 F(dw)\right)^{1/2}\cdot \left(\int \normm{b(w,\hat{\eta}_l(w))}_2^2 F(dw)\right)^{1/2}}{M}P(\mathcal{W}_l^c)+o(1)\\
        = & \int_{\mathcal{W}_l^c\in K_\epsilon^c}
        \frac{\normm{a(w_t,\hat{\eta}_l)}_{L^2}\cdot \normm{b(w_t,\hat{\eta}_l)}_{L^2}}{M}P(\mathcal{W}_l^c)+o(1)\\
        \leq &\frac{\epsilon}{M}+o(1) \quad (\text{by definition of } K_\epsilon^c)\\
        = & o(1)
    \end{align*}
    The conclusion is straightforward with carefully chosen $\epsilon$.
    Similarly, we can prove $\normm{\frac{1}{T_l}\sum_{t\in I_l}a(w_t,\hat{\eta}_l)}_2=o_p(1)$.
    For any $\epsilon>0$, define the event set 
    \[
        K_\epsilon'=\{\mathcal{W}_l^c|\normm{\frac{1}{T_l}\sum_{t\in I_l}a(w_t,\hat{\eta}_l)}_{L^2}>\epsilon\}.
    \] 
    By the condition in the lemma, $P(K_\epsilon')=o(1)$, as $T \to \infty$.
    \begin{align*}
        & P(\normm{\frac{1}{T_l}\sum_{t\in I_l}a(w_t,\hat{\eta}_l)}_2>M)\\
        \leq & P(\normm{\frac{1}{T_l}\sum_{t\in I_l}a(w_t,\hat{\eta}_l)}_2>M,\mathcal{W}_l^c\in K_\epsilon'^c)+P(K_\epsilon')\\
        =&\int_{\mathcal{W}_l^c\in K_\epsilon'^c}P(\normm{\frac{1}{T_l}\sum_{t\in I_l}a(w_t,\hat{\eta}_l)}_2>M|\mathcal{W}_l^c)P(\mathcal{W}_l^c)+o(1)\\
        \leq &\int_{\mathcal{W}_l^c\in K_\epsilon'^c}\frac{\E\left\{\normm{\frac{1}{T_l}\sum_{t\in I_l}a\left(w_t,\hat{\eta}_l\right)}_2^2|\mathcal{W}_l^c\right\}^{1/2}}{M}P(\mathcal{W}_l^c)+o(1)\\
        = & \int_{\mathcal{W}_l^c\in K_\epsilon'^c}\frac{\left(\int\normm{\frac{1}{T_l}\sum_{t\in I_l}a(w_t,\hat{\eta}_l)}_2^2F(dw_{T_l})\right)^{1/2}}{M}P(\mathcal{W}_l^c)+o(1)\quad 
        \text{where }dw_{T_l}=dw_{t_1}\dots dw_{t_{T_l}} \text{ and } t_{i}\in I_l\\
        = & \int_{\mathcal{W}_l^c\in K_\epsilon'^c}\frac{\normm{\frac{1}{T_l}\sum_{t\in I_l}a(w_t,\hat{\eta}_l)}_{L^2}}{M}P(\mathcal{W}_l^c)+o(1)\\
        \leq &\frac{\epsilon}{M}+o(1) \quad (\text{by definition of } K_\epsilon^{\prime c})\\
        =& o(1)
    \end{align*}
    The conclusion is straightforward with carefully chosen $\epsilon$.
\end{proof}

\begin{proof}[Proof of Lemma~\ref{lemma:UniformBoundedness}]
We prove the result for the case where $\sigma$ is scalar. The proof for the vector case is similar.
By abuse of notation, we suppress the market index $t$ in the proof.
Let $\varepsilon$ collect all randomness whose distribution does \emph{not} move with
$\sigma$. Only $b$ has a $\sigma$-dependent distribution, $b\sim f(\cdot,\sigma)$. Write
\[
\E_{\sigma}[g]:=\E_{\varepsilon}\!\int g(b,\varepsilon)\,f(b,\sigma)\,\mathrm{d} b .
\]
Prices are $p=(p_1,\dots,p_J)'$, 
$y(\sigma)=(y_{1}(\sigma),\dots,y_{J}(\sigma))'$ is the imputed mean-utility vector, $y=y(\sigma_0)$ is the mean-utility vector, and the utility vector is $u=(u_1,\dots,u_J)=y(\sigma)+b\,p$. Choice probabilities conditional on $(b,\varepsilon)$ are a function of the utility vector $u$. Let $\pi=(\pi_1,\dots,\pi_J)'$ with outside choice probability $\pi_0$.
There is
\begin{equation}\label{eq:logit}
\pi_j(y(\sigma)+b\,p)=\frac{e^{u_j}}{1+\sum_{k=1}^{J}e^{u_k}},\qquad
\pi_0(y(\sigma)+b\,p)=\frac{1}{1+\sum_{k=1}^{J}e^{u_k}} .
\end{equation}

The defining identity, which holds for \emph{all} $\sigma$, is
\begin{equation}\label{eq:share}
\E_{\sigma}\!\big[\pi\big(y(\sigma)+b\,p\big)\big]=s .
\end{equation}
Define the first- and second-order scores
\[
\ell_{\sigma}(b,\sigma):=\partial_\sigma\log f(b,\sigma)=\frac{\partial_\sigma f}{f},
\qquad
q(b,\sigma):=\frac{\partial_\sigma^{2}f}{f}.
\]
Since
$\partial_\sigma\ell_{\sigma}=\ddfrac{\partial_\sigma^{2}f}{f}
-\Big(\ddfrac{\partial_\sigma f}{f}\Big)^{2}$, these satisfy
\begin{equation}\label{eq:qscore}
q=\partial_\sigma\ell_{\sigma}+\ell_{\sigma}^{2}.
\end{equation}

\noindent
 
\textbf{Step 1: Differentiating the expectation operator}

For $g(b,\sigma)$,
\begin{equation}\label{eq:firstrule}
\frac{\partial}{\partial\sigma}\,\E_{\sigma}[g]
=\int\big(\partial_\sigma g\big)f\,\mathrm{d} b+\int g\,\partial_\sigma f\,\mathrm{d} b
=\int\Big(\partial_\sigma g+g\,\tfrac{\partial_\sigma f}{f}\Big)f\,\mathrm{d} b
=\E_{\sigma}\!\big[\partial_\sigma g+g\,\ell_{\sigma}\big].
\end{equation}
Write $\mc{L} g:=\partial_\sigma g+g\,\ell_{\sigma}$, so that \eqref{eq:firstrule} reads
$\frac{d}{d\sigma}\E_{\sigma}[g]=\E_{\sigma}[\mc{L} g]$. Applying \eqref{eq:firstrule} to
$g\equiv 1$ and to $g\equiv\ell_{\sigma}$, and using $\int f\,\mathrm{d} b\equiv 1$ for every
$\sigma$, gives the two \emph{Centring identities}
\begin{equation}\label{eq:centering}
\E_{\sigma}[\ell_{\sigma}]=\int\partial_\sigma f\,\mathrm{d} b=0,
\qquad
\E_{\sigma}[q]=\int\partial_\sigma^{2}f\,\mathrm{d} b=0 .
\end{equation}
Iterating \eqref{eq:firstrule} gives
$\frac{\partial^{2}}{\partial\sigma^{2}}\E_{\sigma}[g]=\E_{\sigma}[\mc{L}(\mc{L} g)]$, where
\[
\mc{L}(\mc{L} g)
=\partial_\sigma\big(\partial_\sigma g+g\ell_{\sigma}\big)
+\big(\partial_\sigma g+g\ell_{\sigma}\big)\ell_{\sigma}
=\partial_\sigma^{2}g
+\underbrace{(\partial_\sigma g)\ell_{\sigma}+g\,\partial_\sigma\ell_{\sigma}}
_{\text{from }\partial_\sigma(g\ell_{\sigma})}
+(\partial_\sigma g)\ell_{\sigma}+g\,\ell_{\sigma}^{2},
\]
and collecting the last two $g$-terms via \eqref{eq:qscore},
\begin{equation}\label{eq:secondrule}
\frac{\partial^{2}}{\partial\sigma^{2}}\E_{\sigma}[g]
=\E_{\sigma}\!\big[\partial_\sigma^{2}g+2\,\ell_{\sigma}\,\partial_\sigma g+g\,q\big].
\end{equation}
%

\textbf{Step 2: Derivatives of the logit map}
 
Since $\partial u_k/\partial y_j=\delta_{kj}$, derivatives in $u$ and in $y$ coincide. From \eqref{eq:logit},
\begin{equation}\label{eq:H}
H_{ij}(u):=\frac{\partial \pi_i}{\partial u_j}(u)=\pi_i(\delta_{ij}-\pi_j),
\qquad
H(u)=\diag(\pi(u))-\pi(u)\pi'(u),
\end{equation}
This gives
\begin{equation}\label{eq:piDeriv}
    \partial_\sigma \pi (u)=\frac{\partial}{\partial u'}\pi(u)\,\partial_\sigma u=H(u)\,\partial_\sigma y(\sigma),
\end{equation}
From \eqref{eq:H}, the second derivative is
\begin{equation}\label{eq:T}
T_{ijk}(u):=\frac{\partial^{2}\pi_i}{\partial u_j\,\partial u_k}(u)
=\pi_i\big[\delta_{ij}\delta_{ik}-\delta_{ij}\pi_k-\delta_{ik}\pi_j
-\delta_{jk}\pi_j+2\pi_j\pi_k\big],
\end{equation}
which is symmetric in $(j,k)$. This gives
\begin{equation}\label{eq:piDeriv2}
    \partial_\sigma H_{ij}(u)=\sum_{k=1}^{J}T_{ijk}(u)\,\partial_\sigma y_k(\sigma),\quad
    [\partial_\sigma H\cdot \partial_\sigma y(\sigma)]_{i}=\sum_{j=1}^{J}\sum_{k=1}^{J}T_{ijk}(u)\,\partial_\sigma y_j(\sigma)\,\partial_\sigma y_k(\sigma).
\end{equation}
 
Define $T=(T_{ijk})$ as a $J\times J\times J$ tensor. For $v,w\in\R^{J}$, its
contraction in the last two slots is the \emph{$J$-vector} $T[v,w]$ with
entries
\begin{equation}\label{eq:Tdef}
T[v,w]_i:=\sum_{j=1}^{J}\sum_{k=1}^{J}T_{ijk}\,v_j\,w_k,
\qquad i=1,\dots,J .
\end{equation}

\noindent
Throughout, extend any $v\in\R^{J}$ to the outside good by $v_0:=0$ and write
\[
\bar v:=\sum_{j\ge 0}\pi_jv_j,\qquad
\Var_{\pi}(v):=\sum_{j\ge 0}\pi_j(v_j-\bar v)^{2},\qquad
\Cov_{\pi}(v,w):=\sum_{j\ge 0}\pi_j(v_j-\bar v)(w_j-\bar w).
\]
 
Observe that for every $v,w\in\R^{J}$ and every $i=1,\dots,J$,
\begin{equation}\label{eq:Tcontract}
T[v,w]_i=\pi_i\Big[(v_i-\bar v)(w_i-\bar w)-\Cov_{\pi}(v,w)\Big],
\qquad\text{so}\qquad
T[v,v]_i=\pi_i\Big[(v_i-\bar v)^{2}-\Var_{\pi}(v)\Big].
\end{equation}
 
This is because multiply \eqref{eq:T} term by term with $v_jw_k$, using
$\sum_j\delta_{ij}v_j=v_i$, $\sum_j\pi_jv_j=\bar v$, and also
$\sum_{j,k}\delta_{jk}\pi_jv_jw_k=\sum_j\pi_jv_jw_j$:
\[
T[v,w]_i
=\pi_i\Big[v_iw_i-v_i\bar w-w_i\bar v
-\textstyle\sum_j\pi_jv_jw_j+2\bar v\bar w\Big].
\]
Centring the first three terms, $v_iw_i-v_i\bar w-w_i\bar v
=(v_i-\bar v)(w_i-\bar w)-\bar v\bar w$, so the bracket equals
$(v_i-\bar v)(w_i-\bar w)+\bar v\bar w-\sum_j\pi_jv_jw_j$. The last two terms
are $-\Cov_{\pi}(v,w)$ because $v_0=w_0=0$ contributes nothing to
$\sum_j\pi_jv_jw_j$. Setting $w=v$ gives the second identity.
 
\textbf{Step 3: Positive definiteness of $\E_{\sigma}[H]$}
For any $v=(v_1,\dots,v_J)\in\R^{J}$, letting $v_0=0$, the quadratic form of $\E_{\sigma}[H]$ is 
\begin{align*}
    v'\E_{\sigma}[H]v =& \E_{\sigma}[v'Hv] \\
    =& \E_{\sigma}\Big[\sum_{j=0}^{J}\pi_jv_j^2-\sum_{j=0}^{J}\sum_{k=0}^{J}\pi_j\pi_kv_jv_k\Big]\\
    =& \E_{\sigma}\Big[\frac{1}{2}\sum_{j=0}^{J}\sum_{k=0}^{J}\pi_j\pi_k (v_j^2+v_k^2)-\sum_{j=0}^{J}\sum_{k=0}^{J}\pi_j\pi_kv_jv_k\Big]\quad \text{(since $\sum_{j=0}^{J}\pi_j=1$)}\\
    =& \E_{\sigma}\Big[\frac{1}{2}\sum_{j=0}^{J}\sum_{k=0}^{J}\pi_j\pi_k (v_j^2+v_k^2-2v_jv_k)\Big]\\
    =& \E_{\sigma}\Big[\frac{1}{2}\sum_{j=0}^{J}\sum_{k=0}^{J}\pi_j\pi_k (v_j-v_k)^2\Big]\\
    \ge& \E_{\sigma}\Big[\frac{1}{2}\sum_{j=0}^{J}\pi_j\pi_0 v_j^2\Big]\\
    \ge& \frac{1}{2}\min_{j}\E_{\sigma}[\pi_j\pi_0]\|v\|_2^2\\
\end{align*}
This gives that for any $\sigma\in\Theta_\sigma$,
\begin{equation}\label{eq:Hlambda_min}
    \lambda_{\min}(\E_{\sigma}[H])\ge \frac{1}{2}\min_{j}\E_{\sigma}[\pi_j\pi_0]>0.
\end{equation}
By implicit function theorem, $y(\sigma)$ exists and is the solution to the system of equations $\E_{\sigma}[\pi(y(\sigma)+pb)]-s=0$. Observe that $\frac{\partial}{\partial y(\sigma)'}(\E_{\sigma}[\pi(y(\sigma)+pb)]-s)$ is the matrix $\E_{\sigma}[H]$.
By implicit function theorem, $y(\sigma)=f_s^{-1}(s,p,\sigma)\in C^1((0,1)^J\times[- M, M]^J\times \Theta_{\sigma})$. The fact that the mean utility vector $y=y(\sigma_0)$ and price $p$ are bounded implies the share lies in a compact set $[\underline s,\bar s]^J\subset(0,1)^J$. Together, $y(\sigma)$ is a continuous function of $\sigma,s,p$ on a compact set $[\underline s,\bar s]^J\times[- M, M]^J\times \Theta_{\sigma}$. Therefore, $y(\sigma)$ is bounded uniformly, and the bound is denoted by $M_y$. We then prove that the $\lambda_{\min}(\E_{\sigma}[H])$ is bounded away from zero uniformly.
Observe that
\begin{align*}
    \E_{\sigma}[\pi_j\pi_0] =& \E_{\sigma}\Big[\frac{\exp(u_j)}{1+\sum_{k=1}^{J}\exp(u_k)}\frac{1}{1+\sum_{k=1}^{J}\exp(u_k)}\Big]\\
    =& \E_{\sigma}\Big[\frac{\exp(y_j(\sigma)+b\,p_j)}{(1+\sum_{k=1}^{J}\exp(y_k(\sigma)+b\,p_k))^{2}}\Big]\\
    \ge& \E_{\sigma}\Big[\frac{\exp(y_j(\sigma)+b\,p_j)}{(1+\sum_{k=1}^{J}\exp(y_k(\sigma)+b\,p_k))^{2}}I[b\le 0]\Big]\quad\text{w.l.o.g. assume } \inf_{\Theta_\sigma}P_{\sigma}(b\le0)>0\\
    \ge& \E_{\sigma}\Big[\frac{\exp(-M_y+b\,M)}{(1+J\exp(M_y-b\,M))^2}I[b\le0]\Big]
    =:c>0
\end{align*}
\textbf{Step 4: First and second order derivatives of $y(\sigma)$}
Define $A:=\E_{\sigma}[H]$. We have%
\begin{equation}\label{eq:Yprime}
\partial_\sigma y(\sigma)=-A^{-1}\,\E_{\sigma}\!\big[\pi\,\ell_{\sigma}\big]
=-\big(\E_{\sigma}[H]\big)^{-1}\E_{\sigma}[\pi\,\ell_{\sigma}].
\end{equation}
and
\begin{equation}\label{eq:Ysecond}
\partial_\sigma^2 y(\sigma)=-A^{-1}\Big\{\E_{\sigma}\!\big[T[\partial_{\sigma} y(\sigma),\partial_{\sigma} y(\sigma)]\big]
+2\,\E_{\sigma}\!\big[H\,\ell_{\sigma}\big]\,\partial_{\sigma} y(\sigma)
+\E_{\sigma}\!\big[\pi\,q\big]\Big\},
\end{equation}
The component of $T[\partial_{\sigma} y(\sigma),\partial_{\sigma} y(\sigma)]$ in \eqref{eq:Ysecond} is
\begin{equation}\label{eq:Ysecondcomp}
\E_{\sigma}T[\partial_{\sigma} y(\sigma),\partial_{\sigma} y(\sigma)]_{j}=
\E_{\sigma}\Big\{\pi_j\Big[\big(\partial_\sigma y_j(\sigma)-\overline{\partial_\sigma y(\sigma)}\big)^{2}-\sum_{k=0}^{J}\pi_k\big(\partial_\sigma y_k(\sigma)-\overline{\partial_\sigma y(\sigma)}\big)^{2}\Big]
\Big\}.
\end{equation}
For equation \eqref{eq:Yprime}, take $g(b,\sigma)=\pi\big(y(\sigma)+b\,p\big)$. By the chain rule and
\eqref{eq:H}, $\partial_\sigma g=H\,\partial_\sigma y(\sigma)$. Differentiating
\eqref{eq:share} with the rule \eqref{eq:firstrule}, and using that the
right-hand side of \eqref{eq:share} is constant in $\sigma$,
\begin{equation}\label{eq:shareDeriv}
    0=\E_{\sigma}\!\big[H\,\partial_{\sigma}y(\sigma)+\pi\,\ell_{\sigma}\big]=\E_{\sigma}[H]\,\partial_{\sigma}y(\sigma)+\E_{\sigma}[\pi\,\ell_{\sigma}],
\end{equation}
where $\partial_{\sigma}y(\sigma)$ leaves the expectation because $y(\sigma)=f_s^{-1}(s,p,\sigma)$ is a function of $p$ and $s$ that are independent of the random variable $b$. Solving from \eqref{eq:shareDeriv} gives \eqref{eq:Yprime}.

For \eqref{eq:Ysecond}, keep $g(b,\sigma)=\pi\big(y(\sigma)+b\,p\big)$ and differentiate \eqref{eq:share} with the rule \eqref{eq:secondrule}. 
Notice that \eqref{eq:piDeriv} gives $\partial_\sigma g=H\,\partial_\sigma y(\sigma)$ and that \eqref{eq:piDeriv2} gives $\partial_\sigma H(u)\partial_\sigma y(\sigma)=T[\partial_\sigma y(\sigma),\partial_\sigma y(\sigma)]$. There is
\begin{equation}
    \partial_\sigma^2 g=\partial_\sigma(H\,\partial_\sigma y(\sigma))=\partial_\sigma H\,\partial_\sigma y(\sigma)+H\,\partial_\sigma^2 y(\sigma)=T[\partial_\sigma y(\sigma),\partial_\sigma y(\sigma)]+H\,\partial_\sigma^2 y(\sigma).
\end{equation}
and
\begin{equation*}
    0 = \E_{\sigma}\!\partial_\sigma^2 g+2\,\ell_{\sigma}\,\partial_\sigma g+g\,q\big]
    = \E_{\sigma}\!\big[T[\partial_\sigma y(\sigma),\partial_\sigma y(\sigma)]+\E_{\sigma}[H]\,\partial_\sigma^2 y(\sigma)+2\,\E_{\sigma}[\ell_{\sigma}\,H]\,\partial_\sigma y(\sigma)+\E_{\sigma}[\pi\,q],
\end{equation*}
which gives \eqref{eq:Ysecond} after solving for $\partial_\sigma^2 y(\sigma)$.
Expanding the first term with equation \eqref{eq:Tcontract} at $v=w=\partial_\sigma y(\sigma)$ yields \eqref{eq:Ysecondcomp}.

\textbf{Step 5: Bounds}
By \eqref{eq:Yprime}, we have
\begin{equation}\label{eq:boundYprime}
\norm{\partial_\sigma y(\sigma)}\le
\frac{\norm{\E_{\sigma}(\pi\ell_{\sigma})}}{\lambda_{\min}\big(\E_{\sigma}[H]\big)}\le C, \ a.s.
\end{equation}
This is because $\|\E_\sigma(\pi\ell_\sigma)\|\le \E_{\sigma}\abs{\ell_{\sigma}}=\int |\partial_\sigma f(b,\sigma)|db\le C$.
Observe that 
\begin{align*}
    \|\E_{\sigma}H\ell_\sigma\|_{op}\le \E_{\sigma}\|H\ell_\sigma\|_{op}\le \E_{\sigma}|\ell_\sigma|\|H\|_{op}\le \E_{\sigma}|\ell_\sigma|\tr H\le \E_{\sigma}|\ell_\sigma|\le C,
\end{align*}
where $\|H\|_{op}\le \tr H$ because $H$ is a positive semi-definite matrix by Step 3,
and that
\begin{align*}
    \|\E_{\sigma}\pi q\|\le \E_{\sigma}\|\pi q\|\le \E_{\sigma}|q|=\int|\partial_\sigma^2 f|db\le C.    
\end{align*} 
Then by \eqref{eq:Ysecond}, we have
\begin{align}\label{eq:boundYsecond}
\norm{\partial_\sigma^2 y(\sigma)}\le&
\frac{\|\{\E_{\sigma}\!\big[T[\partial_{\sigma} y(\sigma),\partial_{\sigma} y(\sigma)]\big]\|
+2\,\|\E_{\sigma}\!\big[H\,\ell_{\sigma}\big]\|_{op}\,\|\partial_{\sigma} y(\sigma)\|
+\|\E_{\sigma}\!\big[\pi\,q\big]\|\Big\}}
{\lambda_{\min}\big(\E_{\sigma}[H]\big)}\nonumber\\
\le & \frac{\{\|\E_{\sigma}\!\big[T[\partial_{\sigma} y(\sigma),\partial_{\sigma} y(\sigma)]\big]\|
+C\}}
{c}
\end{align}
By \eqref{eq:Ysecondcomp}, we have 
\begin{align*}
    |T[\partial_{\sigma} y(\sigma),\partial_{\sigma} y(\sigma)]_{j}|=&
|\E_{\sigma}\Big\{\pi_j\Big[\big(\partial_\sigma y_j(\sigma)-\overline{\partial_\sigma y(\sigma)}\big)^{2}-\sum_{k=0}^{J}\pi_k\big(\partial_\sigma y_k(\sigma)-\overline{\partial_\sigma y(\sigma)}\big)^{2}\Big]\Big\}|\\
\le& |\E_{\sigma}\Big\{\max_{j}\big(\partial_\sigma y_j(\sigma)-\overline{\partial_\sigma y(\sigma)}\big)^{2}\Big\}|\le C
\end{align*}

Consequently $\norm{\partial_\sigma^2 y(\sigma)}\le C, \,a.s.$

\end{proof}

\begin{proof}[Proof of Lemma~\ref{lemma:ell-infty-id}]

\textbf{\emph{Step 1: from $\|g(\theta)\|_\infty$ to the residualized moment.}}
Fix $k\le d_z$ and let $a_k:=(1,-\Pi_{0,k}')'/(\|\Pi_{0,k}\|_1+1)$ be the weight vector acting on
$(z_{jt,k},x_{jt}')'$, so that $\|a_k\|_1=1$. Since $\tilde z_{jt}$ contains both $z_{jt}$ and
$x_{jt}$, and $\|v\|_\infty=\sup_{\|a\|_1\le1}|a'v|$,
\begin{align*}
    \|g(\theta)\|_\infty
    &\ge \big\|\E_J\E[(z_{jt,k},x_{jt}')'(y_{jt}(\sigma)-\alpha p_{jt}-x_{jt}'\beta)]\big\|_\infty\\
    &\ge \frac{1}{\|\Pi_{0,k}\|_1+1}\,
      \big|\E_J\E[(z_{jt,k}-x_{jt}'\Pi_{0,k})(y_{jt}(\sigma)-\alpha p_{jt}-x_{jt}'\beta)]\big|\\
    &\ge \frac{1}{M+1}\,
      \big|\E_J\E[(\epsilon^z_{jt,k}+r^z_{jt,k})(y_{jt}(\sigma)-\alpha p_{jt}-x_{jt}'\beta)]\big|,
\end{align*}
where the last line uses $\|\Pi_{0,k}\|_1\le\|\Pi_0\|_1\le M$ (\ref{cond:BoundednessVar2}) and the
decomposition $z_{jt,k}=x_{jt}'\Pi_{0,k}+r^z_{jt,k}+\epsilon^z_{jt,k}$ of (\ref{eq:highdimBLP_approx_z}).

For the remainder term, Condition \ref{cond:BoundednessNuisance} gives
$|y_{jt}(\sigma)-\alpha p_{jt}-x_{jt}'\beta|\le C$ uniformly on $\Theta_\sigma\times\Theta_{\alpha}\times\{\beta\in\R^{d_x}:\|\beta\|_1\le M_{\beta}\}$, since $y_{jt}(\cdot)$ is
uniformly bounded on $\Theta_\sigma$, $|p_{jt}|\le M$, and
$|\alpha p_{jt}+x_{jt}'\beta|\le \|(p_{jt},x_{jt}')\|_\infty\|(\alpha,\beta')\|_1\le MM_{\beta}$. Hence by Cauchy--Schwarz and
$\E(r^z_{jt,k})^2\lesssim d_0/T$,
\begin{align*}
    \big|\E_J\E[r^z_{jt,k}(y_{jt}(\sigma)-\alpha p_{jt}-x_{jt}'\beta)]\big|
    \;\le\; C\sqrt{d_0/T}.
\end{align*}
Since $\E[\epsilon^z_{jt}\mid x_{jt}]=0$ kills $x_{jt}'\beta$ for every $\beta$, and
$\E[\epsilon^z_{jt}f_{u0}(x_{jt})]=\E[\epsilon^z_{jt}\xi_{jt}]=0$ by (\ref{eq:highdimBLP_utility}) and (\ref{eq:highdimBLP_instrument}), writing
$y_{jt}(\sigma_0)=f_{u0}(x_{jt})+\alpha_0p_{jt}+\xi_{jt}$ gives
\begin{align*}
    \|g(\theta)\|_\infty
    \;\ge&\; \frac{1}{M+1}\big|\E_J\E[\epsilon^z_{jt,k}\,\Delta_{jt}(\theta_1)]\big|
    - \frac{1}{M+1}\big|\E_J\E[r^z_{jt,k}\,\Delta_{jt}(\theta_1)]\big|\\
    \;\ge&\; \frac{1}{M+1}\big|\E_J\E[\epsilon^z_{jt,k}\,\Delta_{jt}(\theta_1)]\big|
    - \frac{C}{M+1}\sqrt{d_0/T},
    \qquad
    \Delta_{jt}(\theta_1):=y_{jt}(\sigma)-y_{jt}(\sigma_0)-(\alpha-\alpha_0)p_{jt}.
\end{align*}
As $d_z$ is fixed, $\|v\|_2\le\sqrt{d_z}\max_k|v_k|$ and taking the maximum over $k\le d_z$ yields
\begin{align}\label{eq:step1}
    \|g(\theta)\|_\infty
    \;\ge\; \frac{1}{\sqrt{d_z}(M+1)}\big\|\E_J\E[\epsilon^z_{jt}\Delta_{jt}(\theta_1)]\big\|_2
    - \frac{C}{M+1}\sqrt{d_0/T}.
\end{align}

\textbf{\emph{Step 2: from the residualized moment to $\theta_1$.}}
Note $\E_J\E[\epsilon^z_{jt}\Delta_{jt}(\theta_1)]=\E[\psi(w_t,\theta_1,f_{z0},f_{u0})]$.
Fix $\epsilon_*>0$, to be chosen below.

\textbf{\emph{Case 1: $|\sigma-\sigma_0|\le\epsilon_*$.}}
By Lemma~\ref{lemma:UniformBoundedness}, $\sup_{\sigma\in\Theta_\sigma}\max_{j,t}|\partial^2_\sigma y_{jt}(\sigma)|\le C$,
so a second-order Taylor expansion gives
$\Delta_{jt}(\theta_1)=\partial_\sigma y_{jt}(\sigma_0)(\sigma-\sigma_0)
-(\alpha-\alpha_0)p_{jt}+R_{jt}$ with $|R_{jt}|\le C(\sigma-\sigma_0)^2$ uniformly in $(j,t)$.
Since $\|\epsilon^z_{jt}\|_\infty\le M$,
$\|\E_J\E[\epsilon^z_{jt}R_{jt}]\|_2\le C(\sigma-\sigma_0)^2\le C\epsilon_*|\sigma-\sigma_0|$.
Writing $G':=\E_J\E[\epsilon^z_{jt}(\partial_\sigma y_{jt}(\sigma_0),-p_{jt})]\in\R^{d_z\times d_{\theta_1}}$,
\begin{align*}
    \big\|\E[\psi(w_t,\theta_1,f_{z0},f_{u0})]\big\|_2
    \;\ge\; s_{\min}(G')\,\|\theta_1-\theta_{10}\|_2 - C\epsilon_*|\sigma-\sigma_0|,
\end{align*}
where $s_{\min}(G')=\lambda_{\min}(GG')^{1/2}$ is bounded below uniformly in $T$ by
Condition \ref{cond:IdentificationGWG}. Substituting into \eqref{eq:step1} and choosing $\epsilon_*$ small enough that
$C\epsilon_*\le s_{\min}(G')/2$,
\begin{align*}
    \|g(\theta)\|_\infty
    \;\ge\; \frac{s_{\min}(G')}{2\sqrt{d_z}(M+1)}\|\theta_1-\theta_{10}\|_2
    - \frac{C}{M+1}\sqrt{d_0/T}.
\end{align*}

\textbf{\emph{Case 2: $|\sigma-\sigma_0|>\epsilon_*$.}}
Condition 3.3(i) and the continuity of $\theta_1\mapsto\E[\psi(w_t,\theta_1,f_{z0},f_{u0})]$ on the
compact set $\Theta_\sigma\times\Theta_\alpha$ give $\eta_*>0$ with
\begin{align*}
   \inf_{\|\theta_1-\theta_{10}\|_2>\epsilon_*}\big\|\E[\psi(w_t,\theta_1,f_{z0},f_{u0})]\big\|_2\;\ge\;\eta_*.
\end{align*}
Let $D:=\operatorname{diam}(\Theta_\sigma\times\Theta_\alpha)<\infty$, so that
$\|\theta_1-\theta_{10}\|_2\le D$ and hence $\eta_*\ge(\eta_*/D)\|\theta_1-\theta_{10}\|_2$.
Substituting into \eqref{eq:step1},
\begin{align*}
    \|g(\theta)\|_\infty
    \;\ge\; \frac{\eta_*}{D\sqrt{d_z}(M+1)}\|\theta_1-\theta_{10}\|_2
    - \frac{C}{M+1}\sqrt{d_0/T}.
\end{align*}

\textbf{\emph{Conclusion.}} In both cases there is $c>0$, independent of $T$, with
$\|g(\theta)\|_\infty\ge c\|\theta_1-\theta_{10}\|_2-C\sqrt{d_0/T}$, so
$\|\theta_1-\theta_{10}\|_2\le c^{-1}\big(\|g(\theta)\|_\infty+C\sqrt{d_0/T}\big)
\le 2c^{-1}(1\vee C)\big(\|g(\theta)\|_\infty\vee\sqrt{d_0/T}\big)$.
\end{proof}

\begin{proof}[Proof of Lemma \ref{lemma:empbound}]
    It is a direct application of Lemma \ref{lemma:Hoeffding}.
    Note
    \[
        \normm{\hat{g}(\theta_0)}_\infty=\normm{\E_T(\E_J\tilde{z}_{jt}(\xi_{jt}+r_{jt}^u))}_\infty
    \]
    where $\normm{\E_J\tilde{z}_{jt}\xi_{jt}}_\infty\leq \E_J\normm{\tilde{z}_{jt}}_{\infty}\cdot\normm{\xi_{jt}}_\infty\leq M^2$.
    We have
    \begin{align*}
        &P\left(\normm{\hat{g}(\theta_0)}_\infty\geq 2t\right)\\
        =&P\left(\normm{\E_T(\E_J\tilde{z}_{jt}(\xi_{jt}+r_{jt}^u))}_\infty\geq 2t\right)\\
        \leq & P\left(\normm{\E_T(\E_J\tilde{z}_{jt}\xi_{jt})}_\infty\geq t\right) + P\left(\normm{\E_T(\E_J\tilde{z}_{jt}r_{jt}^u)}_\infty\geq t\right)\\
        \leq & 2d_{\tilde{z}} \exp\left(-\frac{Tt^2}{2M^4}\right) + P\left(M\E_T(\E_J|r_{jt}^u|)\geq t\right)\\
        & \text{ (By Lemma \ref{lemma:Hoeffding})}\\
        \leq & 2d_{\tilde{z}} \exp\left(-\frac{Tt^2}{2M^4}\right) + \E(M\E_T(\E_J|r_{jt}^u|))/t\\
        \leq & 2d_{\tilde{z}} \exp\left(-\frac{Tt^2}{2M^4}\right) + \frac{M\E_J\sqrt{\E (r_{jt}^u)^2}}{t}\\
        \leq & 2d_{\tilde{z}} \exp\left(-\frac{Tt^2}{2M^4}\right) + \frac{C}{t}\sqrt{d_0/T}\\
        &\text{Let } t=\frac{1}{2}C_1 T^{-1/2}\sqrt{d_0\log(d_{\tilde{z}}\vee T)}, \text{ then}\\
        = & 2d_{\tilde{z}}(\frac{1}{d_{\tilde{z}}\vee T})^{\frac{C_1^2d_0}{8M^4}} + \frac{2C}{C_1\sqrt{\log(d_{\tilde{z}}\vee T)}}\\
        \leq & 2(d_{\tilde{z}}\vee T)^{1-\frac{C_1^2d_0}{8M^4}} + \frac{2C}{C_1\sqrt{\log(d_{\tilde{z}}\vee T)}}\\
         = & o(1) \text{ when } C_1 \text{ is large enough, for example } C_1>\sqrt{8M^4}
    \end{align*}
\end{proof}

\begin{proof}[Proof of Lemma \ref{lemma:supdiff}]
    Note by Lemma~\ref{lemma:UniformBoundedness}, thus we have 
    \[
    \sup_{\sigma\in\Theta_\sigma}\normm{\partial_{\sigma} y_{jt}(\sigma)}_\infty\leq M_y
    \] for some constant $M_y$.
    Thus we have $\normm{\E_{JT}\tilde{z}_{jt}(y_{jt}(\sigma)-y_{jt}(\sigma'))}_\infty\leq C\|\sigma-\sigma'\|_1$ and $\normm{\E \tilde{z}_{jt}(y_{jt}(\sigma)-y_{jt}(\sigma'))}_\infty\leq C|\sigma-\sigma'|$, where C is the bound s.t. $\sup_{\sigma\in\Theta_\sigma}\norm{\tilde{z}_{jt}\partial_{\sigma} y_{jt}(\sigma)}_\infty\leq C$.
    Then by Hoeffding inequality, we have
    \begin{align*}
        &P\left(\sup_{\sigma\in\Theta_\sigma}\norm{\E_{JT}\tilde{z}_{jt}y_{jt}(\sigma)-\E \tilde{z}_{jt}y_{jt}(\sigma)}_\infty\geq t\right)\\
        \leq & P\Biggl(\sup_{1\leq i\leq \bar C(M_\sigma/\delta)^{d_\sigma}}\norm{\E_{JT}\tilde{z}_{jt}y_{jt}(\sigma_i)-\E \tilde{z}_{jt}y_{jt}(\sigma_i)}_\infty+
        \sup_{\|\sigma-\sigma'\|_{1}\leq \delta}\norm{\E_{JT}\tilde{z}_{jt}y_{jt}(\sigma)-\E_{JT}\tilde{z}_{jt}y_{jt}(\sigma')}_\infty+\\
        & \sup_{\|\sigma-\sigma'\|_{1}\leq \delta}\norm{\E \tilde{z}_{jt}y_{jt}(\sigma)-\E \tilde{z}_{jt}y_{jt}(\sigma')}_\infty\geq t\Biggl)\\
        & \text{choose $\delta\;s.t.\;C\delta= t/3$,}  \text{ then}\\
        \leq& P\left(\sup_{i\leq M_\sigma/\delta\leq \bar C(3CM_\sigma/t)^{d_\sigma}}\norm{\E_{JT}\tilde{z}_{jt}y_{jt}(\sigma_i)-\E \tilde{z}_{jt}y_{jt}(\sigma_i)}_\infty\geq t/3\right)\\
        \leq & \sum_{i=1}^{\bar C (3CM_\sigma/t)^{d_\sigma}}P\left(\norm{\E_{JT}\tilde{z}_{jt}y_{jt}(\sigma_i)-\E \tilde{z}_{jt}y_{jt}(\sigma_i)}_\infty\geq t/3\right)\\
        \leq & \sum_{i=1}^{\bar C (3CM_\sigma/t)^{d_\sigma}} 2d_{\tilde{z}} \exp\left(-\frac{Tt^2}{18M^4}\right)\\
        &\text{(by Lemma \ref{lemma:Hoeffding} and note that $\normm{\tilde{z}_{jt}y_{jt}(\sigma_i)}_\infty\leq \normm{\tilde{z}_{jt}}_\infty\cdot\normm{y_{jt}(\sigma_i)}_\infty\leq M^2$)}\\
        \lesssim& (\frac{M_\sigma}{t})^{d_\sigma}d_{\tilde{z}}\exp\left(-\frac{Tt^2}{18M^4}\right)\\
        &\text{Let } t=C_2 T^{-1/2}\sqrt{\log(d_{\tilde{z}}\vee T)}, \text{ then}\\
        \lesssim & (\frac{T^{1/2}}{\sqrt{\log(d_{\tilde{z}}\vee T)}})^{d_\sigma}d_{\tilde{z}}(\frac{1}{d_{\tilde{z}}\vee T})^{\frac{C_2^2}{18M^4}}\\
        \lesssim & (T\vee d_{\tilde{z}})^{\frac{d_\sigma+2}{2}-\frac{C_2^2}{18M^4}}\\
        = & o(1) \text{   (when $C_2$ is large enough)}
    \end{align*}
    The second inequality is a direct application of Lemma \ref{lemma:Hoeffding}.
    Let $\tilde{x}_{jt}=(p_{jt},x_{jt}')'$, then we have $\normm{\tilde{z}_{jt}\tilde{x}_{jt}'}_\infty\leq \normm{\tilde{z}_{jt}}_\infty\cdot\normm{\tilde{x}_{jt}'}_\infty\leq M^2$.
    \begin{align*}
        & P\left(\norm{\E_{JT}\tilde{z}_{jt}\tilde{x}_{jt}'-\E \tilde{z}_{jt}\tilde{x}_{jt}'}_\infty\geq t\right)\\
        &(\text{by Lemma \ref{lemma:Hoeffding} and note that } \normm{\tilde{z}_{jt}\tilde{x}_{jt}'}_\infty\leq \normm{\tilde{z}_{jt}}_\infty\cdot\normm{\tilde{x}_{jt}'}_\infty\leq M^2)\\
        \leq & 2(d_x+1)d_{\tilde{z}}\exp(-\frac{Tt^2}{2M^4})\\
        \leq & 2d_{\tilde{z}}^2\exp(-\frac{Tt^2}{2M^4})\\
        & \text{ Let } t=C_3 T^{-1/2}\sqrt{\log(d_{\tilde{z}}\vee T)}, \text{ then}\\
        = & 2d_{\tilde{z}}^2(\frac{1}{d_{\tilde{z}}\vee T})^{\frac{C_3^2}{2M^4}}\\
        \leq & 2(d_{\tilde{z}}\vee T)^{2-\frac{C_3^2}{2M^4}}\\
        = & o(1) \text{   (when $C_3$ is large enough, for example $C_3>\sqrt{4M^4}$)}
    \end{align*}
\end{proof}

\begin{proof}[Proof of Lemma~\ref{lem:SE_RE}]
The proof strategy follows \citet{bickel2009simultaneous}.
Write $\normm{\Delta}_M := (\Delta'M\Delta)^{1/2}$, which is a seminorm. Fix $\Delta \in \mc{C}_\nu(S)$. Order the indices in $S^c$ so that $\abs{\Delta_j}$ is nonincreasing, and partition $S^c$ into consecutive blocks $T_1, T_2, \dots$ of size $m$ (the last block may be smaller). Let $T_{01} := S \cup T_1$, so that $\abs{T_{01}} \le d_0 + m$. For $k \ge 2$, every entry of $\Delta_{T_k}$ is bounded in absolute value by $\normm{\Delta_{T_{k-1}}}_1/m$. Hence $\normm{\Delta_{T_k}}_2 \le \normm{\Delta_{T_{k-1}}}_1/\sqrt{m}$, and
\[
\sum_{k\ge2}\normm{\Delta_{T_k}}_2 \le \frac{\normm{\Delta_{S^c}}_1}{\sqrt m} \le \frac{\nu\normm{\Delta_S}_1}{\sqrt m} \le \nu\sqrt{\frac{d_0}{m}}\,\normm{\Delta_S}_2 \le \nu\sqrt{\frac{d_0}{m}}\,\normm{\Delta_{T_{01}}}_2 \le \frac12\sqrt{\frac{\underline\phi}{\bar\phi}}\,\normm{\Delta_{T_{01}}}_2 ,
\]
where the last step uses $m \ge 4\nu^2(\bar\phi/\underline\phi)d_0$. By the triangle inequality for $\normm{\cdot}_M$ and the definitions of the sparse eigenvalues,
\[
\normm{\Delta}_M \ge \normm{\Delta_{T_{01}}}_M - \sum_{k\ge2}\normm{\Delta_{T_k}}_M
\ge \sqrt{\underline\phi}\,\normm{\Delta_{T_{01}}}_2 - \sqrt{\bar\phi}\sum_{k\ge2}\normm{\Delta_{T_k}}_2
\ge \frac{\sqrt{\underline\phi}}{2}\normm{\Delta_{T_{01}}}_2 .
\]
Moreover, since $\underline\phi \le \bar\phi$,
\[
\normm{\Delta}_2 \le \normm{\Delta_{T_{01}}}_2 + \sum_{k\ge2}\normm{\Delta_{T_k}}_2 \le \frac32\normm{\Delta_{T_{01}}}_2 .
\]
Combining the two displays gives $\normm{\Delta}_M \ge \frac{\sqrt{\underline\phi}}{3}\normm{\Delta}_2$.
\end{proof}
 
\begin{proof}[Proof of Lemma~\ref{lem:SE_RE2}]
Since $m + (d_z+1)d_0 = \ell_Td_0$ and $m \le \ell_Td_0$, Condition \ref{cond:SE} and the monotonicity of sparse eigenvalues give (i) w.p.a.1. By monotonicity again, $\phi_{\min}(m + d_0) \ge \phi_{\min}(m + (d_z+1)d_0) \ge \underline\phi$. Since $\ell_T \to \infty$, $m \ge 4\max(\nu_1,\nu_2)^2(\bar\phi/\underline\phi)d_0$ for all $T$ large enough. As $\abs{S_1} \le d_0$, Lemma \ref{lem:SE_RE} gives (ii) and (iii).
\end{proof}

\begin{proof}[Proof of Lemma~\ref{lem:lasso_approx}]
Throughout we work under the event $\mc{E}$ and write $\Delta := \hat\Delta$.
 
\textbf{Step 0: basic inequality.} By the definition of $\hat\beta$,
\[
\frac1n\normm{Y - X\hat\beta}_2^2 - \frac1n\normm{Y - X\beta_0}_2^2 \le \lambda\big(\normm{\beta_{0,P}}_1 - \normm{\hat\beta_P}_1\big).
\]
Since $Y - X\beta_0 = a + e$ and $Y - X\hat\beta = a + e - X\Delta$, the left-hand side equals $\normm{X\Delta}_n^2 - \frac2n(a+e)'X\Delta$. For the penalty, note that $S^c \subset \mc P$ because $\mc U \subset S$, $\beta_{0,j} = 0$ for $j \in S^c$, and $\sum_{j\in S^c}|\hat\beta_{\mc P,j}|=\|\Delta_{S^c}\|_1$. By the triangle inequality,
\[
\normm{\beta_{0,\mc P}}_1 - \normm{\hat\beta_{\mc P}}_1
= \sum_{j\in S\cap \mc P}\big(\abs{\beta_{0,j}} - \abs{\beta_{0,j}+\Delta_j}\big) - \normm{\Delta_{S^c}}_1
\le \normm{\Delta_S}_1 - \normm{\Delta_{S^c}}_1 .
\]
For the noise term, H\"older's inequality and the event $\mc{E}$ give
\[
\frac2n\abs{e'X\Delta} \le 2\norm{\frac{X'e}{n}}_\infty\normm{\Delta}_1 \le \frac{\nu_0-1}{\nu_0+1}\,\lambda\big(\normm{\Delta_S}_1 + \normm{\Delta_{S^c}}_1\big).
\]
For the approximation error, the Cauchy--Schwarz inequality gives
\[
\frac2n\abs{a'X\Delta} \le 2\cdot\frac{\normm{a}_2}{\sqrt n}\cdot\frac{\normm{X\Delta}_2}{\sqrt n} = 2c_a\normm{X\Delta}_n .
\]
Combining the four displays,
\begin{equation}\label{eq:lasso_approx_basic}
\normm{X\Delta}_n^2 - 2c_a\normm{X\Delta}_n \le \frac{2\nu_0}{\nu_0+1}\,\lambda\normm{\Delta_S}_1 - \frac{2}{\nu_0+1}\,\lambda\normm{\Delta_{S^c}}_1 .
\end{equation}
 
\textbf{Step 1: the case $\normm{X\Delta}_n \ge 2c_a$.} The left-hand side of \eqref{eq:lasso_approx_basic} is nonnegative, so $\normm{\Delta_{S^c}}_1 \le \nu_0\normm{\Delta_S}_1$, i.e.\ $\Delta \in \mc{C}_{\nu_0}(S) \subset \mc{C}_{\nu}(S)$. The restricted eigenvalue condition gives $\normm{\Delta}_2 \le \normm{X\Delta}_n/\sqrt{\kappa}$, and hence
\[
\normm{\Delta_S}_1 \le \sqrt{d_0}\,\normm{\Delta_S}_2 \le \sqrt{d_0}\,\normm{\Delta}_2 \le \frac{\sqrt{d_0}}{\sqrt{\kappa}}\normm{X\Delta}_n .
\]
Dropping the last (nonpositive) term in \eqref{eq:lasso_approx_basic},
\[
\normm{X\Delta}_n^2 \le 2c_a\normm{X\Delta}_n + \frac{2\nu_0}{\nu_0+1}\cdot\frac{\sqrt{d_0}\,\lambda}{\sqrt{\kappa}}\normm{X\Delta}_n,
\]
which yields (i). Consequently,
\[
\normm{\Delta}_2 \le \frac{\normm{X\Delta}_n}{\sqrt{\kappa}} \le \frac{2c_a}{\sqrt{\kappa}} + \frac{2\nu_0}{\nu_0+1}\cdot\frac{\sqrt{d_0}\,\lambda}{\kappa},
\qquad
\normm{\Delta}_1 \le (1+\nu_0)\normm{\Delta_S}_1 \le (1+\nu)\sqrt{d_0}\,\normm{\Delta}_2 ,
\]
so (i)--(iii) hold in this case.
 
\textbf{Step 2: the case $\normm{X\Delta}_n < 2c_a$.} Bound (i) holds trivially. Since $x^2 - 2c_a x \ge -c_a^2$ for all $x\in\R$, \eqref{eq:lasso_approx_basic} yields $\frac{2}{\nu_0+1}\lambda\normm{\Delta_{S^c}}_1 \le \frac{2\nu_0}{\nu_0+1}\lambda\normm{\Delta_S}_1 + c_a^2$, that is,
\begin{equation}\label{eq:lasso_approx_cone}
\normm{\Delta_{S^c}}_1 \le \nu_0\normm{\Delta_S}_1 + \frac{\nu_0+1}{2}\cdot\frac{c_a^2}{\lambda} .
\end{equation}
We distinguish two sub-cases.
 
\emph{Case 2a: $(\nu-\nu_0)\normm{\Delta_S}_1 \ge \frac{\nu_0+1}{2}\cdot\frac{c_a^2}{\lambda}$.} Then \eqref{eq:lasso_approx_cone} gives $\normm{\Delta_{S^c}}_1 \le \nu\normm{\Delta_S}_1$, so $\Delta \in \mc{C}_{\nu}(S)$, and by the restricted eigenvalue condition,
\[
\normm{\Delta}_2 \le \frac{\normm{X\Delta}_n}{\sqrt{\kappa}} < \frac{2c_a}{\sqrt{\kappa}},
\qquad
\normm{\Delta}_1 \le (1+\nu)\normm{\Delta_S}_1 \le (1+\nu)\sqrt{d_0}\,\normm{\Delta}_2 < (1+\nu)\frac{2\sqrt{d_0}\,c_a}{\sqrt{\kappa}} .
\]
 
\emph{Case 2b: $(\nu-\nu_0)\normm{\Delta_S}_1 < \frac{\nu_0+1}{2}\cdot\frac{c_a^2}{\lambda}$.} Then $\normm{\Delta_S}_1 < \frac{\nu_0+1}{2(\nu-\nu_0)}\cdot\frac{c_a^2}{\lambda}$, and by \eqref{eq:lasso_approx_cone},
\[
\normm{\Delta}_2 \le \normm{\Delta}_1 = \normm{\Delta_S}_1 + \normm{\Delta_{S^c}}_1
\le (1+\nu_0)\normm{\Delta_S}_1 + \frac{\nu_0+1}{2}\cdot\frac{c_a^2}{\lambda}
< \frac{\nu_0+1}{2}\Big(\frac{1+\nu_0}{\nu-\nu_0} + 1\Big)\frac{c_a^2}{\lambda}
= K_{\nu,\nu_0}\,\frac{c_a^2}{\lambda} .
\]
 
\textbf{Step 3: conclusion.} Collecting steps 1, 2a and 2b gives the bounds in (ii)-(iii).
\end{proof}

\begin{proof}[Proof of Lemma~\ref{lem:lasso_sparse}]
\textbf{Step 1: KKT conditions.} For $j \in \hat M$, the first-order condition of the (convex) LASSO problem gives $\frac2n X_j'(Y - X\hat\beta) = \lambda\,\mathrm{sign}(\hat\beta_j)$, where $X_j$ is the $j$-th column of $X$. Since $Y - X\hat\beta = a + e - X\hat\Delta$ and, under $\mc{E}$, $\abs{X_j'e/n} \le \frac{\nu_0-1}{2(\nu_0+1)}\lambda$, we obtain
\[
\abs{\frac{X_j'(a - X\hat\Delta)}{n}} \ge \frac{\lambda}{2} - \frac{\nu_0-1}{2(\nu_0+1)}\lambda = \frac{\lambda}{\nu_0+1}, \qquad j \in \hat M .
\]
 
\textbf{Step 2: $\hat m \le B\phi_{\max}(\hat m)$.} If $\hat m = 0$ there is nothing to prove. Otherwise, let $X_{\hat M}$ be the submatrix of the columns in $\hat M$. Summing the squares of the display above over $j \in \hat M$,
\[
\hat m\,\frac{\lambda^2}{(\nu_0+1)^2} \le \norm{\frac{X_{\hat M}'(a - X\hat\Delta)}{n}}_2^2
= \sup_{\normm{u}_2 = 1,\ \mathrm{supp}(u)\subset\hat M}\Big(\frac{u'X'(a - X\hat\Delta)}{n}\Big)^2
\le \phi_{\max}(\hat m)\,\normm{a - X\hat\Delta}_n^2 ,
\]
where the last step uses the Cauchy--Schwarz inequality and $\normm{Xu}_n^2 \le \phi_{\max}(\hat m)$. By Lemma \ref{lem:lasso_approx}(i), 
$$\normm{a - X\hat\Delta}_n \le c_a + \normm{X\hat\Delta}_n \le \frac{\lambda}{\nu_0+1}\sqrt{B}.$$ Hence $\hat m \le B\phi_{\max}(\hat m)$.
 
\textbf{Step 3: sublinearity and $\hat m \le m$.} For integers $k, m \ge 1$, $\phi_{\max}(k) \le \lceil k/m\rceil\phi_{\max}(m)$. Indeed, any $u$ with $\normm{u}_0 \le k$ can be written as $u = \sum_{i=1}^{q}u_i$ with $q = \lceil k/m\rceil$ and disjointly supported $u_i$, $\normm{u_i}_0 \le m$. Then 
$$\normm{Xu}_n \le \sum_i\normm{Xu_i}_n \le \sqrt{\phi_{\max}(m)}\sum_i\normm{u_i}_2 \le \sqrt{\phi_{\max}(m)}\sqrt{q}\,\normm{u}_2$$

Now suppose $\hat m > m$. Then $\lceil\hat m/m\rceil \le 2\hat m/m$, and Step 2 gives $\hat m \le B\phi_{\max}(\hat m) \le 2\hat mB\phi_{\max}(m)/m$, i.e.\ $m \le 2B\phi_{\max}(m)$, a contradiction. Hence $\hat m \le m$, and by monotonicity of $\phi_{\max}$, 
$$\hat m \le B\phi_{\max}(\hat m) \le B\phi_{\max}(m).$$
 
\textbf{Step 4: $\ell_2$ norm.} 
$\mathrm{supp}(\hat\Delta) \subset \{j: \hat\beta_j \neq 0\}\cup\{j: \beta_{0,j} \neq 0\}$.
Since $\{j: \hat\beta_j \neq 0\}\subset\hat M\cup P^c\subset\hat M \cup S$ and $\{j: \beta_{0,j} \neq 0\}\subset S$, we have $\mathrm{supp}(\hat\Delta) \subset S \cup \hat M$, so $\normm{\hat\Delta}_0 \le d_0 + m$. Since $\hat\Delta$ is $(m+d_0)$-sparse, the definition of $\phi_{\min}$ and the assumption $\phi_{\min}(m+d_0) \ge \underline\phi$ give $\normm{X\hat\Delta}_n^2 \ge \phi_{\min}(m+d_0)\normm{\hat\Delta}_2^2 \ge \underline\phi\,\normm{\hat\Delta}_2^2$, and (ii) follows from Lemma \ref{lem:lasso_approx}(i).
\end{proof}

\begin{proof}[Proof of Lemma~\ref{lem:lasso_endogeneity}]
\textbf{Step 0: restricted eigenvalue condition and sparse eigenvalues.}
For any $u \in \R^{d_x+1}$, $\hat Xu = \tilde Xu + u_1(\hat F_p - F_p)$. Hence
\[
\abs{\normm{\hat Xu}_n - \normm{\tilde Xu}_n} \le \abs{u_1}c_p \le c_p\normm{u}_2 .
\]
Combined with (a) and (c), this gives 
\[
\normm{\hat X\Delta}_n \ge (\sqrt{\kappa_1} - c_p)\normm{\Delta}_2 \ge \frac{\sqrt{\kappa_1}}{2}\normm{\Delta}_2
\]
for every $\Delta \in \mc{C}_{\nu_1}(S_1)$. Thus $\frac1n\hat X'\hat X$ satisfies the restricted eigenvalue condition over $S_1$ with parameters $(\kappa_1/4, \nu_1)$. The same inequality also transfers the sparse eigenvalues. Take any $u \neq 0$.
\begin{itemize}
    \item If $\normm{u}_0 \le m + d_0$, then $\normm{\tilde Xu}_n \ge \sqrt{\underline\phi}\,\normm{u}_2$ by (b), and $c_p \le \sqrt{\underline\phi}/2$ by (c). Hence
    \[
    \normm{\hat Xu}_n \ge \normm{\tilde Xu}_n - c_p\normm{u}_2 \ge \Big(\sqrt{\underline\phi} - \frac{\sqrt{\underline\phi}}{2}\Big)\normm{u}_2 = \frac{\sqrt{\underline\phi}}{2}\normm{u}_2 .
    \]
    \item If $\normm{u}_0 \le m$, then $\normm{\tilde Xu}_n \le \sqrt{\bar\phi}\,\normm{u}_2$ by (b). Using (c) and $\underline\phi \le \bar\phi$,
    \[
    \normm{\hat Xu}_n \le \normm{\tilde Xu}_n + c_p\normm{u}_2 \le \Big(\sqrt{\bar\phi} + \frac{\sqrt{\underline\phi}}{2}\Big)\normm{u}_2 \le \frac32\sqrt{\bar\phi}\,\normm{u}_2 .
    \]
\end{itemize}
Squaring both bounds, and recalling that $u'(\frac1n\hat X'\hat X)u = \normm{\hat Xu}_n^2$, the definitions of $\phi_{\min}$ and $\phi_{\max}$ give
\[
\phi_{\min}(m + d_0)\big[\tfrac1n\hat X'\hat X\big] \ge \frac{\underline\phi}{4},
\qquad
\phi_{\max}(m)\big[\tfrac1n\hat X'\hat X\big] \le \frac94\bar\phi .
\]
Let $\hat\Delta := \hat\gamma - \gamma_0$ and
\[
a := \alpha_0(F_p - \hat F_p) + R,
\qquad\text{so that}\qquad
Y = \hat F_p\alpha_0 + X\beta_0 + a + e = \hat X\gamma_0 + a + e .
\]
By the triangle inequality, $c_a := \normm{a}_n \le \abs{\alpha_0}c_p + \normm{R}_n = \bar c_a$. The estimator $\hat\gamma$ is the LASSO estimator of Lemma \ref{lem:lasso_approx} with design $\hat X$, target $\gamma_0$, tuning parameter $\lambda_\beta$, and unpenalized set $\mc U = \{1\}$ (the coefficient $\alpha$), so $P = \{2,\dots,d_x+1\}$.

\textbf{Step 1: Lemma \ref{lem:lasso_approx}.} We check its conditions with $(X, \beta_0, \lambda) = (\hat X, \gamma_0, \lambda_\beta)$, $S = S_1$ and $(\kappa, \nu) = (\kappa_1/4, \nu_1)$.
\begin{itemize}
    \item $S_1 \supseteq \{j : \gamma_{0,j} \neq 0\} \cup \mc U$ and $\abs{S_1} \le d_0$, by the definition of $S_1$.
    \item $\frac1n\hat X'\hat X$ satisfies the restricted eigenvalue condition over $S_1$ with parameters $(\kappa_1/4, \nu_1)$, by the first part of this proof.
    \item $\nu_0 \in (1, \nu_1)$, and the event $\mc{E}$ of Lemma \ref{lem:lasso_approx}, computed with $\hat X$ and $\lambda_\beta$, is exactly $\mc{E}_\beta$.
\end{itemize}
Hence, under $\mc{E}_\beta$, Lemma \ref{lem:lasso_approx}(i) with $\sqrt{\kappa} = \sqrt{\kappa_1}/2$ and $c_a \le \bar c_a$ gives
\begin{equation}\label{eq:endog_pred}
\normm{\hat X\hat\Delta}_n \le 2c_a + \frac{2\nu_0}{\nu_0+1}\cdot\frac{\sqrt{d_0}\,\lambda_\beta}{\sqrt{\kappa_1}/2} \le 2\bar c_a + \frac{4\nu_0}{\nu_0+1}\cdot\frac{\sqrt{d_0}\,\lambda_\beta}{\sqrt{\kappa_1}} .
\end{equation}

\textbf{Step 2: Lemma \ref{lem:lasso_sparse}.} Its first requirement, the conditions of Lemma \ref{lem:lasso_approx}, was verified in Step 1. It remains to check the two sparse eigenvalue requirements for $\frac1n\hat X'\hat X$, with $\underline\phi/4$ in place of $\underline\phi$.
\begin{itemize}
    \item $\phi_{\min}(m + d_0)[\frac1n\hat X'\hat X] \ge \underline\phi/4$, by Step 0 of this proof.
    \item $m > 2B\phi_{\max}(m)[\frac1n\hat X'\hat X]$. With $\kappa = \kappa_1/4$, the constant of Lemma \ref{lem:lasso_sparse} is
    \[
    B = (\nu_0+1)^2\Big(\frac{3c_a}{\lambda_\beta} + \frac{2\nu_0}{\nu_0+1}\cdot\frac{2\sqrt{d_0}}{\sqrt{\kappa_1}}\Big)^2 \le B_\beta ,
    \]
    because $B$ is increasing in $c_a$ and $c_a \le \bar c_a$. Together with $\phi_{\max}(m)[\frac1n\hat X'\hat X] \le \frac94\bar\phi$ from Step 0 of this proof,
    \[
    2B\,\phi_{\max}(m)\big[\tfrac1n\hat X'\hat X\big] \le 2B_\beta\cdot\frac94\bar\phi = \frac92\bar\phi B_\beta < m .
    \]
\end{itemize}
Hence, under $\mc{E}_\beta$, Lemma \ref{lem:lasso_sparse} applies. Part (i) shows that $\hat\Delta$ has at most $m + d_0$ nonzero entries, and part (ii) gives
\begin{equation}\label{eq:endog_l2}
\normm{\hat\Delta}_2 \le \frac{\normm{\hat X\hat\Delta}_n}{\sqrt{\underline\phi/4}} = \frac{2}{\sqrt{\underline\phi}}\normm{\hat X\hat\Delta}_n .
\end{equation}

\textbf{Step 3: conclusion.} Combining \eqref{eq:endog_l2} with \eqref{eq:endog_pred} gives, under $\mc{E}_\beta$,
\[
\normm{\hat\gamma - \gamma_0}_2 \le \frac{2}{\sqrt{\underline\phi}}\normm{\hat X\hat\Delta}_n \le \frac{2}{\sqrt{\underline\phi}}\Big(2\bar c_a + \frac{2\nu_0}{\nu_0+1}\cdot\frac{2\sqrt{d_0}\,\lambda_\beta}{\sqrt{\kappa_1}}\Big) \lesssim \bar c_a + \sqrt{d_0}\lambda_\beta . \qedhere
\]
\end{proof}

\section{Appendix: Proofs of Main Theorems}

\begin{proof}[Proof of Theorem \ref{thm:Consistency}]
    Proofs are written for scalar only to lighten notation; the vector case is identical.
    We prove the consistency of $\check{\theta}_1$ using Lemma \ref{lemma:Consistency}.
    Note that $\psi_t(\theta_1,f_{z},f_{u})=\psi(w_t,\theta_1,f_{z},f_{u})$. Define functions
    \begin{align*}
        Q_T(\theta_1)&:=-\E[\psi_t(\theta_1,f_{z0},f_{u0})]'W\E[\psi_t(\theta_1,f_{z0},f_{u0})]\\
        \hat{Q}_T(\theta_1)&:=-\E_{T,L}[\psi_t(\theta_1,\hat{f}_z^{l},\hat{f}_u^{l})]'\hat{W}\E_{T,L}[\psi_t(\theta_1,\hat{f}_z^{l},\hat{f}_u^{l})]
    \end{align*}
    Condition \ref{cond:IdentificationMoment} implies function $Q_T(\theta_1)$ is uniquely maximized at $\theta_1=\theta_{10}$. 
    To verify (ii) in Lemma \ref{lemma:Consistency}, note that 
    \begin{align*}
        Q_T(\theta_1)\leq & -\lambda_{min}(W)\cdot\normm{\E[\psi_t(\theta_1,f_{z0},f_{u0})]}_2^2\\
        \inf_{\theta_1:\normm{\theta_1-\theta_{10}}_2\geq \epsilon}(Q_T(\theta_{10})-Q_T(\theta_1))\geq & \lambda_{min}(W)\cdot\inf_{\theta_1:\normm{\theta_1-\theta_{10}}_2\geq \epsilon}\normm{\E[\psi_t(\theta_1,f_{z0},f_{u0})]}_2^2\\
        \geq & \lambda_{min}(W) \eta^2\\
        \geq & c \eta^2>0
    \end{align*}
    where $\lambda_{min}(W)>c>0$ holds by Condition \ref{cond:WeightMatrix} and $\forall$ $\epsilon>0$, there exists $\eta>0$ such that \\
    $\inf_{\normm{\theta_1-\theta_{10}}_2\geq \epsilon}\normm{\E[\psi_t(\theta_1,f_{z0},f_{u0})]}_2\geq \eta>0$ by Condition \ref{cond:IdentificationMoment}.
    Thus by Lemma \ref{lemma:Consistency}, what remains to be proved is that
    \[
        \sup_{\Theta_\sigma\times\Theta_\alpha}|\hat{Q}_T(\theta_1)-Q_T(\theta_1)|=o_p(1).
    \]
    It can be proved by the following statements.
    \begin{enumerate}
        \item $\normm{\hat{W}-W}_2=o_p(1)$
        \item $\sup_{\Theta_\sigma\times\Theta_\alpha}\normm{\E[\psi_t(\theta_1,f_{z0},f_{u0})]}_2\leq C,\quad\sup_{\Theta_\sigma\times\Theta_\alpha}\normm{\E_{T,L}[\psi_t(\theta_1,\hat{f}_z^{l},\hat{f}_u^{l})]}_2 =O_p(1)$
        \item $\sup_{\Theta_\sigma\times\Theta_\alpha}\normm{\E_{T,L}[\psi_t(\theta_1,\hat{f}_z^{l},\hat{f}_u^{l})]-\E\psi_t(\theta_1,f_{z0},f_{u0})}_2=o_p(1)$
    \end{enumerate}
    \textbf{\textit{Proof of the first statement:}} Condition \ref{cond:WeightMatrix} guarantees the first statement.

    \noindent\textbf{\textit{Proof of the second statement:}} Condition \ref{cond:Regularity1} and \ref{cond:BoundednessNO} imply that 
    \begin{align}
        \norm{\phi_j(w_t,\theta_1,f_{z0},f_{u0})}_2
=&\norm{\epsilon^z_{jt}(y_{jt}(\sigma)-\alpha p_{jt}-f_{u0}(x_{jt}))}_2\nonumber\\
\leq&\sqrt{d_z}\,\norm{\epsilon^z_{jt}}_\infty\Big(\sup_{\sigma\in\Theta_\sigma}\abs{y_{jt}(\sigma)}
+\sup_{\alpha\in\Theta_\alpha}\abs{\alpha}\abs{p_{jt}}+\norm{f_{u0}}_\infty\Big)
\leq C\label{eq:BoundedPhi}
    \end{align}
    Thus $\sup_{\Theta_\sigma\times\Theta_\alpha}\normm{\E[\psi_t(\theta_1,f_{z0},f_{u0})]}_2=\sup_{\Theta_\sigma\times\Theta_\alpha}\normm{\E[\E_J\phi_j(w_t,\theta_1,f_{z0},f_{u0})]}_2\leq C$.
    \begin{align*}
        \psi_t(\theta_1,\hat{f}_z^{l},\hat{f}_u^{l})&=\E_J(z_{jt}-\hat{f}_z^l(x_{jt}))\cdot (y_{jt}(\sigma)-\alpha p_{jt} - \hat{f}_u^{l}(x_{jt}))\\
        & = \E_J(\epsilon_{jt}^z-(\underbrace{\hat{f}_z^l(x_{jt})-f_{z0}(x_{jt})}_{\Delta_l f_z(x_{jt})}))\cdot 
        (y_{jt}(\sigma)-\alpha p_{jt} - f_{u0}(x_{jt}) - (\underbrace{\hat{f}_u^{l}(x_{jt})-f_{u0}(x_{jt})}_{\Delta_l f_{u}(x_{jt})}))\\
        & = \E_J(\epsilon_{jt}^z-\Delta_l f_z(x_{jt}))\cdot (y_{jt}(\sigma)-\alpha p_{jt} - f_{u0}(x_{jt}) - \Delta_l f_{u}(x_{jt}))
    \end{align*}
    Thus
    \begin{align*}
        &\normm{\E_{T,L}[\psi_t(\theta_1,\hat{f}_z^{l},\hat{f}_u^l)]}_2\\
        = &  \norm{\frac{1}{T}\sum_{l=1}^L\sum_{t\in I_l}
        \E_J(\epsilon_{jt}^z-\Delta_l f_z(x_{jt}))\cdot (y_{jt}(\sigma)-\alpha p_{jt} - f_{u0}(x_{jt}) - \Delta_l f_{u}(x_{jt}))}_2\\
        \leq &   \E_J\norm{\frac{1}{T}\sum_{l=1}^L\sum_{t\in I_l}
        (\epsilon_{jt}^z-\Delta_l f_z(x_{jt}))\cdot (y_{jt}(\sigma)-\alpha p_{jt} - f_{u0}(x_{jt}) - \Delta_l f_{u}(x_{jt}))}_2\\
        \leq &\E_J\norm{\E_T(\epsilon_{jt}^z)\cdot (y_{jt}(\sigma)-\alpha p_{jt} - f_{u0}(x_{jt}))}_2\\
         &+ \E_J\norm{\frac{1}{T}\sum_{l=1}^L\sum_{t\in I_l}\Delta_l f_z(x_{jt})\cdot (y_{jt}(\sigma)-\alpha p_{jt} - f_{u0}(x_{jt}))}_2\\
         & + \E_J\norm{\frac{1}{T}\sum_{l=1}^L\sum_{t\in I_l}(\epsilon_{jt}^z)\cdot \Delta_l f_{u}(x_{jt})}_2\\
         & + \E_J\norm{\frac{1}{T}\sum_{l=1}^L\sum_{t\in I_l}\Delta_l f_z(x_{jt})\cdot \Delta_l f_{u}(x_{jt})}_2\\
        \leq & C + C\E_{J}\frac{1}{T}\sum_{l=1}^L\sum_{t\in I_l}\norm{\Delta_l f_z(x_{jt})}_2 + C\E_{J}\frac{1}{T}\sum_{l=1}^L\sum_{t\in I_l}|\Delta_l f_{u}(x_{jt})| + 
        \E_J\norm{\frac{1}{T}\sum_{l=1}^L\sum_{t\in I_l}\Delta_l f_z(x_{jt})\cdot \Delta_l f_{u}(x_{jt})}_2\\
    \end{align*}
    $\implies$
    \begin{align*}
        &\sup_{\Theta_\sigma\times\Theta_\alpha}\normm{\E_{T,L}[\psi_t(\theta_1,\hat{f}_z^{l},\hat{f}_u^l)]}_2\\
         \leq & C + C\E_{J}\E_{T,L}\norm{\Delta_l f_z(x_{jt})}_2 + C\E_{J}\E_{T,L}|\Delta_l f_{u}(x_{jt})| + 
        \E_J\norm{\E_{T,L}\Delta_l f_z(x_{jt})\cdot \Delta_l f_{u}(x_{jt})}_2
    \end{align*}
    Then we can apply Lemma \ref{lemma:ConditionalPopulation}. 
    Let $a(w_t,\hat{\eta}_l) = \normm{\Delta_l f_z(x_{jt})}_2\cdot T^{1/4}$ for a fixed $j$ and $b(w_t,\hat{\eta}_l)=1$. 
    Then $\normm{a(w_t,\hat{\eta}_l)}_{L^2} = [\int \normm{\Delta_l f_z(x_{jt})}_2^2\cdot T^{1/2}dP(x_{jt})]^{1/2} = T^{1/4}\normm{\Delta_l f_z(x_{jt})}_{L^2}$.
    We have that $\normm{\Delta_l f_z(x_{jt})}_{L^2}=o_p(T^{-1/4})\implies \normm{a(w_t,\hat{\eta}_l)}_{L^2}=o_p(1)\implies \E_{T,L}\norm{a(w_t,\hat{\eta}_l)}_2=o_p(1)$, i.e., 
    \begin{equation}\label{eq:Deltafz}
        \E_{T,L}\norm{\Delta_l f_z(x_{jt})}_2=o_p(T^{-1/4})=o_p(1).
    \end{equation}
    Similarly let $a(w_t,\hat{\eta}_l) = |\Delta_l f_{u}(x_{jt})|\cdot T^{1/4}$, for a fixed $j$ and $b(w_t,\hat{\eta}_l)=1$, we have $\normm{a(w_t,\hat{\eta}_l)}_{L^2} = [\int |\Delta_l f_{u}(x_{jt})|^2\cdot T^{1/2}dP(x_{jt})]^{1/2} = T^{1/4}\normm{\Delta_l f_{u}(x_{jt})}_{L^2}=o_p(1)$.
    Thus
    \begin{equation}\label{eq:Deltafbeta}
        \E_{T,L}|\Delta_l f_{u}(x_{jt})|=o_p(T^{-1/4}) = o_p(1).
    \end{equation}
    Let $a(w_t,\hat{\eta}_l) = \normm{\Delta_l f_z(x_{jt})}_2\cdot T^{1/4}$ and $b(w_t,\hat{\eta}_l)=|\Delta_l f_{u}(x_{jt})|\cdot T^{1/4}$, we have $\normm{a(w_t,\hat{\eta}_l)}_{L^2}\cdot\normm{b(w_t,\hat{\eta}_l)}_{L^2}=o_p(1)$, which implies $\E_{T,L}\normm{a(w_t,\hat{\eta}_l)}_2\cdot\normm{b(w_t,\hat{\eta}_l)}_2=o_p(1)$, i.e.,
    \begin{equation}\label{eq:Deltafzfbeta}
        \E_J\normm{\E_{T,L}\Delta_l f_z(x_{jt})\cdot \Delta_l f_{u}(x_{jt})}_2\leq \E_J\E_{T,L}\normm{\Delta_l f_z(x_{jt})}_2\cdot |\Delta_l f_{u}(x_{jt})|=o_p(T^{-1/2})=o_p(1).
    \end{equation}
    Consequently, we have $\sup_{\Theta_\sigma\times\Theta_\alpha}\normm{\E_{T,L}[\psi_t(\theta_1,\hat{f}_z^{l},\hat{f}_u^l)]}_2=O_p(1)$.

    \noindent\textbf{\textit{Proof of the third statement:}} 
    It is sufficient to prove the following two statements.
    \[
    A=\sup_{\Theta_\sigma\times\Theta_\alpha}\norm{\E_{T,L}[\psi_t(\theta_1,\hat{f}_z^{l},\hat{f}_u^l)]-\E_T\psi_t(\theta_1,f_{z0},f_{u0})}_2=o_p(1).
    \]
    and
    \[
    B=\sup_{\Theta_\sigma\times\Theta_\alpha}\norm{\E_T\psi_t(\theta_1,f_{z0},f_{u0})-\E\psi_t(\theta_1,f_{z0},f_{u0})}_2=o_p(1).
    \]
    Note that similar to the proof of the second statement, we have
    \begin{align*}
        A\leq & \E_J\sup_{\Theta_\sigma\times\Theta_\alpha}\norm{\E_{T,L}\phi_j(w_{t},\theta_1,\hat{f}_z^l,\hat{f}_u^l)-\E_T\phi_j(w_{t},\theta_1,f_{z0},f_{u0})}_2\\
        = & \E_J\sup_{\Theta_\sigma\times\Theta_\alpha}\bigg\|
        \frac{1}{T}\sum_{l=1}^L\sum_{t\in I_l}(\epsilon_{jt}^z-\Delta_lf_z(x_{jt}))\cdot (y_{jt}(\sigma)-\alpha p_{jt} - f_{u0}(x_{jt}) - \Delta_l f_{u}(x_{jt}))\\
         & - \E_T(\epsilon_{jt}^z)\cdot (y_{jt}(\sigma)-\alpha p_{jt} - f_{u0}(x_{jt}))\bigg\|_2\\
        \leq & \E_J\sup_{\Theta_\sigma\times\Theta_\alpha}\bigg\| \frac{1}{T}\sum_{l=1}^L\sum_{t\in I_l}\Delta_l f_z(x_{jt})\cdot (y_{jt}(\sigma)-\alpha p_{jt} - f_{u0}(x_{jt}))\bigg\|_2\\
        & + \E_J\sup_{\Theta_\sigma\times\Theta_\alpha}\bigg\|\frac{1}{T}\sum_{l=1}^L\sum_{t\in I_l}(\epsilon_{jt}^z)\cdot \Delta_l f_{u}(x_{jt})\bigg\|_2\\
         & + \E_J\sup_{\Theta_\sigma\times\Theta_\alpha}\bigg\|\frac{1}{T}\sum_{l=1}^L\sum_{t\in I_l}\Delta_l f_z(x_{jt})\cdot \Delta_l f_{u}(x_{jt})\bigg\|_2\\
        \leq & C\E_J\E_{T,L}\normm{\Delta_l f_z(x_{jt})}_2 + C\E_J\E_{T,L}|\Delta_l f_{u}(x_{jt})| + \E_J\E_{T,L}\normm{\Delta_l f_z(x_{jt})\cdot \Delta_l f_{u}(x_{jt})}_2
    \end{align*}
    Equations \eqref{eq:Deltafz}, \eqref{eq:Deltafbeta} and \eqref{eq:Deltafzfbeta} guarantee that $A=o_p(1)$.

    Consider term $B$. Let $a(w_t,\theta_1):=\psi(w_t,\theta_1,f_{z0},f_{u0})
=\E_J\phi_j(w_t,\theta_1,f_{z0},f_{u0})$.
We verify the conditions of Lemma \ref{lemma:UniformConvergence} with $\Theta=\Theta_\sigma\times\Theta_\alpha$.
By equation \eqref{eq:BoundedPhi}, $\norm{a(w_t,\theta_1)}_2\leq C_1=:d(w_t)$,
so $\sup_T\E d(w_t)^2=C_1^2<\infty$.
By Lemma \ref{lemma:UniformBoundedness} and the mean value theorem,
$\abs{y_{jt}(\sigma)-y_{jt}(\sigma')}\leq C\sqrt{d_\sigma}\norm{\sigma-\sigma'}_2$ a.s.
Together with Condition 3.2(i),
\begin{align*}
    \norm{a(w_t,\theta_1)-a(w_t,\theta_1')}_2
    &\leq \E_J\norm{\epsilon^z_{jt}}_2\big(\abs{y_{jt}(\sigma)-y_{jt}(\sigma')}
      +\abs{p_{jt}}\,\abs{\alpha-\alpha'}\big)\\
    &\leq \sqrt{d_z}\,M\big(C\sqrt{d_\sigma}\norm{\sigma-\sigma'}_2+M\abs{\alpha-\alpha'}\big)
     \leq C_2\norm{\theta_1-\theta_1'}_2 ,
\end{align*}
where $C_1,C_2$ do not depend on $T$. Since $\Theta_\sigma\times\Theta_\alpha$ is compact,
of fixed dimension and does not depend on $T$ (Condition 3.2(ii)),
Lemma \ref{lemma:UniformConvergence} gives $B=o_p(1)$.
\end{proof}

\begin{proof}[Proof of Theorem \ref{thm:AsympNorm}]
    Take Taylor expansion at the point of $\theta_{10}$
    \begin{align*}
        0=&(\frac{\partial}{\partial \theta_1}\E_{T,L}[\psi_t(\check{\theta}_1,\hat{f}_z^{l},\hat{f}_u^l)]')\hat{W}(\E_{T,L}[\psi_t(\check{\theta}_1,\hat{f}_z^{l},\hat{f}_u^l)])\\
        -(\frac{\partial}{\partial \theta_1}\E_{T,L}[\psi_t(\check{\theta}_1,\hat{f}_z^{l},\hat{f}_u^l)]')\cdot\hat{W}\cdot\E_{T,L}[\psi_t(\theta_{10},\hat{f}_z^{l},\hat{f}_u^l)]=&
        \{(\frac{\partial}{\partial \theta_1}\E_{T,L}[\psi_t(\check{\theta}_1,\hat{f}_z^{l},\hat{f}_u^l)]')\cdot\\
        &\hat{W}\cdot\frac{\partial}{\partial \theta_1'}\E_{T,L}[\psi_t(\bar{\theta}_1,\hat{f}_z^{l},\hat{f}_u^l)]\}(\check{\theta}_1-\theta_{10})\\
        -\underbrace{(\frac{\partial}{\partial \theta_1}\E_{T,L}[\psi_t(\check{\theta}_1,\hat{f}_z^{l},\hat{f}_u^l)]')}_{A}\hat{W}\cdot\underbrace{\sqrt{T}\E_{T,L}[\psi_t(\theta_{10},\hat{f}_z^{l},\hat{f}_u^l)]}_{B}=&
        \{(\frac{\partial}{\partial \theta_1}\E_{T,L}[\psi_t(\check{\theta}_1,\hat{f}_z^{l},\hat{f}_u^l)]')\cdot\\
        &\hat{W}\cdot(\underbrace{\frac{\partial}{\partial \theta_1'}\E_{T,L}[\psi_t(\bar{\theta}_1,\hat{f}_z^{l},\hat{f}_u^l)]}_{C'})\}\sqrt{T}(\check{\theta}_1-\theta_{10})\\
        \sqrt{T}(\check{\theta}_1-\theta_{10})=&-\{A\hat{W}C'\}^{-1}A\hat{W}B
    \end{align*}
    where $\bar{\theta}$ lies on the line between $\theta_{10}$ and $\check{\theta}_1$. \footnote{
        Each row of $C'$ has different $\bar{\theta}$, but we can still consider it as a single $\bar{\theta}$ for notation simplicity.
    }
    Theorem \ref{thm:Consistency} guarantees the convergence of $\check{\theta}_1\to\theta_{10}$ in probability, and thus we have $\bar{\theta}_1\overset{p}{\to}\theta_{10}$. 
    We can prove the asymptotic normality of $\check{\theta}_1$ following the steps below.
    
    \textbf{Step 1:} Show that $A, B$ and $C$ have the following limits respectively.
    \begin{enumerate}[label=(\roman*)]
        \item $A-\underbrace{\E_J\E[
            \begin{bmatrix}
                \frac{\partial }{\partial \sigma}y_{jt}(\sigma_0)\\
                -p_{jt}
            \end{bmatrix}
            \cdot \epsilon^{z\prime}_{jt}]}_{G}=o_p(1)$
        \item $B-\sqrt{T}\E_{JT}\epsilon^{z}_{jt}\cdot \xi_{jt}=o_p(1)$.
        \item $\Omega^{-1/2}\sqrt{T}\E_{JT}\epsilon^{z}_{jt}\cdot \xi_{jt}\overset{d}{\to}N(0,I_{d_z})$, where $\Omega=Var(\E_J\epsilon^{z}_{jt}\xi_{jt})$.
        \item $C-\E_J\E[
            \begin{bmatrix}
                \frac{\partial }{\partial \sigma}y_{jt}(\sigma_0)\\
                -p_{jt}
            \end{bmatrix}
            \cdot \epsilon^{z\prime}_{jt}]=o_p(1)$
    \end{enumerate}
    We will prove (i), (ii) and (iii), as (iv) is similar to (i).
    Note here we suppress the dependence of $G$ and $W$ on $T$ for notation simplicity.\\
    \noindent\textbf{Proof of (i)}
    \begin{align*}
        A  =& \frac{\partial}{\partial \theta_1}\E_{T,L}[\psi_t(\check{\theta}_1,\hat{f}_z^{l},\hat{f}_u^l)]'\\
        =& \frac{\partial}{\partial \theta_1}\E_{T,L} \E_J\phi_j(w_{t},\check{\theta}_1,\hat{f}_z^{l},\hat{f}_u^l)'\\
        =& \frac{\partial}{\partial \theta_1}\E_{T,L} \E_J\bigg[(z_{jt}-\hat{f}_z^l(x_{jt}))\cdot (y_{jt}(\check{\sigma})-\check{\alpha} p_{jt} - \hat{f}_u^{l}(x_{jt}))\bigg]'\\
        =& \E_J\E_{T,L} 
        \frac{\partial}{\partial \theta_1}\bigg[(z_{jt}-\hat{f}_z^l(x_{jt}))\cdot (y_{jt}(\check{\sigma})-\check{\alpha} p_{jt} - \hat{f}_u^{l}(x_{jt}))\bigg]'\\
        =&\E_J\E_{T,L} \begin{bmatrix}
                \frac{\partial }{\partial \sigma}y_{jt}(\check{\sigma})\\
                -p_{jt}
        \end{bmatrix}\cdot(z_{jt}-\hat{f}_z^l(x_{jt}))'\\
         =&\E_J\E_{T,L}\begin{bmatrix}
                \frac{\partial }{\partial \sigma}y_{jt}(\check{\sigma})\\
                -p_{jt}
        \end{bmatrix}\cdot \epsilon^{z\prime}_{jt}-\E_J\E_{T,L}\begin{bmatrix}
                \frac{\partial }{\partial \sigma}y_{jt}(\check{\sigma})\\
                -p_{jt}
        \end{bmatrix}\cdot(\Delta_l f_z(x_{jt}))'\\
         =& \E_J\E_{T,L}\begin{bmatrix}
                \frac{\partial }{\partial \sigma}y_{jt}(\sigma_0)\\
                -p_{jt}
        \end{bmatrix}\epsilon^{z\prime}_{jt}+
            \E_J\E_{T,L}\begin{bmatrix}
                \frac{\partial }{\partial \sigma}y_{jt}(\check{\sigma})-\frac{\partial }{\partial \sigma}y_{jt}(\sigma_0)\\
                0
        \end{bmatrix}\epsilon^{z\prime}_{jt}
          & - \E_J\E_{T,L}\begin{bmatrix}
                \frac{\partial }{\partial \sigma}y_{jt}(\check{\sigma})\\
                -p_{jt}
        \end{bmatrix}(\Delta_l f_z(x_{jt}))'\\
        =: & A_1 + A_2 + A_3\\
    \end{align*}
    We will analyse the three terms separately.
    \begin{align*}
        A_1 = & \E_J\E_{T,L}\begin{bmatrix}
                \frac{\partial }{\partial \sigma}y_{jt}(\sigma_0)\\
                -p_{jt}
        \end{bmatrix}\epsilon^{z\prime}_{jt}\\
        = & \E_J\E_{T}\begin{bmatrix}
                \frac{\partial }{\partial \sigma}y_{jt}(\sigma_0)\\
                -p_{jt}
        \end{bmatrix}\epsilon^{z\prime}_{jt}\\
        = & \E_J \E[
            \begin{bmatrix}
                \frac{\partial }{\partial \sigma}y_{jt}(\sigma_0)\\
                -p_{jt}
            \end{bmatrix}
            \cdot \epsilon^{z\prime}_{jt}]+o_p(1)\\
        =& G +o_p(1)
    \end{align*}
    \begin{align*}
        A_2 =& \E_J\E_{T,L}\begin{bmatrix}
                \frac{\partial }{\partial \sigma}y_{jt}(\check{\sigma})-\frac{\partial }{\partial \sigma}y_{jt}(\sigma_0)\\
                0
        \end{bmatrix}\epsilon^{z\prime}_{jt}\\
        =& (\check{\sigma}-\sigma_0)\E_J\E_T\begin{bmatrix}
                \frac{\partial^2 }{\partial \sigma^2}y_{jt}(\bar{\sigma})\\
                0
        \end{bmatrix}\cdot\epsilon^{z\prime}_{jt} \\
        & \implies\\
        \normm{A_2}_2 \leq & C |\check{\sigma}-\sigma_0|\E_{JT}\normm{\epsilon^z_{jt}}_2 = o_p(1)\\
    \end{align*}
    The inequality for $A_2$ follows from Lemma~\ref{lemma:UniformBoundedness}
    \begin{align*}
        A_3 = & \E_J\E_{T,L}\begin{bmatrix}
                \frac{\partial }{\partial \sigma}y_{jt}(\check{\sigma})\\
                -p_{jt}
        \end{bmatrix}(\Delta_l f_z(x_{jt}))'\\
        \leq & C \E_J\E_{T,L} \normm{\Delta_l f_z(x_{jt})}_2 \\
        = & o_p(1) \text{    (by equation \ref{eq:Deltafz})}\\
    \end{align*}
    The inequality for $A_3$ is due to the uniform boundedness of $\frac{\partial}{\partial \sigma}y_{jt}(\sigma)$ and $p_{jt}$.

    \noindent\textbf{Proof of (ii)}:
    \begin{align*}
        B-\sqrt{T}\E_{JT}\epsilon^{z}_{jt}\cdot \xi_{jt}=&\sqrt{T}\E_{T,L}[\psi_t(\theta_{10},\hat{f}_z^{l},\hat{f}_u^{l})]-\sqrt{T}\E_{JT}\epsilon^{z}_{jt}\cdot \xi_{jt}\\
        =&\sqrt{T}\E_{T,L}\{\E_J[(z_{jt}-\hat{f}^l_z(x_{jt}))\cdot(y_{jt}(\sigma_0)-\alpha_0p_{jt}-\hat{f}_u^{l}(x_{jt}))\}\\
        &-\sqrt{T}\E_{T}\{\E_J[\epsilon_{jt}^z\cdot\xi_{jt}]\}\\
        =&\sqrt{T}\E_{T,L}\{\E_J[(\epsilon^z_{jt}-\Delta_l f_{z}(x_{jt}))\cdot(\xi_{jt}-\Delta_l f_{u}(x_{jt}))]\}\\
        &-\sqrt{T}\E_{T}\{\E_J[\epsilon_{jt}^z\cdot\xi_{jt}]\}\\
        =&\sqrt{T}\E_{T,L}\E_J[-\Delta_l f_{z}(x_{jt})\cdot\xi_{jt}-\epsilon_{jt}^z\cdot \Delta_lf_{u}(x_{jt})+\Delta_l f_{z}(x_{jt})\cdot\Delta_l f_{u}(x_{jt})]
    \end{align*}
    Equation (\ref{eq:Deltafzfbeta}) implies that $\sqrt{T}\E_{T,L}\E_J\Delta_l f_{z}(x_{jt})\cdot\Delta_l f_{u}(x_{jt})=o_p(1)$.\\
    Then we prove $\sqrt{T}\E_{T,L}\E_J[\Delta_l f_{z}(x_{jt})\cdot\xi_{jt}] = o_p(1)$.
    Let $a(w_t,\hat{\eta}_l) = \Delta_l\hat{f}_{z}(x_{jt})\cdot\xi_{jt}\cdot\sqrt{T}$, then for any $l=1,\dots,L$ and fixed $j$, we have,
    \begin{align*}
        &\normm{\frac{1}{T_l}\sum_{t\in I_l}a(w_t,\hat{\eta}_l)}_{L^2}\\
        =&\sqrt{T}\normm{\frac{1}{T_l}\sum_{t\in I_l}\Delta_l f_{z}(x_{jt})\cdot\xi_{jt}}_{L^2}\\
        =&\sqrt{T}\left\{\int \normm{\frac{1}{T_l}\sum_{t\in I_l}\Delta_l f_{z}(x_{jt})\cdot\xi_{jt}}_2^2 \cdot P(\prod_{t\in I_l}dx_{jt}d\xi_{jt})\right\}^{1/2}\\
        =&\sqrt{T}\left\{\int\frac{1}{T_l}\sum_{t\in I_l}(\Delta_l f_{z}(x_{jt}))'\cdot\xi_{jt}\cdot\frac{1}{T_l}\sum_{t\in I_l}\Delta_l f_{z}(x_{jt})\cdot\xi_{jt}P(\prod_{t\in I_l}dx_{jt}d\xi_{jt})\right\}^{1/2}\\
        =&\sqrt{T}\left\{\int\frac{1}{T_l^2}\sum_{t_1=t_2\in I_l}(\Delta_l f_{z}(x_{jt_1}))'\cdot\xi_{jt_1}\cdot\Delta_l f_{z}(x_{jt_2})\cdot\xi_{jt_2} P(\prod_{t\in I_l}dx_{jt}d\xi_{jt})\right\}^{1/2}\\
        & \text{(the cross terms with $t_1\neq t_2$ vanish because $w_{t_1}\perp w_{t_2}$, $t\in I_l$ implies $w_t\perp I_l^c$ and }\E[\xi_t|x_t]=0)\\
        =&\sqrt{T}\left\{\frac{1}{T_l^2}\sum_{t\in I_l} \int (\Delta_l f_{z}(x_{jt}))'\cdot\xi_{jt}\cdot\Delta_l f_{z}(x_{jt})\cdot\xi_{jt}P(dx_{jt}d\xi_{jt})\right\}^{1/2}\\
        =&\sqrt{T}\left\{\frac{1}{T_l^2} \cdot T_l\int \normm{\Delta_l f_{z}(x_{jt})}_2^2\cdot\xi_{jt}^2 P(dx_{jt}d\xi_{jt})\right\}^{1/2}\\
        \leq & C\sqrt{T}\left\{\frac{1}{T_l} \int \normm{\Delta_l f_{z}(x_{jt})}_2^2 P(dx_{jt})\right\}^{1/2}\\
        & \text{(the inequality holds due to the boundedness of $\E[\xi_{jt}^2|x_{jt}]$)}\\
         =& C\sqrt{T/T_l}\normm{\Delta_l f_{z}(x_{jt})}_{L^2}\\
         \lesssim & \normm{\Delta_l f_{z}(x_{jt})}_{L^2}\\
         = & o_p(1)
    \end{align*}
    where $(\Delta_l f_{z}(x_{jt}))'$ is the transpose of $\Delta_l f_{z}(x_{jt})$.
    By Lemma \ref{lemma:ConditionalPopulation}, we have 
    \begin{equation}\label{eq:Deltafbetai}
            \sqrt{T}\E_{T,L}\Delta_l f_{z}(x_{jt})\cdot\xi_{jt}= o_p(1).
    \end{equation}
    Similarly, we can show that 
    \begin{equation}\label{eq:Deltafbetaxi}
        \sqrt{T}\E_{T,L}\epsilon_{jt}^z\cdot \Delta_l f_{u}(x_{jt})=o_p(1)
    \end{equation}
    Thus equations (\ref{eq:Deltafzfbeta}, \ref{eq:Deltafbetai}) and (\ref{eq:Deltafbetaxi}) guarantee that $B=\sqrt{T}\E_{JT}\epsilon^{z}_{jt}\cdot \xi_{jt}+o_p(1)$.
    
    \noindent\textbf{Proof of (iii)}:
    By Multivariate Lindeberg-Feller CLT for triangular arrays (theorem 1.5 in \citet{berckmoes2016stein}), we only need to verify the following statement
    \begin{align*}
        \forall\,\epsilon>0, 
        \sum_{t=1}^T\E[\normm{\frac{1}{\sqrt{T}}\Omega^{-1/2}\E_J\epsilon^{z}_{jt}\cdot \xi_{jt}}_2^2\cdot 1\{\normm{\frac{1}{\sqrt{T}}\Omega^{-1/2}\E_J\epsilon^{z}_{jt}\cdot \xi_{jt}}_2^2>\epsilon\}]\to 0
    \end{align*}
    This holds since 
    \begin{align*}
        &\sum_{t=1}^T\E[\normm{\frac{1}{\sqrt{T}}\Omega^{-1/2}\E_J\epsilon^{z}_{jt}\cdot \xi_{jt}}_2^2\cdot 1\{\normm{\frac{1}{\sqrt{T}}\Omega^{-1/2}\E_J\epsilon^{z}_{jt}\cdot \xi_{jt}}_2^2>\epsilon\}]\\
        =&\frac{1}{T}\sum_{t=1}^T\E[\normm{\Omega^{-1/2}\E_J\epsilon^{z}_{jt}\cdot \xi_{jt}}_2^2\cdot 1\{\normm{\Omega^{-1/2}\E_J\epsilon^{z}_{jt}\cdot \xi_{jt}}_2^2>\epsilon T\}]\\
        =& \E[\normm{\Omega^{-1/2}\E_J\epsilon^{z}_{jt}\cdot \xi_{jt}}_2^2\cdot 1\{\normm{\Omega^{-1/2}\E_J\epsilon^{z}_{jt}\cdot \xi_{jt}}_2^2>\epsilon T\}]\\
        \leq & c^{-1}\E[\normm{\E_J\epsilon^{z}_{jt}\cdot \xi_{jt}}_2^2\cdot 1\{\normm{\E_J\epsilon^{z}_{jt}\cdot \xi_{jt}}_2^2>c\epsilon T\}]\\
        \to  & 0 \text{  (convergence holds due to the boundedness of $\normm{\E_J\epsilon^{z}_{jt}\cdot \xi_{jt}}_2^2$)}
    \end{align*}

    \textbf{Step 2:} Note we have the following statements
    \begin{enumerate}[label=(\alph*)]
        \item boundedness of $\normm{\Omega}_2$, $\normm{W}_2$, $\normm{\hat{W}}_2$, $\normm{G}_2$, $\normm{A}_2$, $\normm{B}_2$ and $\normm{C}_2$.
        \item $A\hat{W}C' - GWG' = o_p(1)$ and $(A\hat{W}C')^{-1} - (GWG')^{-1} = o_p(1)$.
    \end{enumerate} 
    The boundedness of $\normm{\Omega}_2$ is guaranteed by boundedness of $\normm{\epsilon_{jt}^z}_\infty$ and $\normm{\xi_{jt}}_\infty$. 
    Boundedness of $\normm{W}_2$ is assumed in Condition \ref{cond:WeightMatrix}. 
    The remaining matrices are bounded because of the following reasons.
    \begin{align*}
        \normm{\hat{W}}_2\leq& \normm{\hat{W}-W}_2 + \normm{W}_2 \\
        = & o_p(1) +O(1) = O_p(1)\\
        \normm{G}_2=&\norm{\E_J\E[
            \begin{bmatrix}
                \frac{\partial }{\partial \sigma}y_{jt}(\sigma_0)\\
                -p_{jt}
            \end{bmatrix}
            \cdot \epsilon^{z\prime}_{jt}]}_2\\
            \leq& \E_J\norm{\E[
            \begin{bmatrix}
                \frac{\partial }{\partial \sigma}y_{jt}(\sigma_0)\\
                -p_{jt}
            \end{bmatrix}
            \cdot \epsilon^{z\prime}_{jt}]}_2\\
            \leq & C\E_J(\norm{\frac{\partial }{\partial \sigma}y_{jt}(\sigma_0)}_\infty + \normm{p_{jt}}_\infty+\normm{\epsilon^z_{jt}}_\infty)\\
            \leq & C\\
            \normm{A}_2\leq& \normm{A-G}_2 + \normm{G}_2 = O_p(1)\\
            \normm{B}_2\leq& \normm{B-\sqrt{T}\E_{JT}\epsilon^{z}_{jt}\cdot \xi_{jt}}_2 + \normm{\sqrt{T}\E_{JT}\epsilon^{z}_{jt}\cdot \xi_{jt}}_2\\
            =& o_p(1) + \normm{\Omega^{1/2}}_2\cdot\normm{\Omega^{-1/2}\sqrt{T}\E_{JT}\epsilon^{z}_{jt}\cdot \xi_{jt}}_2\\
            =& o_p(1) + O(1)\cdot O_p(1) = O_p(1)\\
            \normm{C}_2\leq& \normm{C-G}_2 + \normm{G}_2 = O_p(1)\\
    \end{align*}
    Then the multiplication of the matrices $A\hat{W}C'$ is close to $GWG'$,
    \begin{align*}
        &A\hat{W}C' - GWG' \\
        =& (A-G)\hat{W}C' + G\hat{W}(C'-G') + G\hat{W}G' - GWG'\\
        =& (A-G)\hat{W}C' + G\hat{W}(C'-G')  + G(\hat{W}-W)G'\\
        \leq& \normm{A-G}_2\normm{\hat{W}}_2\normm{C}_2 + \normm{C-G}_2\normm{\hat{W}}_2\normm{G}_2  + \normm{G}_2^2\normm{\hat{W}-W}_2\\
        =& o_p(1)
    \end{align*}
    Thus the inverse of $A\hat{W}C'$ is close to the inverse of $GWG'$,
    \begin{align*}
        &\normm{(A\hat{W}C')^{-1} - (GWG')^{-1}}_2\\
        =&\ \normm{(A\hat{W}C')^{-1}(GWG' - A\hat{W}C')(GWG')^{-1}}_2\\
        \leq&\ \normm{(A\hat{W}C')^{-1}}_2\,\normm{GWG' - A\hat{W}C'}_2\,\normm{(GWG')^{-1}}_2\\
        \leq&\ s_{1}\big((A\hat{W}C')^{-1}\big)\sqrt{d_{\theta_1}}\cdot o_p(1)\cdot \lambda_{\max}\big((GWG')^{-1}\big)\sqrt{d_{\theta_1}}
            &&\text{($GWG'$ is positive definite)}\\
        =&\ \big(s_{d_{\theta_1}}(A\hat{W}C')\big)^{-1}\sqrt{d_{\theta_1}}\cdot o_p(1)\cdot \big(\lambda_{\min}(GWG')\big)^{-1}\sqrt{d_{\theta_1}}
            &&\text{($A\hat{W}C'\in\R^{d_{\theta_1}\times d_{\theta_1}}$)}\\
        \leq&\ d_{\theta_1}\,c^{-1}\Big(s_{d_{\theta_1}}(GWG')-\normm{A\hat{W}C'-GWG'}_{op}\Big)^{-1}\cdot o_p(1)\\
        \leq&\ d_{\theta_1}\,c^{-1}\Big(s_{d_{\theta_1}}(GWG')-\normm{A\hat{W}C'-GWG'}_{2}\Big)^{-1}\cdot o_p(1)\\
        \leq&\ d_{\theta_1}\,c^{-1}\big(c-o_p(1)\big)^{-1}\cdot o_p(1)\\
        =&\ o_p(1)
    \end{align*}
    Here we utilized the property that $\normm{M}_2\leq\normm{M}_{op}\sqrt{rank(M)}=s_1(M)\sqrt{rank(M)}$ for any matrix $M$, where $s_1(M)\geq s_2(M)\cdots\geq s_{\min(m,n)}(M)\geq 0$ are all singular value of $M\in\R^{m\times n}$, and Mirsky's singular value inequality (Weyl's inequality for singular values) $|s_{d_{\theta_1}}(M)-s_{d_{\theta_1}}(N)|\leq \normm{M-N}_{op}$ for any two matrices $M$ and $N$.

    \noindent\textbf{Step 3:} Then we can prove the asymptotic normality of $\check{\theta}_1$.
    \begin{align*}
        \sqrt{T}(\check{\theta}_1-\theta_{10}) = & -(A\hat{W}C')^{-1}A\hat{W}B\\
        \sqrt{T}(\check{\theta}_1-\theta_{10}) = & -(GWG')^{-1}A\hat{W}B +o_p(1)\quad \text{as }B=O_p(1)\\
        \sqrt{T}(\check{\theta}_1-\theta_{10}) = & -(GWG')^{-1}GWB +o_p(1)\\
        \sqrt{T}(\check{\theta}_1-\theta_{10}) = & -(GWG')^{-1}GW\Omega^{1/2}(\Omega^{-1/2}\sqrt{T}\E_{JT}\epsilon^{z}_{jt}\cdot \xi_{jt}) +o_p(1)
    \end{align*}
    Define a symmetric matrix $P^{1/2}\in\R^{d_{\theta_1}\times d_{\theta_1}}$ as the symmetric matrix s.t $P^{1/2}\cdot P^{1/2} = (GWG')^{-1}GW\Omega WG'(GWG')^{-1}$, then we have property that 
    \begin{align*}
        \normm{P^{-1/2}}_2 \leq & \sqrt{d_{\theta_1}}\lambda_{max}(P^{-1/2})\\
        = & \sqrt{d_{\theta_1}} (\lambda_{min}(P^{1/2}))^{-1}\\
         = & \sqrt{d_{\theta_1}}(\lambda_{min}(P))^{-1/2} \\
         = & \sqrt{d_{\theta_1}}[\lambda_{min}((GWG')^{-1}GW\Omega WG'(GWG')^{-1})]^{-1/2}\\
         \lesssim & [\lambda_{min}((GWG')^{-1}G WG'(GWG')^{-1})]^{-1/2} \quad \text{(the inequality holds since $\lambda_{min}(\Omega),\lambda_{min}(W)> c $)}\\
         = & [\lambda_{min}((GWG')^{-1})]^{-1/2}\\
         = & [\lambda_{max}(GWG')]^{1/2} \\
         \leq & (\normm{G}_2^2\normm{W}_2)^{1/2} \\
            = & O(1)
    \end{align*}
    Consequently, we have
    \begin{align*}
        P^{-1/2}\sqrt{T}(\check{\theta}_1-\theta_{10}) =& \underbrace{P^{-1/2}(GWG')^{-1}GW\Omega^{1/2}}_{\tilde{P}\in\R^{d_{\theta_1}\times d_{z}}}(\Omega^{-1/2}\sqrt{T}\E_{JT}\epsilon^{z}_{jt}\cdot \xi_{jt}) +o_p(1)
    \end{align*}
    where $\tilde{P}\cdot\tilde{P}' = P^{-1/2}\cdot P\cdot P^{-1/2}=I_{d_{\theta_1}}$ since $P^{1/2}P^{1/2} = (GWG')^{-1}GW\Omega WG'(GWG')^{-1}$.
    Then apply multivariate Lindeberg-Feller CLT, we need to prove that for any $\epsilon>0$,
    \begin{align*}
        &\sum_{i=1}^T\E[\normm{\frac{1}{\sqrt{T}}\tilde{P}(\Omega^{-1/2}\E_J\epsilon^{z}_{jt}\cdot \xi_{jt})}_2^2\cdot 1\{\normm{\frac{1}{\sqrt{T}}\tilde{P}(\Omega^{-1/2}\E_J\epsilon^{z}_{jt}\cdot \xi_{jt})}_2^2>\epsilon \}] \to 0\\
        &\Longleftrightarrow \\
        &\E[\normm{\tilde{P}(\Omega^{-1/2}\E_J\epsilon^{z}_{jt}\cdot \xi_{jt})}_2^2\cdot 1\{\normm{\tilde{P}(\Omega^{-1/2}\E_J\epsilon^{z}_{jt}\cdot \xi_{jt})}_2^2>\epsilon T \}] \to 0
    \end{align*}
    Note that $\normm{\tilde{P}}_2=\sqrt{tr(\tilde{P}\tilde{P}')}=\sqrt{tr(I_{d_{\theta_1}})}=\sqrt{d_{\theta_1}}$, we have
    \begin{align*}
        \normm{\tilde{P}(\Omega^{-1/2}\E_J\epsilon^{z}_{jt}\cdot \xi_{jt})}_2 \leq &\normm{\tilde{P}}_2\normm{\Omega^{-1/2}}_2\normm{\E_J\epsilon^{z}_{jt}\cdot \xi_{jt}}_2\le  \sqrt{d_{\theta_1}}\cdot \sqrt{d_{z}}\normm{\Omega^{-1/2}}_{op}\cdot\normm{\E_J\epsilon^{z}_{jt}\cdot \xi_{jt}}_2\leq C
    \end{align*}
    Thus the Lindeberg-Feller condition is satisfied and we have that
    \begin{align*}
        P^{-1/2}\sqrt{T}(\check{\theta}_1-\theta_{10}) \overset{d}{\to} & N(0, I_{d_{\theta_1}})
    \end{align*}
\end{proof}

\begin{proof}[Proof of Theorem \ref{thm:Nuisance}]
Fix $l \in \{1,\dots,L\}$. All nuisance estimators for fold $l$ are computed from the observations in $I_l^c$; let $n := J\abs{I_l^c}$, which satisfies $n \asymp T$ because $L$ is fixed. Let $L_T := \log(d_x\vee T)$, and note that $\log(d_{\tilde z}\vee T) \lesssim L_T$ because $d_{\tilde z} = O(d_x)$.
To lighten notation we drop the superscript $l$. 
By Lemma \ref{lem:SE_RE2}, there is an event of probability approaching one on which the conclusions (i)--(iii) of Lemma \ref{lem:SE_RE2} hold for $\hat\Sigma_l^{(1)}$ and $\hat\Sigma_l^{(2)}$. Throughout we work on this event. We fix $\nu_1, \nu_2 > 1$, set $\nu_0 := (1 + \min(\nu_1,\nu_2))/2$, let $m := (\ell_T - d_z - 1)d_0$, and write $\kappa_1 = \kappa_2 = \underline\phi/9$.

We will repeatedly use the following fact: if $X_T = O_p(a_T)$ and $g_T \to \infty$, then $X_T \le g_Ta_T$ w.p.a.1. In particular, by Markov's inequality and (15)--(17),
\begin{equation}\label{eq:approx_err_rates}
\normm{r_i^z}_n,\ \normm{r^p}_n,\ \normm{r^u}_n = O_p\big(\sqrt{d_0/T}\big) \le \sqrt{\frac{d_0L_T}{T}} \quad \text{w.p.a.1},
\end{equation}
where $r_i^z$, $r^p$ and $r^u$ are the vectors of $r_{jt,i}^z$, $r_{jt}^p$ and $r_{jt}^u$ over $t \in I_l^c$.

We complete the proof in 3 steps.
\begin{enumerate}[label=Step \arabic*:]
    \item Prove that
    \[
        \normm{\hat{f}_z(x_{jt})-f_{z0}(x_{jt})}_{L^2} \lesssim_P \sqrt{\frac{d_0L_T}{T}},
    \]
    and that
    \[
        c_p := \Big[\frac1n\sum_{j,t}\big(\hat{f}_p(x_{jt},z_{jt})-f_{p0}(x_{jt},z_{jt})\big)^2\Big]^{1/2} \lesssim_P \sqrt{\frac{d_0L_T}{T}},
        \qquad
        \normm{\hat\Delta_p}_1 \lesssim_P d_0\sqrt{\frac{L_T}{T}} .
    \]
    \item Prove the convergence rate of $\tilde{\sigma}$:
    \[
        |\tilde{\sigma}-\sigma_0|\lesssim_P \sqrt{\frac{d_0L_T}{T}}.
    \]
    \item Prove that
    \[
        \normm{\hat{f}_u(x_{jt})-f_{u0}(x_{jt})}_{L^2} \lesssim_P \sqrt{\frac{d_0L_T}{T}} .
    \]
\end{enumerate}
Since $\sqrt{d_0L_T/T} = o(T^{-1/4})$ (Condition \ref{cond:ApproxSparsity}(ii)) implies $d_0\sqrt{L_T/T} \to 0$, Steps 1 and 3 give the theorem.

\textbf{Step 1: Estimation of $\hat f_z$ and $\hat f_p$.}
Let $X$ be the matrix of all $x_{jt}'$ and $\tilde{X}_p$ the matrix of all $(z_{jt}',x_{jt}')$, so that $\frac1n\tilde X_p'\tilde X_p = \hat\Sigma_l^{(2)}$, and $\frac1nX'X$ is its $x$-block. Let $S_{z,i} := \{j : \Pi_{0,ij} \neq 0\}$ and $S_p := \{j : \beta_{p0,j} \neq 0\}$, where $\beta_{p0}=(\beta_{pz0}',\beta_{px0}')'$.
We first prove that $\hat \Sigma_{l}^{(2)}$ satisfies the restricted eigenvalue condition with parameters $(\kappa_2,\nu_2)$ over $S_p$ and that $\frac1nX'X$ satisfies the restricted eigenvalue condition over $S_{z,i}$ with parameters $(\kappa_2,\nu_2)$.

By \eqref{eq:highdimBLP_approx_z} and \eqref{eq:highdimBLP_approx_p}, $\abs{S_{z,i}} \le d_0$ and $\abs{S_p} \le d_0$. Hence Lemma \ref{lem:SE_RE2}(iii) gives that $\hat\Sigma_l^{(2)}$ satisfies the restricted eigenvalue condition with parameters $(\kappa_2,\nu_2)$ over $S_p$ and over $\{d_z + j : j \in S_{z,i}\}$. The quadratic form of $\frac1nX'X$ at $\Delta \in \R^{d_x}$ equals that of $\hat\Sigma_l^{(2)}$ at $(0',\Delta')'$, and $\Delta \in \mc C_{\nu_2}(S_{z,i})$ if and only if $(0',\Delta')' \in \mc C_{\nu_2}(\{d_z + j : j \in S_{z,i}\})$. Therefore $\frac1nX'X$ satisfies the restricted eigenvalue condition over $S_{z,i}$ with parameters $(\kappa_2,\nu_2)$. 

By Lemma \ref{lem:SE_RE2}(i), with $m=(\ell_T-d_z-1)d_0$ and the monotonicity of sparse eigenvalues, using
    $m\le m+d_0= (\ell_T-d_z)d_0\le \ell_Td_0-d_z$, both $\mathcal M =\hat\Sigma_l^{(2)}$ and $\frac1nX'X$ satisfy
\begin{equation}\label{eq:step1_SE}
\phi_{\min}(m + d_0)[\mathcal M] \ge \underline\phi,
\qquad
\phi_{\max}(m)[\mathcal M] \le \bar\phi.
\end{equation}

\emph{The estimator $\hat f_z$.} Fix $i \in \{1,\dots,d_z\}$, and let $Z_i$, $r_i^z$ and $\epsilon_i^z$ be the vectors of all $z_{jt,i}$, $r^z_{jt,i}$ and $\epsilon^z_{jt,i}$. By (16),
\[
Z_i = X\Pi_{0,i}' + r_i^z + \epsilon_i^z ,
\]
which is the model of Lemma \ref{lem:lasso_approx} with $(\beta_0, a, e) = (\Pi_{0,i}', r_i^z, \epsilon_i^z)$, $\mc U = \emptyset$, $S = S_{z,i}$, $(\kappa,\nu) = (\kappa_2,\nu_2)$ and $\lambda = \lambda_\Pi$. We verify its event and rates.
\begin{itemize}
    \item Lemma \ref{lemma:Hoeffding} gives $\normm{X'\epsilon_i^z/n}_\infty \lesssim_P \sqrt{L_T/T}$. With $C_\Pi$ large enough, the event $\mc{E}$ of Lemma \ref{lem:lasso_approx} holds w.p.a.1.
    \item \emph{Approximation error.} By \eqref{eq:approx_err_rates}, $c_a = \normm{r_i^z}_n \lesssim_P \sqrt{d_0L_T/T}$, so $c_a/\lambda_\Pi \lesssim_P \sqrt{d_0}$. Hence the constant $B$ of Lemma \ref{lem:lasso_sparse} satisfies $B \le  C_Bd_0$ w.p.a.1.
    \item \emph{Sparse eigenvalues.} By \eqref{eq:step1_SE} and $\ell_T \to \infty$, $m = (\ell_T - d_z - 1)d_0 > 2C_Bd_0\bar\phi \ge 2B\phi_{\max}(m)$ for all $T$ large enough, and $\phi_{\min}(m + d_0) \ge \underline\phi$.
\end{itemize}
Lemma \ref{lem:lasso_sparse}(ii) therefore gives, w.p.a.1,
\[
\normm{\hat{\Pi}_i-\Pi_{0,i}}_2 \lesssim c_a + \sqrt{d_0}\lambda_\Pi \lesssim \sqrt{\frac{d_0L_T}{T}} .
\]
Since $\hat\Pi$ is computed from $I_l^c$ and $x_{jt}$ with $t \in I_l$ is independent of it, conditional on the estimation sample,
\begin{align*}
    \normm{\hat{f}_z(x_{jt})-f_{z0}(x_{jt})}_{L^2}& \leq \normm{\hat{f}_z(x_{jt})-\Pi_0 x_{jt} }_{L^2} + \sqrt{\E \normm{r_{jt}^z}_2^2}\\
    & \leq \sum_{i=1}^{d_z}\sqrt{\E\big[\big((\hat{\Pi}_i-\Pi_{0,i})x_{jt}\big)^2\big]} + \sqrt{\E \normm{r_{jt}^z}_2^2}\\
    &\le \sqrt{\lambda_{\max}(\E x_{jt}x_{jt}')}\sum_{i=1}^{d_z} \normm{\hat{\Pi}_i-\Pi_{0,i}}_2 + c_z\\
    &\lesssim_P \sqrt{\frac{d_0L_T}{T}},
\end{align*}
where we used Condition \ref{cond:Regularity2}(ii) ($\lambda_{\max}(\E x_{jt}x_{jt}') \le M$), $d_z$ fixed, and $c_z \lesssim \sqrt{d_0/T}$.

\emph{The estimator $\hat f_p$.} Let $P$ be the vector of all $p_{jt}$, $\hat{\beta}_p=(\hat{\beta}_{pz}',\hat{\beta}_{px}')'$, $\beta_{p0}=(\beta_{pz0}',\beta_{px0}')'$, $\hat{\Delta}_p = \hat{\beta}_p-\beta_{p0}$, and let $r^p$ and $\epsilon^p$ be the vectors of all $r^p_{jt}$ and $\epsilon^p_{jt}$. By \eqref{eq:highdimBLP_price}, and \eqref{eq:highdimBLP_approx_p},
\[
P = \tilde X_p\beta_{p0} + r^p + \epsilon^p ,
\]
which is the model of Lemma \ref{lem:lasso_approx} with $(X, \beta_0, a, e) = (\tilde X_p, \beta_{p0}, r^p, \epsilon^p)$, $\mc U = \emptyset$, $S = S_p$, $(\kappa,\nu) = (\kappa_2,\nu_2)$ and $\lambda = \lambda_{\beta_{pz},\beta_{px}}$.
\begin{itemize}
    \item Lemma \ref{lemma:Hoeffding} gives $\normm{\tilde X_p'\epsilon^p/n}_\infty \lesssim_P \sqrt{L_T/T}$, so with $C_{\beta_{pz},\beta_{px}}$ large enough, the event $\mc{E}$ holds w.p.a.1.
    \item \emph{Approximation error.} By \eqref{eq:approx_err_rates}, $c_{r,p} := \normm{r^p}_n \lesssim_P \sqrt{d_0L_T/T}$.
\end{itemize}
Lemma \ref{lem:lasso_approx}(i)--(ii) therefore gives,
\begin{align*}
\normm{\tilde X_p\hat\Delta_p}_n &\lesssim_P c_{r,p} + \sqrt{d_0}\lambda_{\beta_{pz},\beta_{px}} \lesssim_P \sqrt{\frac{d_0L_T}{T}},\\
\normm{\hat\Delta_p}_1 &\lesssim_P \sqrt{d_0}\,c_{r,p} + d_0\lambda_{\beta_{pz},\beta_{px}} + \frac{c_{r,p}^2}{\lambda_{\beta_{pz},\beta_{px}}}
\lesssim_P d_0\sqrt{\frac{L_T}{T}} + d_0\sqrt{\frac{L_T}{T}} + \frac{d_0L_T/T}{\sqrt{L_T/T}} \lesssim d_0\sqrt{\frac{L_T}{T}} .
\end{align*}
Finally, since $\hat f_p(x_{jt},z_{jt}) - f_{p0}(x_{jt},z_{jt}) = w_{jt}'\hat\Delta_p - r^p_{jt}$, the triangle inequality gives
\[
c_p = \normm{\tilde X_p\hat\Delta_p - r^p}_n \le \normm{\tilde X_p\hat\Delta_p}_n + c_{r,p} \lesssim_P \sqrt{\frac{d_0L_T}{T}}.
\]
This completes Step 1.

\textbf{Step 2: Convergence rate of $\tilde{\sigma}$.}
Throughout this step, $\E_n$ denotes the average over the estimation sample, $\E_n[v_{jt}] := \frac1n\sum_{j=1}^J\sum_{t\in I_l^c}v_{jt}$, and $\E$ denotes $\E\,\E_J$. Thus $\hat g(\theta) = \E_n[\tilde z_{jt}(y_{jt}(\sigma) - \alpha p_{jt} - x_{jt}'\beta)]$ and $g(\theta) = \E[\tilde z_{jt}(y_{jt}(\sigma) - \alpha p_{jt} - x_{jt}'\beta)]$. Write $\tilde L_T := \log(d_{\tilde z}\vee T) \lesssim L_T$. By Lemmas \ref{lemma:empbound} and \ref{lemma:supdiff}, there is a constant $C_1$such that, w.p.a.1,
\begin{align*}
    \norm{\hat{g}(\theta_0)}_\infty &\leq  C_1\sqrt{\frac{d_0\tilde L_T}{T}},\\
    \sup_{\sigma\in\Theta_\sigma}\norm{(\E_n-\E)\tilde{z}_{jt}y_{jt}(\sigma)}_\infty &\leq  C_1\sqrt{\frac{\tilde L_T}{T}},\\
    \norm{(\E_n-\E)\tilde{z}_{jt}(p_{jt},x_{jt}')}_\infty &\leq  C_1\sqrt{\frac{\tilde L_T}{T}} .
\end{align*}
The existence of such a $C_1$ is guaranteed by Lemmas \ref{lemma:empbound} and \ref{lemma:supdiff}. Let $\lambda_{\tilde{\theta}}\geq2C_1\sqrt{d_0\tilde L_T/T}$, $\gamma_0 = (\alpha_0,\beta_0')'$ and $\tilde{\gamma} = (\tilde{\alpha},\tilde{\beta}')'$. We work on the intersection of the three events above.

\emph{Bounds on $\tilde\gamma$.} By the definition of $\tilde\theta$ in \eqref{eq:theta_tilde},
\begin{equation}\label{ineq:theta_tilde}
    \normm{\hat{g}(\tilde{\theta})}_\infty + \lambda_{\tilde{\theta}}\normm{\tilde{\beta}}_1 \leq \normm{\hat{g}(\theta_0)}_\infty + \lambda_{\tilde{\theta}}\normm{\beta_0}_1 .
\end{equation}
This bounds $\|\tilde{\beta}\|_1$ by $\|\beta_0\|_1+\|\hat{g}(\theta_0)\|_\infty/\lambda_{\tilde{\theta}}$. Consequently,
\begin{equation}\label{ineq:gamma_tilde}
    \|\tilde \beta\|_1\le \|\beta_0\|_1+\frac{\|\hat{g}(\theta_0)\|_\infty}{\lambda_{\tilde{\theta}}}\le M+\frac{C_1\sqrt{d_0\tilde L_T/T}}{\lambda_{\tilde{\theta}}}\le M+\frac{1}{2} .
\end{equation}

$\tilde \alpha$ can be bounded with a little more effort. By \eqref{ineq:theta_tilde},
\[
    \|\hat g(\tilde \theta)\|_\infty\le \|\hat g(\theta_0)\|_\infty+\lambda_{\tilde \theta}\|\beta_0\|_1\le C_1\sqrt{\frac{d_0\tilde L_T}{T}}+M\lambda_{\tilde \theta}\lesssim 1 .
\]
Moreover, by \eqref{ineq:gamma_tilde},
\begin{align*}
    \|\hat g(\tilde \theta)-g(\tilde \theta)\|_\infty &\le  \sup_{\sigma\in\Theta_\sigma}\norm{(\E_n-\E)\tilde{z}_{jt}y_{jt}(\sigma)}_\infty+\normm{(\E_n-\E)\tilde{z}_{jt}p_{jt}}_\infty|\tilde \alpha|\\
    &\quad +\max_{k,k'}\abs{\big((\E_n-\E)\tilde{z}_{jt}x_{jt}'\big)_{kk'}}\,\normm{\tilde{\beta}}_1 \\
    & \le C_1\sqrt{\frac{\tilde L_T}{T}}\,\big(1+|\tilde \alpha|+\|\tilde \beta\|_1\big)
    \lesssim \sqrt{\frac{\tilde L_T}{T}}\,\big(1+|\tilde \alpha|\big).
\end{align*}
This gives a lower bound on $\|\hat g(\tilde \theta)\|_\infty$. Since $\normm{\E\tilde z_{jt}y_{jt}(\tilde\sigma)}_\infty \lesssim 1$ by Lemma \ref{lemma:UniformBoundedness} and Condition \ref{cond:BoundednessNuisance}(i), and $\max_{k,k'}\abs{(\E\tilde z_{jt}x_{jt}')_{kk'}}\,\|\tilde\beta\|_1 \lesssim 1$,
\begin{align*}
    \|\hat g(\tilde \theta)\|_\infty&\ge \|g(\tilde \theta)\|_\infty-\|\hat g(\tilde \theta)-g(\tilde \theta)\|_\infty\\
    &\ge |\tilde{\alpha}|\,\|\E \tilde z_{jt}p_{jt}\|_\infty-\|\E \tilde z_{jt}y_{jt}(\tilde{\sigma})\|_\infty - \max_{k,k'}\abs{(\E \tilde z_{jt}x_{jt}')_{kk'}}\,\|\tilde{\beta}\|_1\\
    &\quad - C\sqrt{\frac{\tilde L_T}{T}}\,(1+|\tilde \alpha|)\\
    &\ge \Big(\|\E \tilde z_{jt}p_{jt}\|_\infty - C\sqrt{\frac{\tilde L_T}{T}}\Big)|\tilde \alpha| - C'
\end{align*}
for constants $C, C'$. To bound $\|\E \tilde z_{jt}p_{jt}\|_\infty$ from below, we use the argument of Lemma \ref{lemma:ell-infty-id}. For a unit vector $e=\E \epsilon_{jt}^{z}p_{jt}/\|\E \epsilon_{jt}^{z}p_{jt}\|_2 \in \R^{d_z}$, let $c_e := (e'(I_{d_z}, -\Pi_0))'$, so that $c_e'\tilde z_{jt} = e'(z_{jt} - \Pi_0x_{jt}) = e'(r^z_{jt} + \epsilon^z_{jt})$ and $\normm{c_e}_1 \le \sqrt{d_z}( 1 + \normm{\Pi_0}_1) \le \sqrt{d_z}( 1 + \normm{\Pi_0}_1)$. Then
\[
\|\E \tilde z_{jt}p_{jt}\|_\infty \ge \frac{\abs{c_e'\E\tilde z_{jt}p_{jt}}}{\normm{c_e}_1} \ge \frac{\abs{e'\E\epsilon^z_{jt}p_{jt}} - \abs{e'\E r^z_{jt}p_{jt}}}{\sqrt{d_z}(1 + M)} .
\]
Choosing $e$ along $\E\epsilon^z_{jt}p_{jt}$, and using $\abs{e'\E r^z_{jt}p_{jt}} \le M\sqrt{\E\normm{r^z_{jt}}_2^2} \lesssim \sqrt{d_0/T}$, we get
\[
\|\E \tilde z_{jt}p_{jt}\|_\infty \gtrsim \|\E \epsilon_{jt}^z p_{jt}\|_2 - \sqrt{d_0/T} \gtrsim \|G'e_2\|_2 - \sqrt{d_0/T} \gtrsim 1 .
\]
Collecting the last three displays,
\begin{equation}\label{ineq:alpha_tilde}
    |\tilde \alpha|\lesssim 1 + \|\hat{g}(\tilde{\theta})\|_\infty \lesssim 1 .
\end{equation}
Let $M_{\beta}$ be a constant such that $\|\tilde \gamma\|_1 \le M_\beta$ w.p.a.1, which exists by \eqref{ineq:gamma_tilde} and \eqref{ineq:alpha_tilde}.

\emph{Rate of $\tilde\sigma$.} Lemma \ref{lemma:ell-infty-id} applies to $\tilde\theta$ on this event and gives, for a constant $c > 0$,
\begin{align*}
    c|\tilde{\sigma}-\sigma_0|
    \leq &\normm{g(\tilde{\theta})}_\infty+\sqrt{\frac{d_0}{T}}\\
    \leq &\normm{g(\tilde{\theta})-\hat{g}(\tilde{\theta})}_\infty+\normm{\hat{g}(\tilde{\theta})}_\infty+\sqrt{\frac{d_0}{T}}\\
    \leq &\sup_{\sigma\in\Theta_\sigma}\norm{(\E_n-\E)\tilde{z}_{jt}y_{jt}(\sigma)}_\infty+\norm{(\E_n-\E)\tilde{z}_{jt}(p_{jt},x_{jt}')}_\infty\normm{\tilde{\gamma}}_1 + \normm{\hat{g}(\tilde{\theta})}_\infty+\sqrt{\frac{d_0}{T}}\\
    \leq &\sup_{\sigma\in\Theta_\sigma}\norm{(\E_n-\E)\tilde{z}_{jt}y_{jt}(\sigma)}_\infty+\norm{(\E_n-\E)\tilde{z}_{jt}(p_{jt},x_{jt}')}_\infty\normm{\tilde{\gamma}}_1 \\
    & +\normm{\hat{g}(\theta_0)}_\infty +\lambda_{\tilde{\theta}}\normm{\beta_0}_1 - \lambda_{\tilde{\theta}}\normm{\tilde{\beta}}_1 +\sqrt{\frac{d_0}{T}} \qquad (\text{by } \eqref{ineq:theta_tilde})\\
    \leq & C_1\sqrt{\frac{\tilde L_T}{T}}\big(1 + \normm{\tilde{\gamma}}_1\big) + C_1\sqrt{\frac{d_0\tilde L_T}{T}} + M\lambda_{\tilde{\theta}}- \lambda_{\tilde{\theta}}\normm{\tilde{\beta}}_1+\sqrt{\frac{d_0}{T}} \qquad (\text{since } \normm{\beta_0}_1 \le M)\\
    = & C_1\sqrt{\frac{\tilde L_T}{T}}\big(1 + |\tilde\alpha|\big) + \Big(C_1\sqrt{\frac{\tilde L_T}{T}} - \lambda_{\tilde{\theta}}\Big)\normm{\tilde{\beta}}_1 + C_1\sqrt{\frac{d_0\tilde L_T}{T}} + M\lambda_{\tilde{\theta}} + \sqrt{\frac{d_0}{T}}\\
    \lesssim & \sqrt{\frac{d_0\tilde L_T}{T}} + \lambda_{\tilde{\theta}} + \sqrt{\frac{d_0}{T}}
    \lesssim \sqrt{\frac{d_0\tilde L_T}{T}} .
\end{align*}
In the last inequality we used $\normm{\tilde\gamma}_1 = |\tilde\alpha| + \normm{\tilde\beta}_1$, $|\tilde\alpha| \le M_\beta$, and $\lambda_{\tilde\theta} \ge C_1\sqrt{\tilde L_T/T}$, so that the coefficient of $\normm{\tilde{\beta}}_1$ is nonpositive. Thus, since $\tilde L_T \lesssim L_T$,
\[
    |\tilde{\sigma}-\sigma_0|\lesssim_P \sqrt{\frac{d_0\tilde L_T}{T}} \lesssim \sqrt{\frac{d_0L_T}{T}}.
\]

\textbf{Step 3: Convergence rate of $\hat{f}_u$.}
\emph{Setup.} By Step 2, $|\tilde{\sigma}-\sigma_0|\lesssim \sqrt{d_0L_T/T}$ w.p.a.1. Define $r^\sigma_{jt} := y_{jt}(\tilde{\sigma})-y_{jt}(\sigma_0)$. By Lemma \ref{lemma:UniformBoundedness} and the mean value theorem,
\[
    \max_{j,t}|r^\sigma_{jt}| \leq C | \tilde{\sigma}-\sigma_0|\lesssim \sqrt{\frac{d_0L_T}{T}}\quad w.p.a.1.
\]
By \eqref{eq:highdimBLP_utility}, \eqref{eq:highdimBLP_price} and \eqref{eq:highdimBLP_approx_u},
\begin{align*}
    y_{jt}(\tilde{\sigma}) &= y_{jt}(\sigma_0) + r^\sigma_{jt}
    = \alpha_0 p_{jt} + \beta_0'x_{jt} + \xi_{jt} + \underbrace{r^\sigma_{jt} + r_{jt}^u}_{=: r_{jt}}
    = \alpha_0 f_{p0}(x_{jt},z_{jt}) + \beta_0'x_{jt} + r_{jt} + \underbrace{\xi_{jt} + \alpha_0\epsilon^p_{jt}}_{=: e_{jt}} .
\end{align*}
Let $Y, F_p, \hat F_p, R, e$ be the vectors of $y_{jt}(\tilde{\sigma})$, $f_{p0}(x_{jt},z_{jt})$, $\hat{f}_p(x_{jt},z_{jt})$, $r_{jt}$ and $e_{jt}$ over $t \in I_l^c$. Let $X$ be the matrix of all $x_{jt}'$, $\tilde{X}:=(F_p,X)$, $\hat{X}:=(\hat{F}_p,X)$ and $\gamma_0 := (\alpha_0, \beta_0')'$. Then
\[
Y = F_p\alpha_0 + X\beta_0 + R + e ,
\]
and $\hat\gamma = (\hat\alpha, \hat\beta')'$ in \eqref{eq:theta_hat} is the estimator of Lemma \ref{lem:lasso_endogeneity}. We verify the conditions of Lemma \ref{lem:lasso_endogeneity} w.p.a.1. Recall from Step 1 that $\hat F_p = \tilde X_p\hat\beta_p$ and $F_p = \tilde X_p\beta_{p0} + r^p$, so that $\hat F_p - F_p = \tilde X_p\hat\Delta_p - r^p$, and that $\tilde{z}_{jt} = (z_{jt}',x_{jt}')'$.

\emph{(a) The approximation error $\bar c_a$.} By the display above and \eqref{eq:approx_err_rates},
\[
\normm{R}_n \le \max_{j,t}\abs{r_{jt}^\sigma} + \normm{r^u}_n \lesssim_P \sqrt{\frac{d_0L_T}{T}} .
\]
Together with $c_p \lesssim \sqrt{d_0L_T/T}$ from Step 1, this gives
\[
\bar c_a = \abs{\alpha_0}c_p + \normm{R}_n \lesssim_P \sqrt{\frac{d_0L_T}{T}}.
\]

\emph{(b) The event $\mc{E}_\beta$.} We show that $\normm{\hat X'e/n}_\infty \lesssim \sqrt{L_T/T}$ w.p.a.1. Since $\hat X = (\hat F_p, X)$ and $\hat F_p = F_p + \tilde X_p\hat\Delta_p - r^p$,
\[
\norm{\frac{\hat X'e}{n}}_\infty \le \norm{\frac{X'e}{n}}_\infty + \abs{\frac{F_p'e}{n}} + \abs{\hat\Delta_p'\frac{\tilde X_p'e}{n}} + \abs{\frac{r^{p\prime}e}{n}} .
\]
By Condition \ref{cond:BoundednessNuisance}(i), $e_{jt}$, $f_{p0}(x_{jt},z_{jt})$ and $\normm{(z_{jt}',x_{jt}')'}_\infty$ are bounded by a constant. We bound the four terms in turn.
\begin{itemize}
    \item Lemma \ref{lemma:Hoeffding} gives $\normm{X'e/n}_\infty \le \normm{\tilde X_p'e/n}_\infty \lesssim_P \sqrt{L_T/T}$.
    \item $\E[f_{p0}(x_{jt},z_{jt})\xi_{jt}] = 0$, and $\E[f_{p0}(x_{jt},z_{jt})\epsilon_{jt}^p] = 0$ by \eqref{eq:highdimBLP_price}. Since the summands are bounded and independent across markets, Chebyshev's inequality gives $F_p'e/n = O_p(T^{-1/2}) = o_p(\sqrt{L_T/T})$.
    \item By H\"older's inequality, the first item and Step 1, $\abs{\hat\Delta_p'\tilde X_p'e/n} \le \normm{\hat\Delta_p}_1\normm{\tilde X_p'e/n}_\infty \lesssim d_0L_T/T = o(\sqrt{L_T/T})$ w.p.a.1, because $d_0\sqrt{L_T/T} \to 0$. 
    \item $r_{jt}^p$ is a function of $(x_{jt},z_{jt})$, so $\E[r_{jt}^p\xi_{jt}] = 0$, and $\E[r_{jt}^p\epsilon_{jt}^p] = \E[f_{p0}(x_{jt},z_{jt})\epsilon_{jt}^p] - \beta_{p0}'\E[(z_{jt}',x_{jt}')'\epsilon_{jt}^p] = 0$ by \eqref{eq:highdimBLP_price}. Independence across markets and boundedness of $e_{jt}$ give $\Var(r^{p\prime}e/n) \lesssim T^{-1}\E(r_{jt}^p)^2 \lesssim d_0/T^2$. Hence $r^{p\prime}e/n = O_p(\sqrt{d_0}/T) = o_p(\sqrt{L_T/T})$.
\end{itemize}
Therefore there exists a constant $C_e$ such that $\normm{\hat X'e/n}_\infty \le C_e\sqrt{L_T/T}$ w.p.a.1. Choosing $\lambda_\beta = C_\beta\sqrt{L_T/T}$ with $C_\beta$ large enough, the event $\mc{E}_\beta$ holds w.p.a.1.

\emph{(c) Conditions (a)--(c) of Lemma \ref{lem:lasso_endogeneity}.} Recall that $\frac1n\tilde X'\tilde X = \hat\Sigma_l^{(1)}$, and that we work on the event of Lemma \ref{lem:SE_RE2}, with $m = (\ell_T - d_z - 1)d_0$, $\kappa_1 = \underline\phi/9$ and $\nu_0 \in (1,\nu_1)$.
\begin{itemize}
    \item Condition (a) holds by Lemma \ref{lem:SE_RE2}(ii).
    \item Lemma \ref{lem:SE_RE2}(i) and the monotonicity of sparse eigenvalues give $\phi_{\min}(m + d_0)[\hat\Sigma_l^{(1)}] \ge \phi_{\min}(m + (d_z+1)d_0)[\hat\Sigma_l^{(1)}] \ge \underline\phi$ and $\phi_{\max}(m)[\hat\Sigma_l^{(1)}] \le \bar\phi$. By (a) of this step and $\lambda_\beta = C_\beta\sqrt{L_T/T}$, $\bar c_a/\lambda_\beta \lesssim \sqrt{d_0}$ w.p.a.1, so $B_\beta \le C_Bd_0$ w.p.a.1 for a constant $C_B$. Since $\ell_T \to \infty$, $m > \frac92\bar\phi C_Bd_0 \ge \frac92\bar\phi B_\beta$ for all $T$ large enough. Hence condition (b) holds.
    \item By Step 1, $c_p \lesssim \sqrt{d_0L_T/T} \to 0$, so $c_p \le \frac12\min(\sqrt{\kappa_1},\sqrt{\underline\phi}) = \frac{\sqrt{\underline\phi}}{6}$ w.p.a.1. Hence condition (c) holds.
\end{itemize}

\emph{(d) Conclusion.} On an event with probability approaching one, all conditions of Lemma \ref{lem:lasso_endogeneity} and the event $\mc{E}_\beta$ hold. The lemma and (a) of this step give
\[
\normm{\hat\gamma - \gamma_0}_2 \lesssim_P \bar c_a + \sqrt{d_0}\lambda_\beta \lesssim_P \sqrt{\frac{d_0L_T}{T}}.
\]
Since $\hat\beta$ is computed from $I_l^c$ and $x_{jt}$ with $t \in I_l$ is independent of it, Condition \ref{cond:Regularity2}(ii) gives, conditional on the estimation sample,
\begin{align*}
    \normm{\hat{f}_u(x_{jt})-f_{u0}(x_{jt})}_{L^2} = & \normm{\hat{\beta}'x_{jt} - \beta_0'x_{jt} -r_{jt}^u}_{L^2}\\
    \leq & \normm{(\hat{\beta}-\beta_0)'x_{jt}}_{L^2} + \sqrt{\E(r_{jt}^u)^2}\\
    \leq & \sqrt{\lambda_{\max}(\E x_{jt}x_{jt}')}\,\normm{\hat{\gamma}-\gamma_0}_2 + c_u\\
    \lesssim_P & \sqrt{\frac{d_0L_T}{T}} .
\end{align*}
Hence $\normm{\hat{f}_u(x_{jt})-f_{u0}(x_{jt})}_{L^2} = o_p(T^{-1/4})$ under Condition \ref{cond:ApproxSparsity}(ii). This completes the proof.

\end{proof}

\begin{proof}[Proof of Corollary \ref{cor:inference}]
Throughout, $\Delta_l f_z(x_{jt}):=\hat f_z^l(x_{jt})-f_{z0}(x_{jt})$, $\Delta_l f_u(x_{jt}):=\hat f_u^l(x_{jt})-f_{u0}(x_{jt})$, and, as in the proof of Theorem~\ref{thm:AsympNorm}, $\Delta_{jt}(\theta_1):=y_{jt}(\sigma)-y_{jt}(\sigma_0)-(\alpha-\alpha_0)p_{jt}$. By Theorem~\ref{thm:Nuisance}, the rates in Condition~\ref{cond:PrelimConvergence} hold, so equations (\ref{eq:Deltafz})--(\ref{eq:Deltafbetaxi}) in the proof of Theorem~\ref{thm:Consistency} and \ref{thm:AsympNorm} are available.

\textbf{Step 1: $\check\theta_{1,I}$.} With $\hat W=W=I_{d_z}$, Condition \ref{cond:WeightMatrix} holds trivially. Theorems 3.1 and 3.2 then give $\check\theta_{1,I}\xrightarrow{p}\theta_{10}$ and $P_I^{-1/2}\sqrt T(\check\theta_{1,I}-\theta_{10})\xrightarrow{d}N(0,I_{d_{\theta_1}})$. In particular $\|\check\theta_{1,I}-\theta_{10}\|_2=O_p(T^{-1/2})$.

\textbf{Step 2: $\hat G(\check\sigma)\xrightarrow{p}G$ for any $\check\theta_1\xrightarrow{p}\theta_{10}$.} By definition \eqref{eq:G_hat},
\begin{equation*}
\hat G(\check\sigma)=\frac{\partial}{\partial\theta_1}\E_{T,L}\big[\psi_t(\check\theta_1,\hat f_z^l,\hat f_u^l)\big]',
\end{equation*}
which is exactly the matrix $A$ in the proof of Theorem~\ref{thm:AsympNorm}. Step 1(i) of that proof ($A=A_1+A_2+A_3$ with $A_1=G+o_p(1)$, $A_2=o_p(1)$ by Lemma~\ref{lemma:UniformBoundedness}, and $A_3=o_p(1)$ by (\ref{eq:Deltafz})) uses only $\check\theta_1\xrightarrow{p}\theta_{10}$. Hence $\hat G_I-G=o_p(1)$, and $\hat G_{\mathrm{opt}}-G=o_p(1)$ once Step 1 has established $\check\theta_{1,I}\xrightarrow{p}\theta_{10}$ and $\check\theta_{1,\mathrm{opt}}\xrightarrow{p}\theta_{10}$.

\textbf{Step 3: $\hat\Omega\xrightarrow{p}\Omega$.} Since $y_{jt}(\check\sigma_I)-\check\alpha_Ip_{jt}-\hat f_u^l(x_{jt})=\xi_{jt}+\Delta_{jt}(\check\theta_{1,I})-\Delta_l f_u(x_{jt})$ and $z_{jt}-\hat f_z^l(x_{jt})=\epsilon^z_{jt}-\Delta_l f_z(x_{jt})$, we can write
\begin{equation*}
\psi_t(\check\theta_{1,I},\hat f_z^l,\hat f_u^l)=\psi_t^0+R_t^l,\qquad \psi_t^0:=\E_J\epsilon^z_{jt}\xi_{jt},
\end{equation*}
\begin{equation*}
R_t^l:=\E_J\Big[\epsilon^z_{jt}\big(\Delta_{jt}(\check\theta_{1,I})-\Delta_l f_u(x_{jt})\big)-\Delta_l f_z(x_{jt})\big(\xi_{jt}+\Delta_{jt}(\check\theta_{1,I})-\Delta_l f_u(x_{jt})\big)\Big].
\end{equation*}
Hence
\begin{equation*}
\hat\Omega-\Omega=\underbrace{\E_T[\psi_t^0\psi_t^{0\prime}]-\Omega}_{D_1}
+\underbrace{\E_{T,L}[\psi_t^0R_t^{l\prime}]+\E_{T,L}[R_t^l\psi_t^{0\prime}]}_{D_2}
+\underbrace{\E_{T,L}[R_t^lR_t^{l\prime}]}_{D_3}.
\end{equation*}

\emph{Term $D_1$.} $\E[\psi_t^0]=0$, so $\E[\psi_t^0\psi_t^{0\prime}]=\Omega$. Because $\|\psi_t^0\|_2\le C$ by Condition \ref{cond:BoundednessVar2}, Chebyshev's inequality for the i.i.d.\ triangular array gives $D_1=O_p(T^{-1/2})$.

\emph{Terms $D_2$ and $D_3$.} By the Cauchy--Schwarz inequality and $\|\psi_t^0\|_2\le C$, we have $\|D_2\|_2\le C\big(\E_{T,L}\|R_t^l\|_2^2\big)^{1/2}$ and $\|D_3\|_2\le\E_{T,L}\|R_t^l\|_2^2$, so it suffices to show $\E_{T,L}\|R_t^l\|_2^2=o_p(1)$. By Jensen's inequality over $j$ and the boundedness of $\epsilon^z_{jt}$ and $\xi_{jt}$,
\begin{align*}
\E_{T,L}\|R_t^l\|_2^2\lesssim\ &\E_J\E_{T,L}\big[\Delta_{jt}(\check\theta_{1,I})^2\big]
+\E_J\E_{T,L}|\Delta_lf_u(x_{jt})|^2
+\E_J\E_{T,L}\|\Delta_lf_z(x_{jt})\|_2^2\\
&+\E_J\E_{T,L}\Big[\|\Delta_lf_z(x_{jt})\|_2^2\cdot\big(\Delta_{jt}(\check\theta_{1,I})^2+|\Delta_lf_u(x_{jt})|^2\big)\Big].
\end{align*}
We treat each term in turn.
\begin{enumerate}[label=(\alph*)]
\item By the mean value theorem, Lemma \ref{lemma:UniformBoundedness} and $|p_{jt}|\le M$, $\sup_{j,t}|\Delta_{jt}(\check\theta_{1,I})|\le C\|\check\theta_{1,I}-\theta_{10}\|_2=O_p(T^{-1/2})$ by Step 1.

\item Apply Lemma~\ref{lemma:ConditionalPopulation} as in (\ref{eq:Deltafz})--(\ref{eq:Deltafbeta}), with $a(w_t,\hat\eta_l)=b(w_t,\hat\eta_l)=\|\Delta_lf_z(x_{jt})\|_2\cdot T^{1/4}$ (respectively $|\Delta_lf_u(x_{jt})|\cdot T^{1/4}$). This gives $\E_{T,L}\|\Delta_lf_z(x_{jt})\|_2^2=o_p(T^{-1/2})$ and $\E_{T,L}|\Delta_lf_u(x_{jt})|^2=o_p(T^{-1/2})$.

\item For the product term $\E_J\E_{T,L}\|\Delta_lf_z(x_{jt})\|_2^2\cdot\Delta_{jt}(\check\theta_{1,I})^2$. By (a), there is $\sup_{j,t}|\Delta_{jt}(\check\theta_{1,I})|\le C\|\check\theta_{1,I}-\theta_{10}\|_2=O_p(T^{-1/2})$. This gives the convergence.

\item For the product term $\E_J\E_{T,L}\|\Delta_lf_z(x_{jt})\|_2^2\cdot|\Delta_lf_u(x_{jt})|^2$, it suffices to show 
$$\max_{l}\sup_{j,t\in I_l}\|\Delta_lf_z(x_{jt})\|_2=o_p(1).$$ 
First decompose $\|\Delta_lf_z(x_{jt})\|_2\le \|(\hat{\Pi}^l-\Pi_0)x_{jt}\|_2+\|r_{jt}^z\|_2\le \|(\hat{\Pi}^l-\Pi_0)x_{jt}\|_2+\sqrt{d_0/T}$.
It boils down to bounding $\sup_{j,t\in I_l}\|(\hat{\Pi}^l-\Pi_0)x_{jt}\|_2$.
By Condition~\ref{cond:BoundednessVar2} and Lemma~\ref{lem:lasso_approx}(ii),
\begin{align*}
    \sup_{j,t\in I_l}\|(\hat{\Pi}^l-\Pi_0)x_{jt}\|_2&\lesssim \|\hat{\Pi}^l-\Pi_0\|_1=\sum_{i=1}^{d_z}\|\hat{\Pi}_i^l-\Pi_{0i}\|_1\lesssim_P d_0\sqrt{\log (d_x\vee T)/T} = o_p(1).
\end{align*}
Hence the uniform convergence is established.
\end{enumerate}
Therefore $\E_{T,L}\|R_t^l\|_2^2=o_p(1)$ and $\hat\Omega-\Omega=o_p(1)$.

\textbf{Step 4: $\hat P_I\xrightarrow{p}P_I$ and the first display.} By Condition \ref{cond:IdentificationG}, $\lambda_{\min}(G_*G_*')\ge c$ and $\lambda_{\min}(\Omega_*)\ge c$. The map $h_I(G,\Omega)=(GG')^{-1}G\Omega G'(GG')^{-1}$ is continuous at $(G_*,\Omega_*)$. Since $\hat G_I\xrightarrow{p}G$ and $\hat\Omega\xrightarrow{p}\Omega$ by Steps 2--3 and the assumption $G\to G_*$, $\Omega\to\Omega_*$, the continuous mapping theorem gives $\hat P_I\xrightarrow{p}P_{I,*}:=h_I(G_*,\Omega_*)$. Because also $P_I\to P_{I,*}$, we have $\hat P_I-P_I=o_p(1)$. Observing that $G_*$ is a bounded matrix, there is
\begin{align*}
    \lambda_{\min}(P_{I,*})\ge \lambda_{\min}(\Omega_*)\lambda_{\min}((G_*G_*')^{-1})\ge  c/\|G_*G_*'\|_2>0.
\end{align*}
So $\hat P_I^{-1/2}-P_I^{-1/2}=o_p(1)$. Hence
\begin{equation*}
\hat P_I^{-1/2}\sqrt T(\check\theta_{1,I}-\theta_{10})
=P_I^{-1/2}\sqrt T(\check\theta_{1,I}-\theta_{10})+o_p(1)\cdot O_p(1)\xrightarrow{d}N(0,I_{d_{\theta_1}})
\end{equation*}
by Step 1 and Slutsky's lemma.

\textbf{Step 5: $\check\theta_{1,\mathrm{opt}}$.} The argument is similar.
\end{proof}

\section{Appendix: Contraction Mapping}

In the theory, we need to back out $y_t(\sigma)$ from $\sigma$ and $s_t$.
To achieve this in practice, we define, on the cube
$\mathcal{B}_M:=\{x\in\mathbb{R}^{J}:\norm{x-y_t}_{\infty}\leq M\}$, the mapping
\[
    \mathcal{T}: \mathcal{B}_M\to \mathcal{B}_M,\quad \mathcal{T}(x)=x+\log(s_t)-\log(f_s(p_t,x,\sigma))
\]
For abbreviation, we denote $f_s(p_t,x,\sigma)$ as $f_s(x,\sigma)$.
And let $f_{s,i}(x,\sigma)$ be the $i$-th element of $f_s(x,\sigma)$, and $\partial f_{s,i}(x,\sigma)/\partial x_j$ be the partial derivative of $f_{s,i}(x,\sigma)$ w.r.t. $x_j$, the $j$-th component of $x$.
Recall that the taste shifter $b$ has distribution $F(\cdot,\sigma)$, so that
\[
    f_{s,i}(x,\sigma)=\int \pi_i(x,b)\,dF(b,\sigma),\qquad
    \pi_i(x,b)=\frac{\exp\{x_{i}+p_{i}b\}}{1+\sum_{k=1}^{J}\exp\{x_{k}+p_{k}b\}},
\]
and we write $\pi_0(x,b)=\left(1+\sum_{k=1}^{J}\exp\{x_{k}+p_{k}b\}\right)^{-1}$ for the outside choice probability, so that $\pi_0(x,b)+\sum_{k=1}^{J}\pi_k(x,b)=1$.
The arguments below use only this representation and therefore cover any parametric family $F(\cdot,\sigma)$, such as the normal and lognormal specifications of Section 2.
\begin{theorem}\label{thm:contraction}
    $\mathcal{T}$ is a contraction mapping, and thus the iteration admits a fixed point.
\end{theorem}
\begin{proof}
    There are two things we need to show: first mapping $\mathcal{T}$ is well-defined; second it is a contraction mapping. It turns out the two statements can be shown at the same time.\\
    Define the metric as $d(x,w)=\norm{x-w}_{\infty}$, where $x,w\in \mathcal{B}_M$.
    Since $\mathcal{B}_M$ is convex, applying the mean value theorem to each coordinate gives
    \begin{align*}
        [\mathcal{T}(x)-\mathcal{T}(w)]_{i}&=[(x-\log(f_s(x,\sigma)))-(w-\log(f_s(w,\sigma)))]_{i}\\
        &=\left[P(z^{(i)})(x-w)\right]_{i}
    \end{align*}
    where $z^{(i)}\in\mathcal{B}_M$ lies on the line between $x$ and $w$. $P(z)=[\delta_{ij}-f_{s,i}(z,\sigma)^{-1}\partial f_{s,i}(z,\sigma)/\partial z_j]_{i,j}$, where $\delta_{ij}=1 $ iff $i=j$.
    Note that when $i=j$, we have
    \begin{align*}
    \frac{\partial f_{s,i}(z,\sigma)}{\partial z_i}&=\int \pi_i(z,b)\left(1-\pi_i(z,b)\right)dF(b,\sigma)>0\\
    \end{align*}
    When $i\neq j$, we have
    \begin{align*}
        \frac{\partial f_{s,i}(z,\sigma)}{\partial z_j}&=\int \pi_i(z,b)\left(-\pi_j(z,b)\right)dF(b,\sigma)<0
    \end{align*}
    Moreover, since $0<\pi_i(z,b)<1$,
    \[
    0<f_{s,i}(z,\sigma)^{-1}\frac{\partial f_{s,i}(z,\sigma)}{\partial z_i}
    =1-\frac{\int \pi_i(z,b)^2dF(b,\sigma)}{\int \pi_i(z,b)dF(b,\sigma)}<1,
    \]
    so the $i$-th diagonal entry of $P(z)$ is itself nonnegative, while the off-diagonal entries are nonnegative by the display above.
    Thus
    \begin{align*}
        &\left|[\mathcal{T}(x)-\mathcal{T}(w)]_{i}\right|\\
        &=\left|\sum_{j=1}^{J}(\delta_{ij}-f_{s,i}(z^{(i)},\sigma)^{-1}\partial f_{s,i}(z^{(i)},\sigma)/\partial z_j)(x_j-w_j)\right|\\
        &\leq \sum_{j=1}^{J}|\delta_{ij}-f_{s,i}(z^{(i)},\sigma)^{-1}\partial f_{s,i}(z^{(i)},\sigma)/\partial z_j|\cdot|x_j-w_j|\\
        &\leq \sum_{j=1}^{J}|\delta_{ij}-f_{s,i}(z^{(i)},\sigma)^{-1}\partial f_{s,i}(z^{(i)},\sigma)/\partial z_j|\norm{x-w}_\infty\\
        &=\left(1-f_{s,i}(z^{(i)},\sigma)^{-1}\sum_{j=1}^{J}\partial f_{s,i}(z^{(i)},\sigma)/\partial z_j\right)\norm{x-w}_\infty\\
        &=\left(1-\underbrace{\left(\int \pi_i(z^{(i)},b)\,dF(b,\sigma)\right)^{-1}\int \pi_i(z^{(i)},b)\,\pi_0(z^{(i)},b)\,dF(b,\sigma)}_{g_i(z^{(i)},\sigma)}\right)\norm{x-w}_\infty\\
    \end{align*}
    where the last equality uses $\sum_{j=1}^{J}\partial f_{s,i}(z,\sigma)/\partial z_j=\int \pi_i(z,b)\left(1-\sum_{k=1}^{J}\pi_k(z,b)\right)dF(b,\sigma)=\int \pi_i(z,b)\pi_0(z,b)\,dF(b,\sigma)$.
    Obviously, $0<g_i(z,\sigma)<1$ because $0<\pi_0(z,b)<1$, and $g_i$ is a continuous function of $z$ on the compact set $\mathcal{B}_M$ by dominated convergence, the integrands being bounded by one. Since there are finitely many coordinates $i$, $\min_{1\leq i\leq J}\min_{z\in\mathcal{B}_M}g_i(z,\sigma)$ is attained and is a positive number.
    Thus there is a number $0<C(M,p,\sigma)<1$ s.t. $\forall x,w\in \mathcal{B}_M$ and every $i$
    \[
    |[\mathcal{T}(x)-\mathcal{T}(w)]_{i}|\leq C(M,p,\sigma)\norm{x-w}_\infty
    \]
    Consequently,
    \[
    \norm{\mathcal{T}(x)-\mathcal{T}(w)}_\infty\leq C(M,p,\sigma)\norm{x-w}_\infty
    \]
    Specifically, for any $\sigma$, let $y_t$ be the solution to the equation $f_s(y_t,\sigma)=s_t$, then $\mathcal{T}(y_t)=y_t$, i.e., $y_t$ is the fixed point of $\mathcal{T}$.
    Consequently $\forall x\in \mathcal{B}_M$
    \[
    \norm{\mathcal{T}(x)-y_t}_\infty=\norm{\mathcal{T}(x)-\mathcal{T}(y_t)}_\infty\leq C(M,p,\sigma)\norm{x-y_t}_\infty< \norm{x-y_t}_\infty \leq  M
    \]
    i.e. $\mathcal{T}(x)\in \mathcal{B}_M$. So $\mathcal{T}$ is well-defined.
\end{proof}



\noindent\textbf{Declaration of generative AI and AI-assisted technologies in the manuscript preparation process}

\noindent During the preparation of this work, the author used Claude and ChatGPT for structuring the paper's writing framework and supporting code writing and debugging.
The author reviewed and edited the output as needed and takes full responsibility for the content of the published article. 

\bibliographystyle{apalike}
\bibliography{bib.bib}

\end{document}